\documentclass[11pt, sort]{article}
\usepackage{hyperref, bm}
\usepackage{tabularx}
\usepackage{comment}
\usepackage{soul}
\usepackage{geometry}
\newcommand{\hili}[1]{\colorbox{yellow}{#1}}
 \usepackage{mathtools}
\newcommand{\x}{{\bf{x}}}
\AtBeginDocument{
\addtolength{\abovedisplayskip}{-0.9ex}
\addtolength{\abovedisplayshortskip}{-0.9ex}
\addtolength{\belowdisplayskip}{-0.9ex}
\addtolength{\belowdisplayshortskip}{-0.9ex}
}
\makeatletter
\def\thm@space@setup{\thm@preskip=3pt
\thm@postskip=3pt}
\makeatother
\usepackage{titlesec}
\titlespacing*{\section}
  {0pt}{0.7\baselineskip}{0.2\baselineskip}
\titlespacing*{\subsection}
  {0pt}{0.7\baselineskip}{0.2\baselineskip}

\usepackage{graphicx}
\usepackage{amsmath, amsthm, amssymb, amsfonts, latexsym}
\usepackage{array}
\usepackage{caption}
\usepackage{color}
\newcommand{\hilight}[1]{\colorbox{yellow}{#1}}
\renewcommand{\captionfont}{\small\itshape}
\renewcommand{\captionlabelfont}{\small\scshape}
\usepackage{subfigure}
\renewcommand{\subcapsize}{\scriptsize}
\renewcommand{\subcaplabelfont}{\normalfont}
\renewcommand{\subcapfont}{\itshape}
\renewcommand{\subfigcapmargin}{30pt}
\renewcommand{\subfigcapskip}{0pt}
\renewcommand{\subfigbottomskip}{0pt}

\usepackage{epsfig}
\usepackage{fancybox}
\usepackage{array}

\usepackage{mathrsfs}
\usepackage{longtable,}
\usepackage{calc}
\usepackage{multirow}
\usepackage{hhline}
\usepackage{ifthen}
\usepackage{algorithmic}
\usepackage{algorithm}
\usepackage{color}

\usepackage[running, mathlines]{lineno}

\DeclareMathOperator{\sgn}{sgn}

\newcommand*\patchAmsMathEnvironmentForLineno[1]{%
  \expandafter\let\csname old#1\expandafter\endcsname\csname #1\endcsname
  \expandafter\let\csname oldend#1\expandafter\endcsname\csname end#1\endcsname
  \renewenvironment{#1}%
     {\linenomath\csname old#1\endcsname}%
     {\csname oldend#1\endcsname\endlinenomath}}%
\newcommand*\patchBothAmsMathEnvironmentsForLineno[1]{%
  \patchAmsMathEnvironmentForLineno{#1}%
  \patchAmsMathEnvironmentForLineno{#1*}}%
\AtBeginDocument{%
\patchBothAmsMathEnvironmentsForLineno{equation}%
\patchBothAmsMathEnvironmentsForLineno{align}%
\patchBothAmsMathEnvironmentsForLineno{flalign}%
\patchBothAmsMathEnvironmentsForLineno{alignat}%
\patchBothAmsMathEnvironmentsForLineno{gather}%
\patchBothAmsMathEnvironmentsForLineno{multline}%
}

\numberwithin{equation}{section}
\numberwithin{table}{section}
\numberwithin{figure}{section}

\newtheorem{definition}{Definition}
\newtheorem{theorem}{Theorem}
\newtheorem{lemma}{Lemma}
\newtheorem{remark}{Remark}

\newtheorem{proposition}{Proposition}
\newtheorem{corollary}{Corollary}

\numberwithin{definition}{section}
\numberwithin{theorem}{section}
\numberwithin{lemma}{section}
\numberwithin{remark}{section}
\numberwithin{assumption}{section}
\numberwithin{condition}{section}
\numberwithin{property}{section}
\numberwithin{proposition}{section}
\numberwithin{corollary}{section}
\numberwithin{algorithm}{section}
\usepackage{empheq}
\pdfoutput=1
\newcommand{\e}{\mathrm{e}}
\newcommand{\blue}{\color{black}}
\newcommand{\myblue}{\color{black}}

\newcommand{\g}{\color{black}}
\newcommand{\red}{\color{black}}
\newcommand{\EQ}{\begin{equation}}
\newcommand{\EN}{\end{equation}}
\newcommand{\EQS}{\begin{equation*}}
\newcommand{\ENS}{\end{equation*}}
\newcommand{\EQA}{\begin{eqnarray}}
\newcommand{\ENA}{\end{eqnarray}}
\newcommand{\EQAS}{\begin{eqnarray*}}
\newcommand{\ENAS}{\end{eqnarray*}}
\newcommand{\AL}{\begin{align}}
\newcommand{\AN}{\end{align}}
\newcommand{\ALS}{\begin{align*}}
\newcommand{\ANS}{\end{align*}}

\newcommand{\md}{\mathrm{d}}
\newcommand{\Ebb}{\mathbb{E}}
\newcommand{\Ibb}{\mathbb{I}}
\newcommand{\Qbb}{\mathbb{Q}}

\newcommand{\tn}{t_m}
\newcommand{\tnn}{t_{m+1}}
\newcommand{\tnp}{t_{m-1}}

\newcommand{\zn}{z_n}
\newcommand{\znn}{z_{n+1}}
\newcommand{\znp}{z_{n-1}}
\newcommand{\zp}{z^{\prime}}
\newcommand{\Var}{\mathcal{V}}
\newcommand{\var}{\nu}
\newcommand{\varn}{\nu_n}
\newcommand{\varnn}{\nu_{n+1}}
\newcommand{\varnp}{\nu_{n-1}}
\newcommand{\Lnv}{\mathcal{Q}}
\newcommand{\lnv}{\varsigma}
\newcommand{\lnvn}{\varsigma_n}
\newcommand{\lnvnn}{\varsigma_{n+1}}
\newcommand{\lnvnp}{\varsigma_{n-1}}
\newcommand{\lnvp}{\varsigma^{\prime}}

\newcommand{\ds}{\displaystyle}
\def\r{\right}
\def\l{\left}
\def\f{\frac}
\def\b{\textbf}
\def\s{\sqrt}
\def\i{\infty}
\def\t{\text}

\newcommand{\Acal}{\mathcal{A}}
\newcommand{\Bcal}{\mathcal{B}}
\newcommand{\Ccal}{\mathcal{C}}
\newcommand{\Fcal}{\mathcal{F}}
\newcommand{\Gcal}{\mathcal{G}}
\newcommand{\Hcal}{\mathcal{H}}
\newcommand{\Jcal}{\mathcal{J}}
\newcommand{\Lcal}{\mathcal{L}}
\newcommand{\Mcal}{\mathcal{M}}
\newcommand{\Ncal}{\mathcal{N}}
\newcommand{\Ocal}{\mathcal{O}}
\newcommand{\Pcal}{\mathcal{P}}
\newcommand{\Qcal}{\mathcal{Q}}
\newcommand{\Rcal}{\mathcal{R}}
\newcommand{\Scal}{\mathcal{S}}
\newcommand{\Tcal}{\mathcal{T}}
\newcommand{\Zcal}{\mathcal{Z}}

\newcommand{\sinc}{\text{sinc}}
\newcommand{\Rit}{\mathit{R}}
\newcommand{\winf}{s_{\scalebox{0.6}{-$\infty$}}}
\newcommand{\wpinf}{s_{\scalebox{0.6}{$\infty$}}}
\newcommand{\epinf}{e^{\wpinf}}
\newcommand{\einf}{e^{\winf}}

\def\r{\right}
\def\l{\left}
\newcommand{\Oinf}{\Omega^{\scalebox{0.5}{$\infty$}}}
\newcommand{\myin}{\scalebox{0.7}{\text{in}}}

\makeatletter
\newcommand{\subalign}[1]{%
  \vcenter{%
    \Let@ \restore@math@cr \default@tag
    \baselineskip\fontdimen10 \scriptfont\tw@
    \advance\baselineskip\fontdimen12 \scriptfont\tw@
    \lineskip\thr@@\fontdimen8 \scriptfont\thr@@
    \lineskiplimit\lineskip
    \ialign{\hfil$\m@th\scriptstyle##$&$\m@th\scriptstyle{}##$\crcr
      #1\crcr
    }%
  }
}
\makeatother

\AtBeginDocument{
\addtolength{\abovedisplayskip}{-0.2ex}
\addtolength{\abovedisplayshortskip}{-0.2ex}
\addtolength{\belowdisplayskip}{-0.2ex}
\addtolength{\belowdisplayshortskip}{-0.2ex}
}

\usepackage{titlesec}
\titlespacing*{\section}
  {0pt}{0.5\baselineskip}{0.2\baselineskip}
\titlespacing*{\subsection}
  {0pt}{0.5\baselineskip}{0.2\baselineskip}
\titlespacing*{\subsubsection}
  {0pt}{0.5\baselineskip}{0.2\baselineskip}

\newcommand{\errorf}{\scalebox{0.9}{$\mathcal{E}_{f}$}}
\newcommand{\errorfprime}{\scalebox{0.9}{$\mathcal{E}_{f}$}}
\newcommand{\errorp}{\scalebox{0.9}{$\mathcal{E}_{o}$}}
\newcommand{\errorpprime}{\scalebox{0.9}{$\mathcal{E}_{o}$}}
\newcommand{\errorg}{\scalebox{0.9}{$\mathcal{E}_{\hat{g}}$}}
\newcommand{\errorgprime}{\scalebox{0.9}{$\mathcal{E}_{\hat{g}}$}}
\newcommand{\errorb}{\scalebox{0.9}{$\mathcal{E}_{b}$}}
\newcommand{\errorbprime}{\scalebox{0.9}{$\mathcal{E}_{b}$}}
\newcommand{\errorc}{\scalebox{0.9}{$\mathcal{E}_{c}$}}
\newcommand{\errorcprime}{\scalebox{0.9}{$\mathcal{E}_{c}$}}

\newcommand{\errorq}{\scalebox{0.9}{$\mathcal{E}_{q}$}}
\newcommand{\errorqprime}{\scalebox{0.9}{$\mathcal{E}_{q}\varsigma$}}

\def\allfiles{}
\begin{document}
\title{A monotone numerical integration method for mean-variance portfolio optimization
under jump-diffusion models}
\author{
Hanwen Zhang
         \thanks{School of Mathematics and Physics, The University of Queensland, St Lucia, Brisbane 4072, Australia,
email: \texttt{hanwen.zhang1@uqconnect.edu.au}.}
\and
Duy-Minh Dang\thanks{School of Mathematics and Physics, The University of Queensland, St Lucia, Brisbane 4072, Australia,
email: \texttt{duyminh.dang@uq.edu.au}}
}
\date{May 15, 2023}
\maketitle
\begin{abstract}
We develop a highly efficient, easy-to-implement, and strictly monotone numerical integration method for Mean-Variance (MV) portfolio optimization.
This method proves very efficient in realistic contexts, which involve
jump-diffusion dynamics of the underlying controlled processes, discrete rebalancing, and the application of investment constraints, namely no-bankruptcy and leverage.

A crucial element of the MV portfolio optimization formulation over each rebalancing interval is a convolution integral, which involves a conditional density of the logarithm of the amount invested in the risky asset. Using a known closed-form expression for the Fourier transform of this conditional density, we derive an infinite series representation for the conditional density where each term is strictly positive and explicitly computable. As a result, the convolution integral can be readily approximated through a monotone integration scheme, such as a composite quadrature rule typically available in most programming languages.
The computational complexity of our proposed method is an order of magnitude lower than that of existing monotone finite difference methods. To further enhance efficiency, we propose an implementation of this monotone integration scheme via Fast Fourier Transforms, exploiting the Toeplitz matrix structure.

The proposed monotone numerical integration scheme is proven to be both $\ell_{\infty}$-stable and pointwise consistent, and we rigorously establish its pointwise convergence to the unique solution of the MV portfolio optimization problem. We also intuitively demonstrate that,
as the rebalancing time interval approaches zero, the proposed scheme converges to a continuously observed impulse control formulation for MV optimization expressed as a Hamilton-Jacobi-Bellman Quasi-Variational Inequality. Numerical results show remarkable agreement with benchmark solutions obtained through finite differences and Monte Carlo simulation, underscoring the effectiveness of our approach.

\vspace{9pt}
\end{abstract}

\section{Introduction}

Long-term investors, such as holders of Defined Contribution plans,
are typically motivated by asset allocation strategies which are optimal under multi-period criteria.\footnote{The holder of a Defined Contribution  plan is effectively responsible to make investment decisions for both (i) the accumulation phase
(pre-retirement) of about thirty years or more, and (ii) the decumulation phase (in retirement), of perhaps twenty years.}
As a result, multi-period portfolio optimisation plays a central role in asset allocation.
In particular, originating with \cite{Markowitz1952}, mean-variance (MV) portfolio optimization forms the cornerstone of asset allocation (\cite{EltonGruberEtAlBOOK}), in part due to its intuitive nature which is the trade-off between risk (variance) and reward (mean). In multi-period  settings, MV portfolio optimization aims to obtain an investment strategy (or control) that maximizes the expected value of the terminal wealth of the portfolio, for a given level of risk as measured by the associated variance of the terminal wealth \cite{ZhouLi2000}.
In recent years, 
multi-period MV optimization has received considerable attention in in institutional settings, including
in  pension fund and insurance applications - see for example \cite{HojgaardVigna2007,MenoncinVigna2013,LiangBaiGuo2014,Nkeki2014,SunLiZeng2016,Vigna_efficiency2014,WangChen2018,WangChen2019,WuZeng2015, vigna2022tail, ForsythVetzalWestmacott2019,ForsythVetzal2019,ChenLiGuo2013,LinQian2016,WeiWang2017,ZhaoShenZeng2016,ZhouXiaoYinZengLin2016,
LiForsyth2019}, among many others.

It is important to distinguish between two categories of optimal investment strategies (optimal controls) for portfolio optimization.
The first category, referred to as pre-commitment, typically results in time-inconsistent optimal strategies (\cite{ZhouLi2000, LiNg2000,Vigna2016TC, DangForsyth2016, DangForsyth2014}). The second category, namely the time-consistent or game theoretical approach, guarantees the time-consistency of the resulting optimal strategy by imposing a time-consistency constraint (\cite{BasakChabakauri2010,BjorkMurgoci2014,WangForsyth2011,CongOosterlee2016, PvSDangForsyth2018_TCMV}).
{\myblue{The time-inconsistency of pre-commitment strategies is because}}
the variance term in the MV-objective is not separable in the sense of dynamic programming
(see \cite{BasakChabakauri2010, Vigna2016TC}). However, pre-commitment
strategies are typically {\myblue{time-consistent}} under an alternative induced
objective function \cite{StrubLiCui2019}, and hence implementable.
The merits and demerits of time consistent and pre-commitment strategies
are also discussed in \cite{vigna2022tail}.
In subsequent discussions, unless otherwise stated, both time consistent and pre-commitment strategies
are  collectively referred to strategies or controls.

\subsection{Background}
In the parametric approach, a parametric stochastic model is postulated, e.g.\ diffusion dynamics,
and then is calibrated to market-observed data.\footnote{Recently, data-driven (i.e.\ non-parametric) methods have been proposed for portfolio optimization under different
optimality criteria, including mean-variance  \cite{LiForsyth2019, butler2021data, ni2022optimal}.
Nonetheless, monotonicity of NN-based methods has not been established.}
A key concern about, and perhaps also a criticism against, MV portfolio optimization in a parametric setting is
its potential lack of robustness to model misspecification error.
This criticism originated from the fact that,
in single-period settings, MV portfolio optimization can provide notoriously unstable asset allocation
strategies arising from small changes in the underlying asset parameters (\cite{MichaudBOOK, Sato2019, PerrinRoncalli2019,BourgeronEtAl2018}).
Nonetheless, in the case of multi-period MV optimization, research findings  indicate that,
when the risky asset dynamics are allowed to follow pure-diffusion dynamics (e.g.\ GBM) or any of the standard finite-activity
jump-diffusion models commonly encountered in financial settings, such as those considered in this work,
the pre-commitment and time-consistent MV outcomes of terminal wealth are generally very robust to model
misspecification errors \cite{PvSDangForsyth2019_Robust}.

It is well-documented in the finance literature that jumps are often present
in the price processes of risky assets (see, for example, \cite{ContTankovBOOK,RamezaniZeng2007}).
In addition, findings in previous research work on MV portfolio optimization (pre-commitment and time-consistency strategies)
also indicate that (i) jumps in the price processes of risky assets, such as Merton model \cite{MertonJumps1976}
and the Kou model \cite{KouOriginal}, and (ii) realistic investment  constraints, such as no-bankruptcy or leverage, have substantial impact on efficient frontiers and optimal investment strategies of MV portfolio optimization \cite{DangForsyth2014, PvSDangForsyth2018_TCMV}.
Furthermore, the results of \cite{MaForsyth2016} show that the effects of stochastic
volatility, with realistic mean-reverting dynamics, are not important for long-term investors with time
horizons greater than 10 years.

Furthermore,  for multi-period MV optimization, it is documented in the literature that
the composition of the risky asset basket remains relatively stable over time, which suggests that the primary question remains
the overall risky asset basket vs.\ the risk-free asset composition of the portfolio, instead of the exact composition of the risky asset basket.
See the available analytical solutions for multi-asset time-consistent MV problems  (see, for example, \cite{ZengLi2011}) as well as pre-commitment MV problems (see for example \cite{LiNg2000}). Therefore, it is reasonable to consider a well-diversified index, instead of a single stock or a basket of stocks,
as common in the MV literature  \cite{PvSDangForsyth2018_MQV, van2021practical, PvSDangForsyth2019_Robust,  DangForsythVetzal2017, PvSDangForsyth2019_Distributions,  PvSDangForsyth2018_TCMV}. This is the modeling approach adopted in this work, resulting in
a low dimensional multi-period MV optimization problem.

\subsubsection{Monotonicity and no-arbitrage}
In general, since solutions to stochastic optimal control problems, including that of the MV portfolio optimization problem,
are often non-smooth, convergence issues of numerical methods, especially monotonicity considerations, are of primary importance.
This is because, in the context of numerical methods for optimal control problems, optimal decisions are determined by comparing numerically computed value functions. Non-monotone schemes could produce numerical solutions that fail to converge to financially relevant solution,
i.e.\  a violation of the discrete no-arbitrage principle \cite{Oberman2006, pooley2003, warin2016}.

To illustrate the above point further, consider a generic time-advancement scheme
from time-($m$-$1$) to time-$m$ of the form
\EQ
\label{eq:ex_scheme}
v_n^{m} = \sum_{\ell \in \mathcal{L}_n} \omega_{n, \ell}~v_{\ell}^{m-1}.
\EN
Here,
$\omega_{n, \ell}$ are the weights and and  $\mathcal{L}_n$ is an index set typically capturing  the computational stencil associated
with the $n$-th spatial partition point. This time-advancement scheme is monotone if, for any $n$-th spatial partition point, we have
$\omega_{n, \ell} \ge 0$, $\forall \ell \in \mathcal{L}_n$. Optimal controls at time-$m$ are determined typically  by comparing  candidates numerically computed
from applying intervention on time-advancement results $v_n^{m}$.  Therefore, these candidates
need to be approximated using a monotone scheme as well. If interpolation is needed in this step,
linear interpolation is commonly chosen, due to its monotonicity\footnote{Other non-monotone interpolation schemes are discussed in, for example,  \cite{FC1980, RF2016}.}.
Loss of monotonicity occurring in the time-advancement may result in $v_n^{m} < 0$ even
$v_{\ell}^{m-1} \ge0$ for all $\ell \in \mathcal{L}_n$.

\subsubsection{Numerical methods}
For stochastic optimal control problems with a small number of stochastic factors,
the PDE approach is often a natural choice.
To the best of our knowledge, finite difference (FD) methods remain the only pointwise convergent methods established for pre-commitment and time-consistent MV portfolio optimization in realistic investment scenarios. These scenarios involve the simultaneous application of various types of investment constraints and modeling assumptions, including jumps in the price processes of risky assets, as highlighted in \cite{DangForsyth2014, PvSDangForsyth2018_TCMV}.
These FD methods achieve monotonicity in time-advancement through a positive coefficient finite difference discretization method (for the partial derivatives), which is combined with implicit time-stepping. Despite their effectiveness, finite difference methods present significant computational challenges in multi-period settings with long maturities. In particular, they necessitate time-stepping between rebalancing dates, which often occur annually (i.e., control monitoring dates). This time-stepping requirement introduces errors and substantially increase the computational cost of FD  methods.

Fourier-based integration methods frequently rely on the presence of an analytical expression for the Fourier transform of the underlying transition density function, or an associated Green's function, as highlighted in various research such as \cite{Pavel2015, Pavel2016, alonso-garcca_wood_ziveyi_2018, Huang2018, ForsythLabahn2017, online23, LuDang2022}. Notably, the Fourier cosine series expansion method \cite{Fang2008, Ruijter2013} can achieve high-order convergence for piecewise smooth problems. However, in cases of optimal control problems, which are usually non-smooth, such high-order convergence should not be anticipated.

When applicable, Fourier-based methods offer unique advantages over FD methods and Monte Carlo simulation.  These advantages include the absence of timestepping errors between rebalancing (or control monitoring) dates, and the ability to handle complex underlying dynamics such as jump-diffusion, regime-switching, and stochastic variance in a straightforward manner.
However, standard Fourier-based methods, much like Monte Carlo simulations, do have a significant drawback: they can potentially lose monotonicity. This potential loss of monotonicity in the context of variable annuities is discussed in depth in \cite{Huang2018, Huang2015}.

In more detail, consider $g(s, s', t_m - t_{m-1})$ as the underlying (scaled) transition density, or a related Green's function. For L\'{e}vy processes, which have independent and stationary increments, $g(\cdot)$ relies on $s$ and $s'$ only through their difference, i.e., $g(s, s', \cdot)= g(s-s', \cdot)$. Thus, the advancement of solutions between control monitoring dates takes the form of a convolution integral as follows
\EQ
\label{eq:gen}
v(s, t_{m-1}) = \int_{\mathbb{R}} g\l(s - s', t_m - t_{m-1}\r) v\l(s', t_m\r)\md s'.
\EN
In the case of L\'{e}vy processes, even though $g(\cdot)$ is not known analytically, the L\'{e}vy-Khintchine formula provides an explicit representation of the Fourier transform (or the characteristic function) of $g(\cdot)$, denoted by $G(\cdot)$. This permits the use of Fourier series expansion to reconstruct the entire integral \eqref{eq:gen}, not just the integrand. The approach creates a numerical integration scheme of the form \eqref{eq:ex_scheme}, with the weights $\omega_{n, \ell}$ typically available in the Fourier domain via $G(\cdot)$. Consequently, the algorithms boil down to the utilization of finite FFTs, which operate efficiently on most platforms. However, there is no assurance that the weights $\omega_{n, \ell}$ are non-negative for all $n$ and $l$, which can potentially lead to a loss of monotonicity.

As highlighted in \cite{barles-souganidis:1991}, the requirement for monotonicity in a numerical scheme can be relaxed. This notion of weak monotonicity was initially explored in \cite{Bokanowski2018} and was later examined in great detail in \cite{ForsythLabahn2017, online23, LuDang2022, LuDang2023} for general control problems in finance, including variable annuities. More specifically, the condition for monotonicity, i.e.\ $\omega_{n, \ell} \ge 0$ for all $\ell \in \mathcal{L}_n$, is relaxed to $\sum{\ell \in \mathcal{L}_n} |\min(\omega{n, \ell}, 0)| \le \epsilon$, with $\epsilon > 0$ being a user-defined tolerance for monotonicity. By projecting the underlying transition density or an associated Green's function onto linear basis functions, this approach allows for full control over potential monotonicity loss via the tolerance $\epsilon > 0$: the potential monotonicity loss can be quantified and restricted to $\mathcal{O}(\epsilon)$, thereby enabling (pointwise) convergence as $\epsilon \to 0$.

\subsection{Objectives}
In general, many industry practitioners find implementing monotone finite difference methods for jump-diffusion models to be complex and time-consuming, particularly when striving to utilize central differencing as much as possible, as proposed in \cite{wang08}.
As well-noted in the literature (e.g.\ \cite{pooley2003,RF2016}), many seemingly reasonable finite difference discretization schemes can yield incorrect solutions. In addition, while the concept of (strict) monotonicity in numerical schemes is directly tied to the discrete no-arbitrage principle, making it easy to comprehend, weak monotonicity is less clear, which further hinders its application in practice. Moreover, the convergence analysis of weakly monotone schemes is often complex, potentially introducing additional obstacles to their practical application.


This paper aims to fill the aforementioned research gap through the development of
an efficient, easy-to-implement and monotone numerical integration method for MV portfolio optimization
under a realistic setting. This setting involves the simultaneous application of different types of investment constraints
and jump-diffusion dynamics for the price processes of risky assets.
In our method, the guarantee of the key monotonicity property comes with a trade-off: a modest level of analytical tractability is required for the jump sizes. Essentially, this stipulates that the associated jump-diffusion model should provide a closed-form expression for European vanilla options. This condition, however, conveniently accommodates a number of widely used models in finance, marking it as a relatively mild requirement in practice.
To demonstrate this, we presents results for two popular distributions for jump sizes in finance: the normal distribution as put forth by \cite{MertonJumps1976}, and the asymmetric double-exponential distribution as proposed in \cite{KouOriginal}.
Although we focus on the pre-commitment strategy case, the proposed method can be extended to time-consistent MV optimization
in a straightforward manner by enforcing a time-consistency constraint as in \cite{PvSDangForsyth2018_TCMV}.
%
%

The main contributions of the paper are as follows.
\begin{itemize}
\item[(i)]
 We present a recursive and localized formulation of the pre-commitment MV portfolio optimization
 under a realistic context that involves the simultaneous application of different types of investment constraints
 and jump-diffusion dynamics \cite{KouOriginal, MertonJumps1976}.
Over each rebalancing interval, the key component of the formulation of MV portfolio optimization is a convolution integral involving a conditional density of the logarithm of amount invested in the risky asset.

\item[(ii)]
Through a known closed-form expression of the Fourier transform of the underlying transition density,
we derive an infinite series representation for this density in which all the terms of the series
are non-negative and readily computable explicitly. Therefore, the convolution integral can be approximated in a straightforward manner using a monotone integration scheme via a composite quadrature rule.
A significant benefit of the proposed method compared to existing monotone finite difference methods is that it offers a computational complexity that is an order of magnitude lower. Utilizing the Toeplitz matrix structure, we propose an efficient implementation of the proposed
monotone integration scheme via FFTs.

\item[(iii)]
We mathematically demonstrate that the proposed monotone scheme is also $\ell_{\infty}$-stable and pointwise consistent with the convolution integral formulation.
We rigourously prove the pointwise convergence of the scheme as the discretization parameter approach zero.
As the rebalancing time interval approaches zero, we intuitively demonstrate that the proposed scheme converges
to a continuously observed impulse control formulation for MV optimization in the form of a Hamilton-Jacobi-Bellman Quasi-Variational Inequality.

\item[(iv)]
All numerical experiments are conducted using model parameters calibrated to
inflation-adjusted, long-term US market data (89 years), enabling realistic conclusions to be
drawn from the results. Numerical experiments demonstrate an agreement with benchmark results obtained by FD method and Monte Carlo simulation
of  \cite{DangForsyth2014}.
\end{itemize}
Although we focus specifically on monotone integration methods for multi-period MV portfolio optimization, our comprehensive and systematic approach could serve as numerical and convergence analysis framework for the development of similar monotone integration methods
for other multi-period or continuously observed control problems in finance.

In Section~\ref{sc:model}, we describe the underlying dynamics and a multi-period rebalancing framework for MV portfolio optimization. A localization of the
pre-commitment MV portfolio optimization in the form of an convolution integral together with appropriate boundary conditions are presented in Section~\ref{sc:MV}. Also therein, we present an infinite series representation of the transition density.
A simple and easy-to-implement monotone numerical integration method via a composite quadrature rule is described in Section~\ref{sc:num}.
In Section~\ref{sc:conv}, we mathematically establish pointwise convergence the proposed integration method.
Section~\ref{sc:conti} explore possible convergence between the proposed scheme and a Hamilton-Jacobi-Bellman Quasi-Variational Inequality
arising from continuously  observed impulse control formulation for MV optimization.
Numerical results  are given in Section~\ref{sc:num}.
Section~\ref{sc:conclude} concludes the paper and outlines possible future work.

\section{Modelling}
\label{sc:model}
We consider portfolios consisting of a risk-free asset and a well-diversified stock index (the risky asset).
{\red{With respect to the risk-free asset, we consider different lending and borrowing rates. Specifically, we denote by $r_b$ and $r_{\iota}$
the positive, continuously compounded rates at which the investor can respectively borrow funds or earn on cash deposits (with $r_b > r_{\iota}$). We make the standard assumption that the real world drift rate $\mu$ of the risky asset is strictly greater than $r_{\iota}$.
Since there is only one risky asset, with a constant risk-aversion parameter, it is never MV-optimal to short stock.}}
Therefore, the amount invest in the risky-asset is non-negative for all  $t\in [0, T]$,
where $T>0$ denotes the fixed investment time horizon or maturity.
In contrast, we do allow short positions in the risk-free asset, i.e.\ it is possible that the amount
invested in the risk-free asset is negative. With this in mind, we denote by
$B_t \equiv B(t)$ the time-$t$ amount invested in the risk-free asset
and by $S_t \equiv S(t)$ the natural logarithm of the time-$t$ amount invested in the risky
(so that $e^{S_t}$ is the amount).

For defining the jump-diffusion model dynamics, let $\xi$ be a random variable denoting the jump size. For any functional $f$, we let $f_{t^{-}} := \lim_{\epsilon \rightarrow 0^+} f_{t-\epsilon}$ and $f_{t^{+}} := \lim_{\epsilon \rightarrow 0^+} f_{t+\epsilon}$. Informally, $t^{-}$ (resp.\ $t^{+}$) denotes the instant of time immediately before (resp.\ after) the forward time $t \in [0,T]$. When a jump occurs, we have
$S_{t} = S_{t^{-}} + \xi $.

\subsection{Discrete portfolio rebalancing}
Define $\mathcal{T}_{M}$ as the set of $M$ predetermined, equally
spaced rebalancing times in $\left[0,T\right]$,
\begin{eqnarray}
\mathcal{T}_{M} & = & \left\{ \left.t_{m}\right|t_{m}=m\Delta t,\;m=0,\ldots,M-1\right\} ,\quad\Delta t=T/M.\label{eq:T_N}
\end{eqnarray}
We adopt the convention that $t_M = T$ and the portfolio is not rebalanced at the end of the investment horizon
$t_M = T$. The evolution of the portfolio over a rebalancing interval $[t_{m-1}, t_{m}]$, $t_{m-1}\in \Tcal_{M}$,
can be viewed as consisting of three steps as follows.
Over $[t_{m-1}, t_{m-1}^+]$, $(S_t, B_t)$, change according to some rebalancing strategy (i.e.\ an impulse control).
Over the time period $[t_{m-1}^+, t_m^-]$,
there is no intervention by the investor according to some control (investment strategy),
and therefore $(S_t, B_t)$ are uncontrolled, and are assumed to follow some dynamics for all
$t \in [t_{m-1}^+, t_m^-]$.
Over $[t_m^-, t_m]$, the settlement (payment or receipt) of interest {\myblue{due for}} the time period $[t_{m-1}, t_{m}]$.
In the following, we first discuss stochastic modeling for $(S_t, B_t)$ over $[t_{m-1}^+, t_m^-]$,
then describe settlement of interest and modelling of rebalancing strategies using impulse controls.

Over the time period $[t_{m-1}^+, t_m^-]$,
in the absence of control (investor's intervention according to some control strategy),
the amounts in the risk-free and risky assets are assumed to have the
following dynamics:
\EQA
\label{eq:FX-HHW.1}
\md B_t &=& \Rcal\l(B_t\r) B_t \, \md t,  \text{ where } \Rcal\l(B_t\r) =
r_{\iota} + (r_b - r_{\iota}) \Ibb_{\{ B_t < 0\}},
\\
\md S_t &=&
\l(\mu - \lambda \kappa - \f{\sigma^2}{2} \r) \md t
+ \sigma \, \md W_t\
+ \md\l( \sum_{\ell = 1}^{\pi_t} \xi_{\ell}  \r), \quad
t\in [t_{m-1}^+, t_{m}^-].
\nonumber
\ENA
Here, as noted earlier,  $r_b$ and $r_{\iota}$ denote the positive, continuously compounded rates at which the investor can respectively borrow funds or earn on cash deposits (with $r_b > r_{\iota}$), while $\Ibb_{\l[A\r]}$ denotes the indicator function of the event $A$;
$\{W_t\}_{t \in [0, T]}$ is a standard Wiener process,
and $\mu$ and $\sigma$ are the real world drift rate and the instantaneous volatility, respectively.
The jump term $\sum_{\ell=1}^{\pi(t)}\xi_{\ell}$ is a compound Poisson process. Specifically, $\{\pi(t)\}_{0\le t \le T}$ is a Poisson process with a constant finite jump intensity $\lambda\geq 0$; and, with $\xi$ being a random variable representing the jump size, $\{\xi_{\ell}\}_{\ell = 1}^{\infty}$ are independent and identically distributed (i.i.d.) random variables having the same same distribution as the random variable $\xi$.  In the dynamics \eqref{eq:FX-HHW.1},  $\kappa=\mathbb{E}\left[e^{\xi}-1\right]$. Here,  $\mathbb{E}[\cdot]$ is the expectation operator taken under
a suitable measure. The Poisson process $\{\pi(t)\}_{0\le t \le T}$, the sequence of random variables $\{\xi_{\ell}\}_{\ell = 1}^{\infty}$,
and the Wiener process and $\{W_t\}_{0\le t \le T}$ are mutually independent.

We consider two distributions for the random variable $\xi$, namely the normal distribution \cite{MertonJumps1976} and the
asymmetric double-exponential distribution \cite{KouOriginal}.
To this end, let $p(y)$ be the probability density function (pdf) of $\xi$.
In the former case,  $\xi \sim \text{Normal}\l(\widetilde{\mu}, \widetilde{\sigma}^2\r)$,
so that its pdf is given by
\EQA
\label{eq:log_norm_pdf}
{\myblue{p(y)}} = \f{1}{\sqrt{2\pi\widetilde{\sigma}^2}}
\exp\l\{-\f{\l(y - \widetilde{\mu}\r)^2}{2\widetilde{\sigma}^2}\r\}.
\ENA
Also, in this case, $\mathbb{E}[e^{\xi}] = \exp( \widetilde{\mu} + \widetilde{\sigma}^2/2)$, and hence
$\kappa=\mathbb{E}\left[e^{\xi}-1\right]$ can be computed accordingly.
In the latter case, we consider an asymmetric double-exponential distribution for $\xi$.
Specifically, we consider
$\xi\sim \text{Asym-Double-Exponential}(q_1,\eta_1,\eta_2)$, $\l(q_1 \in(0, 1), \, \eta_1 > 1, \, {\red{\eta_2 > 0 }}\r)$
so that its pdf is given by
\EQA
\label{eq:log_exp_pdf}
{\myblue{p(y)}}  = q_1 \eta_1 e^{-\eta_1 y} \Ibb_{\l[y \ge 0\r]}
{\red{+ q_2 \eta_2 e^{\eta_2 y } \Ibb_{\l[ y < 0\r]},}}
\quad
q_1+q_2=1.
\ENA
Here $q_1$ and $q_2 = 1- q_1$ respectively are the probabilities of upward and downward
jump sizes. In this case,
$\mathbb{E}[e^\xi] = \frac{q_1\eta_1}{\eta_1 - 1} + \frac{q_2\eta_2}{\eta_2 + 1}$,
so $\kappa=\mathbb{E}\left[e^{\xi}-1\right]$ can be computed accordingly.

\subsection{Impulse controls}
Discrete portfolio rebalancing is modelled using the discrete impulse
control formulation as discussed in for example \cite{DangForsyth2014,PvSDangForsyth2018_TCMV,PvSDangForsyth2018_MQV},
which we now briefly summarize below. Let $c_{m}$ denote the impulse applied
at rebalancing time $t_{m}\in\mathcal{T}_{M}$, which corresponds
to the amount invested in the risk-free asset according to the investor's intervention
at time $t_{m}$, and let $\mathcal{Z}$ denote the set of admissible
impulse values, i.e.\ $c_{m} \in \mathcal{Z}$ for all  $t_m \in \Tcal_M$.

Let $X_t = \l(S_t, B_t\r), \, t \in \l[0,T\r]$ be the multi-dimensional underlying process,
and $x = (s,b)$ denote the state of the system.
Suppose that at time-$t_m$, the state of the system is
$x=\left(s,b\right)=\left(S\left(t_{m}\right),B\left(t_{m}\right)\right)$
for some $t_{m}\in\mathcal{T}_{M}$.
We denote by $(S_{t_m^+}, B_{t_m^+}) \equiv \left(s^+(s, b, c_m),b^+(s, b, c_m)\right)$
the state of the system immediately after the application of
the impulse $c_{m}$ at time $t_{m}$, where
\begin{eqnarray}
{\myblue{S_{t_m^+} \equiv s^+(s, b, c_m)  =   \ln\l(\max(e^s+b - c_{m} - \delta, \, \einf)\r), }}
\quad B_{t_m^+} \equiv b^+(s, b, c_m) = c_{m},
\quad t_m \in \Tcal_M.\label{eq:S+_B+}
\end{eqnarray}
{\myblue{Here, $\delta \ge 0$ is a fixed cost\footnote{It is straightforward to include a proportional cost into \eqref{eq:S+_B+}
as in \cite{DangForsyth2014}. However, to focus on the main advantages of the proposed method, we do not consider a proportional cost in this work.}; }}
since $\log(\cdot)$ is undefined if {\myblue{$e^s+b - c_{m} - \delta\le 0$}}, the amount $S_{t_{m}^+}$ becomes  {\myblue{$\ln\l(\max(e^s+b - c_{m} - \delta, \, \einf)\r)$}}
for a finite $\winf \ll 0$. 

Associated with the fixed set of rebalancing times $\Tcal_M$, defined in \eqref{eq:T_N},
an impulse control $\Ccal$ will be written as the set of impulse values
\EQA
\label{eq:control_C}
\Ccal =
\l\{ c_{m} \,|\, c_{m} \in \Zcal, \, m=0,\ldots, M-1\r\},
\ENA
and we define $\Ccal_{m}$ to be the subset of the control $\Ccal$ applicable to the set of times $\l\{\tn , \ldots, t_{M-1} \r\}$,
\EQA
\label{eq:control_C_n}
\Ccal_{m} 
=
\l\{c_l \,| \, c_l \in \Zcal, \, l=m,\ldots,M-1\r\}
\subset
\Ccal_0 \equiv \Ccal.
\ENA
In a discrete setting, the amount invested in the risk-free asset
changes only at rebalancing date. Specifically,
{\myblue{over each time interval $\l[t_{m-1}, t_m\r], \, m = 1,\ldots,M$,
we}} suppose the amount invested in
the risk-free asset at time $t_{m-1}^+$ {\myblue{after rebalancing}} being $B_{t_{m}^+} = b$.
For test function $f(S_t, B_t, t)$ with both $S_t$ and $B_t$ varying,
we model the change in $f(S_t, B_t, t)$ with $(S_t, B_t = b)$ for
$t \in [t_{m-1}^+, t_{m}^-]$.
{\myblue{Then, the amount in the risk-free asset would jump to $b e^{R(b) \Delta t}$ at time $t_m$,}}
reflecting the settlement (payment or receipt) of interest due for the time interval $[t_{m-1}, t_{m}]$,
$m = 1, \ldots, M$. Here, we note that, although there is no rebalancing at time $t_M = T$,
there is still settlement of interest for the interval $[t_{M-1}, t_{M}]$.

\subsection{Investment constraints}
\label{subsec:admissible_ctrl}
With the time-$t$ state of the system being $(s,b)$,
{\myblue{to include transaction cost, we define the liquidation value $W_{\t{liq}}(t) \equiv W_{\t{liq}}(s,b)$ to be
\EQA
\label{eq:W_dis}
W_{\t{liq}}(t) \equiv W_{\t{liq}}(s, b) = e^s + b - \delta, \quad t \in [0, T].
\ENA
}}
We strictly enforce two realistic investment constraints on the joint values of $S$ and $B$, namely a solvency
condition and a maximum leverage condition. The solvency condition takes the following form:
{\myblue{when $W_{\t{liq}}(s, b) \le 0$}}, we require that the position in the risky asset
be liquidated, the total remaining wealth be placed in the risk-free asset, and the ceasing of all
subsequent trading activities. Specifically, assume that the system is in the state
$x = \l(s,b\r) \in \Omega^{\i}$ at time $\tn$, where
 $\tn\in \Tcal_M$ and
\EQ
\label{eq:Omega}
\Omega^{\i} = \l(-\i, \i\r) \times \l(-\i, \i\r).
\EN
We define a solvency region $\Ncal$ and an insolvency or bankruptcy region $\Bcal$ as follows
{\myblue{
\EQA
\label{eq:NBcal_def}
\Ncal = \l\{\l(s,b\r) \in \Omega^{\i}: \, W_{\t{liq}}(s,b) > 0 \r\},
\,\,
\Bcal = \l\{\l(s,b\r) \in \Omega^{\i}: \, W_{\t{liq}}(s,b) \le 0 \r\}, \,\,
W_{\t{liq}}(s,b) \t{ defined in \eqref{eq:W_dis}}.
\ENA
}}
The solvency constraint can then be stated as
\EQA
\label{eq:solvency}
\t{If } \l(s,b\r) \in \Bcal \t{ at } \tn
\Rightarrow
\begin{cases}
\t{we require} \l( S_{t_m^+} = \winf, B_{t_m^+} = W(s,b)\r)
\\
\t{and {\myblue{$S_t$ remains}} so } \forall t \in \l[\tn^+, T\r],
\end{cases}
\ENA
where, as noted above, $\winf\ll 0$ and is finite.
This effectively means that  the investment in the risky asset has to be liquidated, the total wealth is
to be placed in the risk-free asset, and all subsequent trading activities much cease.

The maximum leverage constraint specifies that the leverage ratio after rebalancing
at $\tn$, where $\tn\in \Tcal_M$, is stipulated by \eqref{eq:S+_B+}
must satisfy
\EQA
\label{eq:leverage}
\f{ {\myblue{\exp(S_{t_m^+})}} }{\exp(S_{t_m^+}) + B_{t_m^+}} \le q_{\max},
\ENA
for some positive constant $q_{\max}$  typically in the range $\l[1.0,2.0\r]$.
Given above the solvency constraint and the maximum leverage constraint,
the set of admissible impulse values, namely the set $\Zcal$  is therefore
defined as follows
\begin{eqnarray}
\label{eq:Zcal_def}
  \mathcal{Z} &=&
  \begin{cases}
    ~\Big\{c_{m} \equiv B_{t_m^+} \in
  \mathbb{R}:
  (S_{t_m^+}, B_{t_m^+}) \text{ via \eqref{eq:S+_B+}}
   \Big\}
   \qquad
    \text{no constraints},
\\
  ~
      \begin{cases}
    \left\{c_{m} \equiv B_{t_m^+} \in
  \mathbb{R}:
  (S_{t_m^+}, B_{t_m^+})  \text{ via \eqref{eq:S+_B+}, s.t.\ }
     S_{t_m^+} \ge \winf \text{ and \eqref{eq:leverage}}
   \right\}
   &
   (s,b) \in \mathcal{N}
   \\
   \left\{{c_{m} = {\myblue{ W_{\t{liq}}(s,b)} }} \right\}
   &
   (s,b) \in \mathcal{B}
    \end{cases}
    \\
        \qquad \qquad\qquad \qquad\qquad \qquad\qquad \qquad\qquad \qquad
        \text{solvency \& maximum leverage}
    \end{cases}
    \nonumber
\label{Z_sol_def}
\end{eqnarray}
Based on the definition~\eqref{eq:control_C_n}, the set of admissible impulse controls is given by
\EQA
\label{A_sol_def}
\mathcal{A} &=& \l\{ \Ccal_{m} \, | \,
\Ccal_{m} \t{ defined in \eqref{eq:control_C_n}}, \, m = 0,\ldots, M-1 \r\}.
\ENA

\section{Formulation}
\label{sc:MV}
Let {\myblue{$E_{\Ccal_m}^{x,\tn} \l[W_{\t{liq}}(T)\r]$ and $Var_{\Ccal_m}^{x,\tn} \l[W_{\t{liq}}(T)\r]$}}
respectively denote the mean and variance of the terminal {\myblue{liquidation}} wealth,
given the system state $x = (s,b)$ at time $t_m$ for some $t_m \in \Tcal_M$
following the control $\Ccal_{m} \in \Acal$ over $[t_m, T]$, assuming the underlying dynamics \eqref{eq:FX-HHW.1}.
The standard scalarization method for multi-criteria optimization problem in \cite{yu1974cone} gives the mean-variance (MV) objective as
\EQA
\label{eq:MV}
\sup_{\Ccal_m \in \Acal}
\l\{
E_{\Ccal_m}^{x, \tn} \l[{\myblue{ W_{\t{liq}}(T) }}\r] -
\rho \cdot
Var_{\Ccal_m}^{x, \tn} \l[{\myblue{ W_{\t{liq}}(T) }}\r]
\r\},
\ENA
where the scalarization parameter $\rho >0$ reflects the investor's risk aversion level.
\subsection{Value function}
Dynamic programming cannot be applied directly to \eqref{eq:MV}, since no smoothing property of conditional expectation for variance. The technique of \cite{li2000optimal, zhou2000continuous} embeds \eqref{eq:MV} in a new optimisation problem, often referred to as the embedding problem, which is amenable to dynamic programming techniques.
We follow the example of \cite{cong2017robust, dang2014continuous} in defining the PCMV optimization problem as the associated embedding MV problem\footnote{For a discussion of the elimination of spurious optimization results when using the embedding formulation, see \cite{dang2016convergence}.}.
Specifically, with $\gamma \in \mathbb{R}$ being the embedding parameter, we define the value function $v \l(s,b,\tn\r)$, $m = M-1,\ldots,0$ as
follows
\EQA
\label{eq:PCMV_dis}
(PCMV_{\Delta t}(\tn;\gamma)):
\quad
v \l(s,b,\tn\r) =
\inf_{\Ccal_m \in \Acal}
E_{\Ccal_m}^{x, \tn}
\l[\l( {\myblue{ W_{\t{liq}}(T) }} - \f{\gamma}{2}\r)^2\r],
\quad
\gamma \in \mathbb{R}, \quad m = 0,\dots,M-1,
\ENA
where $W_T$ is given in \eqref{eq:W_dis}, subject to dynamics \eqref{eq:FX-HHW.1} between
rebalancing times. We denote by $\Ccal_m^*$ the optimal control for the problem $PCMV_{\Delta t} (\tn;\gamma)$,
where
\EQA
\label{eq:optimal_ctrl}
\Ccal_m^* = \l\{ c_m^*,\ldots,c_{M-1}^* \r\},
\quad m = 0, \ldots, M-1.
\ENA
For an impulse value $c\in \mathcal{Z}$,  we define the intervention operator
$\mathcal{M}(\cdot)$ applied at $t_m \in \Tcal_M$ as follows
\EQ
\label{eq:Operator_M}
\mathcal{M}(c)~v (s, b , t_m^+) =
  v\l(s^+(s, b, c), b^+(s, b, c), t_m^+\r), \quad s^+(s, b, c) \text{ and } b^+(s, b, c)
  \text{ are given in \eqref{eq:S+_B+}}.
\EN
By dynamic programming arguments \cite{Pham, puter94}, for a fixed  embedding parameter $\gamma \in \mathbb{R}$,
and $(s, b) \in \Omega^{\i}$, the recursive relationship for the value function $v(s, b, t_m)$ in \eqref{eq:PCMV_dis} is given by
\begin{linenomath}
\begin{subequations}\label{eq:recursive_comp}
\begin{empheq}[left={\empheqlbrace}]{alignat=6}
&v \l(s,b,\tn\r) &~=~& {\myblue{\l({\myblue{ W_{\t{liq}}(s,b) }} - \f{\gamma}{2} \r)^2}},
&& \quad\quad m = M,
\label{eq:recursive_b_comp_initial}
\\
&v \l(s,b,\tn\r) &~=~&  \min \left\{v \l(s,b,\tn^+\r), \inf_{c \in \mathcal{Z}} \mathcal{M}(c)~v \l(s,b,\tn^+\r)\right\},
&& \quad\quad m = M-1, \ldots, 0,
\label{eq:recursive_b_comp}
\\
&v \l(s,b,\tn^-\r) &~=~&  v \l(s,be^{R(b) \Delta t},\tn\r),
&&\quad\quad m = M, \ldots, 1,
\label{eq:recursive_interest}
\\
&v(s, b, t_{m-1}^+) &~=~&  \int_{-\infty}^{\infty}
v \l(s',b,\tn^-\r) g(s, s'; \Delta t)~ \md s',  &&\quad\quad m = M, \ldots, 1.
\label{eq:recursive_c_comp_int}
\end{empheq}
\end{subequations}
\end{linenomath}
Here, in \eqref{eq:recursive_b_comp}, the intervention operator $\mathcal{M}(\cdot)$  is
given by  \eqref{eq:Operator_M}, with the $\min\{ \cdot, \cdot\}$ operator reflecting
the optimal choice between no-rebalancing and rebalancing (which is subject to {\myblue{a fixed cost $\delta$ }});
\eqref{eq:recursive_interest} reflects the settlement (payment or receipt)
of interest due for the time interval $[t_{m-1}, t_{m}]$, $m = 1, \ldots, M$.
 In the integral \eqref{eq:recursive_c_comp_int}
the functions $g(s, s'; \Delta t)$ denotes the probability {\myblue{density of $s$, the log of the
amount invested in the risky asset at a future time ($t_{m}^-$)}}, {\myblue{and
the information $s'$ at the current time ($t_{m-1}^+$), given $\Delta t = t_{m} - t_{m-1}$}}.
Also, we note that the fact that amount invested in the risk-free asset does not change
in the in the interval $[t_{m-1}^+, t_m^-]$ is reflected in \eqref{eq:recursive_c_comp_int}
since this amount is kept constant ($=b$) on both sides of \eqref{eq:recursive_c_comp_int}.

It can be shown that $g(s, s';\Delta t)$ has the form $g(s- s'; \Delta t)$, and therefore,
in \eqref{eq:recursive_c_comp_int}, the integral  takes the form of the convolution of $g(\cdot)$ and $v(\cdot, t_m^{-})$.
That is, \eqref{eq:recursive_c_comp_int} becomes
\EQ
v(s, b, t_{m-1}^+) ~=~  \int_{-\infty}^{\infty}
v \l(s',b,\tn^-\r) g(s-s'; \Delta t)~ \md s', \quad m = M, \ldots, 1.
\label{eq:recursive_c_comp_int_1}
\EN
Although a closed-form expression for $g(s; \Delta t)$ is not known to exist, its Fourier transform, denoted by
$G(\cdot; \Delta t)$, is known in closed-form. Specifically, we recall the Fourier transform pair
\EQA
\label{eq:small_g}
\mathfrak{F}[g(s;\cdot)] =  G(\eta; \cdot) = 
\int_{-\infty}^{\infty} e^{-i \eta s}g(s; \cdot)~ds,
~\quad~
\mathfrak{F}^{-1}[G(\eta;\cdot)] = g(s; \cdot) =  
\f{1}{2\pi}
\int_{-\infty}^{\infty} e^{i \eta s}G(\eta; \cdot)~d\eta.
\ENA
A closed-form expression for $G \l(\eta; \Delta t \r)$ is given by
\EQA
G \l( \eta;\Delta t \r) = \exp \left(\Psi \l( \eta \r) \Delta t \right), ~~\text{with}~~
\Psi(\eta) = \l( - \frac{\sigma^2 \eta^2}{2}  + \l( \mu - \lambda \kappa - \frac{\sigma^2}{2} \r)\l(i \eta \r) -  \lambda
+ \lambda \Gamma\l( \eta \r) \r).
\label{eq:big_F}
\ENA
Here, $\Gamma\l( \eta \r) = \int_{-\infty}^{\infty} p(y)~e^{i \eta y}~dy$, where $p(y)$ is the probability density function of
$\ln\left(\xi\right)$ with $\xi$ being the random variable representing the jump multiplier.

\subsection{An infinite series representation of $\boldsymbol{g\l(\cdot\r)}$}
The proposed monotone integration method depends on an infinite series representation of
the probability density function $g(\cdot)$, which is presented in Lemma~\ref{lemma:1}.
\begin{lemma}
\label{lemma:1}
Let $g(s; \Delta t)$ and $G(\eta; \Delta t)$ be a Fourier transform pair defined in \eqref{eq:small_g} and
$G(\eta; \Delta t)$ is given in \eqref{eq:big_F}. Then $g(s; \Delta t) \equiv g\l(s; \Delta t,\infty\r)$ can be written as
\begin{linenomath}
\postdisplaypenalty=0
\begin{align}
g\l(s; \Delta t,\infty\r)
&=
{\red{\f{1}{\s{4 \pi \alpha}}}} \sum_{k=0}^{\infty} \f{\l(\lambda \Delta t\r)^{k}}{k!}
\int_{-\infty}^\infty \ldots \int_{-\infty}^\infty
            	\exp\l(\theta  -
	\f{\l(\beta + s +  Y_k\r)^2}
	{4 \alpha}\r)
    \l(\prod_{\ell = 1}^{k} p(y_{\ell})\r) 
    \md y_1 \ldots \md y_k,
\nonumber\\
{\text{where }} & \alpha = \f{\sigma^2}{2} \, \Delta t,
\quad
\beta = \l(\mu - \lambda \kappa - \f{\sigma^2}{2}\r) \, \Delta t,
\quad
\theta = -\lambda \Delta t,\quad Y_k = \sum_{\ell=1}^{k} y_{\ell}, \quad Y_0 = 0,
\label{eq:g_series}
\end{align}
\end{linenomath}
and $p(y)$ is the PDF of the random variable $\xi$.
{\myblue{When $k = 0$, we have $g\l(s; \Delta t,0\r) = \f{1}{\s{4\pi \alpha }} \exp\l(\theta - \f{\l(\beta + s + \r)^2}{4 \alpha}\r)$.}}
%
\end{lemma}
A proof of of Lemma~\ref{lemma:1} is given in Appendix~\ref{sec:app_g_lemma}.

The infinite series representation in \eqref{eq:g_series} cannot be applied directly for computation as the $k$-th term of the series involves a multiple integral with {\myblue{$\l(\prod_{\ell = 1}^{k} p(y_{\ell})\r)$}}, where $p(y)$ denotes the probability density of $\xi$.
To feasibly evaluate this multiple integral analytically, as noted earlier, our method requires a quite modest level of analytical tractability for the random variable $\xi$. To put it more precisely, we need the associated jump-diffusion model of the form
$
\frac{\md X_t}{X_t^-} = (\mu - \lambda \kappa) \md t + \sigma \, \md W_t\ + \md\l( \sum_{\ell = 1}^{\pi_t} (\xi_{\ell}-1) \r)
$
to yield a closed-form expression for European vanilla options. Arguments in \cite{dangJacksonSues2016}[Section 3.4.1] can be invoked to demonstrate this relationship. Thus, the requirement of analytical tractability for the random variable $\xi$ is not overly stringent, thereby broadening the applicability of our approach in practice.

Next, we demonstrate the applicability  of our approach for two popular distributions for $\xi$ in finance:
when $\xi$ follow a normal distribution \cite{MertonJumps1976}  and
an asymmetric double-exponential distribution \cite{KouOriginal}.
\begin{corollary}
\label{cor:twodis}
For the case $\xi \sim \text{Normal}\l(\widetilde{\mu}, \widetilde{\sigma}^2\r)$
{\myblue{whose}}
PDF is given by \eqref{eq:log_norm_pdf},
the infinite series representation of the conditional density $g(s;\Delta t, \infty)$ given in Lemma~\ref{lemma:1}
is evaluated to
\EQA
\ds
g(s; \Delta t, \infty)
\label{eq:g_log_norm}
&=&
g(s; \Delta t, 0) + \sum_{k=1}^{\infty} \Delta g_k(s; \Delta t),\\
\t{where} \qquad
g(s; \Delta t, 0) &=& \f{\exp\l(\theta - \f{\l(\beta + s + \r)^2}{4 \alpha}\r)}
{\s{4\pi \alpha }},
~~\t{and}~~
\Delta g_k(s; \Delta t) = \f{\l(\lambda \Delta t\r)^{k}}{k!} \, \f{\exp\l(\theta - \f{\l(\beta + s + k \widetilde{\mu} \r)^2}{4 \alpha + 2 k \widetilde{\sigma}^2}\r)}
{\s{4\pi \alpha + 2 \pi k \widetilde{\sigma}^2}},
\nonumber
\ENA
with $\alpha$, $\beta$ and $\theta$ are given in \eqref{eq:g_series}.


For the case $\xi\sim \text{Asym-Double-Exponential}(q_1,\eta_1,\eta_2)$, $\l(q_1 \in(0, 1), \, \eta_1 > 1, \, {\red{\eta_2 > 0}}\r)$
{\myblue{whose}}
PDF given by \eqref{eq:log_exp_pdf},
the infinite series representation of the conditional density $g(s;\Delta t, \infty)$ given in Lemma~\ref{lemma:1}
is evaluated to
$\ds g(s;\Delta t, \infty) = g(s; \Delta t, 0) + \sum_{k=1}^{\infty} \Delta g_k(s; \Delta t)$,
where $g(s; \Delta t, 0) = \f{\exp\l(\theta - \f{\l(\beta + s \r)^2}{4 \alpha}\r)}{\s{4\pi \alpha}}$, and
\begin{align}
\label{eq:g_EJ_double_exp}
\Delta g_k(s; \Delta t) =
&\f{e^{\theta}}{\s{4 \pi \alpha}}
\f{\l(\lambda \Delta t\r)^{k}}{k!}
\l[
\sum_{\ell=1}^k \, Q_1^{k,\ell} \,
\l(\eta_1 \, \sqrt{2\alpha}\r)^{\ell} \,
\e^{\eta_1 \, \l(\beta+s-s' \r) + \eta_1^2 \alpha} \,
\mathrm{Hh}_{\ell-1}\l(\eta_1 \sqrt{2\alpha} + \f{\beta+s-s'}{\sqrt{2\alpha}}\r)
\r. \nonumber\\
& \qquad \qquad \qquad
\l.
+ \sum_{\ell=1}^k \, Q_2^{k,\ell} \,
\l(\eta_2 \, \sqrt{2\alpha}\right)^{\ell} \,
\e^{-\eta_2 \, \l(\beta+s - s' \r) + \eta_2^2 \alpha} \,
\mathrm{Hh}_{\ell-1}\l(\eta_2 \sqrt{2\alpha} - \f{\beta + s - s'}{\sqrt{2\alpha}} \r)
\r].
\end{align}
Here, $\alpha$, $\beta$ and $\theta$ are given in \eqref{eq:g_series}; $Q_1^{k,\ell}$, $Q_2^{k,\ell}$ and $\mathrm{Hh}_{\ell}$ are defined as follows
\begin{align}
\label{eq:PQ}
Q_1^{k,\ell} &= \sum_{i=\ell}^{k-1} \binom{k-\ell-1}{i-\ell} \binom{k}{i}
\l(\f{\eta_1}{\eta_1 + \eta_2}\r)^{i-\ell}\l(\f{\eta_2}{\eta_1 + \eta_2}\r)^{k-i} q_1^{i}q_2^{k-i},
\quad
1 \le \ell \le k-1,
\nonumber\\
Q_2^{k,\ell} &= \sum_{i=\ell}^{k-1} \binom{k-\ell-1}{i-\ell} \binom{k}{i}
\l(\f{\eta_1}{\eta_1 + \eta_2}\r)^{k-i}\l(\f{\eta_2}{\eta_1 + \eta_2}\r)^{i-\ell} q_1^{k-i}q_2^{i},
\quad
1 \le \ell \le k-1,
\end{align}
where $q_1+q_2 = 1$, $Q_1^{k,k}=q_1^k$ and $Q_2^{k,k}=q_2^k$, and
\begin{align}
\label{eq:Hhk}
Hh_{\ell} (x) = \f{1}{\ell !} \int_{x}^{\i} \l(y-x\r)^{\ell} e^{-\f{1}{2} y^2} \md y,
\t{ with }
Hh_{-1} (x) = e^{-x^2/2}, \t{ and }
Hh_{0} (x) = \s{2\pi}\text{NorCDF}(-x).
\end{align}
\end{corollary}
{\myblue{Here, NorCDF denotes CDF of standard normal distribution $\mathcal{N}(0,1)$.}}
For brevity, we omit a straightforward proof for the log-normal case \eqref{eq:g_log_norm} {\myblue{using Equation~\eqref{eq:g_proof_int}}}.
A proof for the log-double exponential case \eqref{eq:g_EJ_double_exp} is given in Appendix~\ref{sec:app_g}.
For this case, we note that function $Hh_{\ell}(\cdot)$ can be evaluated very efficiently using the standard normal density function
and standard normal distribution function via the three-term recursion \cite{AbramowitzStegun1972}
\EQA
\ell \, Hh_{\ell}(x) = Hh_{\ell-2}(x) - x Hh_{\ell-1}(x),
\quad \ell \ge 1.
\nonumber
\ENA
In the subsequent section, we present a definition of the localized problem
to be solved numerically.

\subsection{Localization and problem statement}
The MV formulation \eqref{eq:recursive_comp}
is posed on an infinite domain.  For the problem statement and convergence analysis of numerical schemes, we define a localized MV portfolio optimsation formulation.
To this end, with $s_{\min}^{\dagger} <  s_{\min} < 0 < s_{\max} < s_{\max}^{\dagger}$, $-b_{\max} < 0 < b_{\max}$, where
$|s_{\min}^{\dagger}|$,  $|s_{\min}|$, $s_{\max}$, $s_{\max}^{\dagger}$, and  $b_{\max}$  are sufficiently large,
we define the following spatial sub-domains:
\begin{linenomath}
\postdisplaypenalty=0
\begin{alignat}{8}
\label{eq:sub_domain_whole}
&\Omega &&= [s_{\min}^{\dagger}, s_{\max}^{\dagger}] \times \l[-b_{\max}, b_{\max}\r],
&\quad \Omega_{\Bcal} &= \l\{
\l(s,b\r) \in \Omega \, \backslash \, \Omega_{s_{\max}} \, \backslash \, \Omega_{s_{\min}}: \,
{\myblue{W_{\t{liq}}\l(s, b\r)}} \le 0 \r\},
\nonumber
\\
&\Omega_{s_{\max}}  &&= \l[s_{\max}, s_{\max}^{\dagger}\r] \times \l[-b_{\max}, b_{\max}\r],
&\quad \Omega_{b_{\max}} &= \l(s_{\min}, s_{\max}\r) \times \l[-b_{\max}e^{r_{b}T}, -b_{\max}\r) \cup \l(b_{\max}, b_{\max}e^{r_{\iota}T}\r],
\nonumber
\\
&\Omega_{s_{\min}}  &&=  \l[s_{\min}^{\dagger}, s_{\min}\r] \times \l[-b_{\max}, b_{\max}\r],
&\quad \Omega_{\myin} &= \Omega \, \backslash \, \Omega_{s_{\max}} \, \backslash \, \Omega_{s_{\min}} \, \backslash \, \Omega_{\Bcal}.
\end{alignat}
\end{linenomath}
We emphasize that we do not actually solve the MV optimization problem in $\Omega_{b_{\max}}$.
However, we may use an approximate value to the solution in $\Omega_{b_{\max}}$,
obtained by means of extrapolation of the computed solution in $\Omega_{\myin}$, to provide
any information required by the MV optimization problem in $\Omega_{\myin}$.
We also define the following sub-domains:
\begin{linenomath}
\postdisplaypenalty=0
\begin{alignat}{8}
&\Omega_{s_{\max}^{\dagger}} &&= \l[s_{\max}^{\dagger}, s_{\max}^{\ddagger} \r] \times \l[-b_{\max}, b_{\max}\r],
\quad
\Omega_{s_{\min}^{\dagger}}  =  \l[s_{\min}^{\ddagger}, s_{\min}^{\dagger}\r] \times \l[-b_{\max}, b_{\max}\r],
\nonumber
\\
&&&\qquad \qquad \text{where } s_{\max}^{\ddagger} = s_{\max} - s_{\min}^{\dagger}
\text{ and } \quad s_{\min}^{\ddagger} = s_{\min} - s_{\max}^{\dagger}.
\label{eq:extra_dom}
\end{alignat}
\end{linenomath}
The solutions within the sub-domains $\Omega_{s_{\min}^{\dagger}}$ and $\Omega_{s_{\max}^{\dagger}}$ are not required for our purposes.
These sub-domains are introduced to ensure the well-defined computation of the conditional probability density function $g(\cdot)$ in \eqref{eq:recursive_c_comp_int_1} for the convolution integral \eqref{eq:recursive_c_comp_int_1} in the MV optimization problem within $\Omega_{\myin}$. To simplify our discussion, we will adopt a zero-padding convention going forward. This convention assumes that the value functions within these sub-domains are zero for all time $t$, and we will exclude these sub-domains from further discussions.


Due to rebalancing, the intervention operator $\mathcal{M}(\cdot)$ for $\Omega_{\myin}$,
defined in \eqref{eq:Operator_M}, may require evaluating a candidate value
{\myblue{at a point having $s^+ = \ln(\max( W_{\t{liq}}(s,b)- c , \, e^{\winf}))$, and $s^+$ could be}}
outside $[s_{\min}^{\dagger}, s_{\max}^{\dagger}]$, if $\winf < s^{\dagger}_{\min}$.
{\myblue{Therefore, with $|s^{\dagger}_{\min}|$ selected sufficiently large, we assume $\winf = s^{\dagger}_{\min}$.}}

We now present equations for spatial sub-domains defined in \eqref{eq:sub_domain_whole}.
We note that boundary conditions for $s \to -\infty$ and $s \to \infty$
are obtained by relevant asymptotic forms $e^s \to 0$ and $e^s \to \infty$, respectively,
similar to \cite{DangForsyth2014}. This is detailed below.
\begin{itemize}
\item For  $(s, b, T) \in  \Omega \times\{T\}$,
we apply the terminal condition \eqref{eq:recursive_b_comp_initial}
\EQ
\label{eq:recursive_b_comp_initial*}
v(s, b, T) =  {\myblue{\l(W_{\t{liq}}(s,b) - \f{\gamma}{2} \r)^2}}.
\EN

\item 
{\myblue{For $(s, b, t_m) \in  \Omega \times \Tcal_M$, $m = M-1, \ldots, 0$,}}
the intervention result \eqref{eq:recursive_b_comp} is given by
\EQ
\label{eq:recursive_b_comp_dom}
v \l(s,b,t_m \r) ~=~  \min\left\{v \l(s,b,t_{m}^+\r), \inf_{c \in \mathcal{Z}} \mathcal{M}(c)~v \l(s,b,t_{m}^+\r)\right\},
\EN
where the intervention $\mathcal{M}(\cdot)$ is defined in \eqref{eq:Operator_M}.

\item For $(s, b, t_m^-) \in  \Omega \times\{t_m^-\}$, $m = M, \ldots, 1$,
settlement of interest \eqref{eq:recursive_interest} is enforced by
\EQ
v \l(s,b,\tn^-\r) ~=~  v \l(s,be^{R(b) \Delta t},\tn\r),
\quad m = M, \ldots, 1, \quad
\text{ and $v \l(s,\cdot,\tn\r)$ is given in \eqref{eq:recursive_b_comp_dom}. }
\label{eq:recursive_interest*}
\EN

\item For $(s, b, t_{m}^+) \in  \Omega_{b_{\max}} \times\{t_{m}^+\}$, where $m = M, \ldots, 1$,
we impose the boundary condition
\EQ
\label{eq:bmax}
{\red{
v \l(s, b, t_m^+\r)
= \l(\f{b}{b_{\max}}\r)^2v \l(s, \sgn(b) b_{\max},t_m^+\r).}}
\EN

\item 
{\myblue{
For $(s, b, t_{m-1}^+) \in  \Omega_{s_{\min}} \times\{t_{m-1}^+\}$, where $t_{m-1} \in \Tcal_M$,
from \eqref{eq:recursive_b_comp_initial*}, we assume that $v(s, b, t) \approx A_0(t) b^2$ for some unknown function $A_0(t)$,
which mimics asymptotic behaviour of the value function as $s \to -\infty$ (or equivalently, $e^z \to  0$).}}
Substituting this asymptotic form into the integral~\eqref{eq:recursive_c_comp_int} gives
{\myblue{the boundary condition}}
    \EQA
    \label{eq:left_padding}
    v(s, b, t_{m-1}^+) =  A_0(t_{m}^-) b^{2} \int_{-\infty}^{\infty}  g(s- s'; \Delta t)~ \md s' = v(s, b, t_{m}^-),
    \ENA
    {\myblue{where $v(s, b, t_{m}^-)$ is given by \eqref{eq:recursive_interest*}.}}

\item 
{\myblue{
For $(s, b, t_{m-1}^+) \in  \Omega_{s_{\max}} \times\{t_{m-1}^+\}$, where $t_{m-1} \in \Tcal_M$,
from \eqref{eq:recursive_b_comp_initial*}, for fixed $b$, we assume that $v(z, b, t) \approx A_1(t) e^{2s}$ for some unknown function $A_1(t)$, which mimics asymptotic behaviour of the value function as $s \to \infty$ (or equivalently, $e^z \to  \infty$).}}
We substitute this asymptotic form into the integral~\eqref{eq:recursive_c_comp_int},
noting the infinite series representation of $g(\cdot;\Delta t)$ given Lemma~\ref{lemma:1},
and obtain the correspoding boundary condition:
\EQ
\label{eq:right_padding}
v \l(s, b, t_{m -1}^+\r) =
v(s, b,t_m^-) ~e^{\l(\sigma^2 + 2\mu + \lambda \kappa_2\r)\Delta t},
\quad
\kappa_2 = \Ebb \l[\l(\e^{\xi}-1\r)^2\r],
\EN
{\myblue{where $v(s, b, t_{m}^-)$ is given by}} \eqref{eq:recursive_interest*}.
For a proof, see  Appendix~\ref{sec:v_u_s_max}.

\item
For $(s, b, t_{m-1}^+) \in  \Omega_{\myin} \times\{t_{m-1}^+\}$, where $t_{m-1} \in \Tcal_M$, from the convolution integral $\eqref{eq:recursive_c_comp_int_1}$,
we have
\EQ
\label{eq:omega_in}
v\l(s,b,t_{m-1}^+\r) = \int_{s_{\min}^{\dagger}}^{s_{\max}^{\dagger}} v \l(s',b,\tn^-\r) g(s - s'; \Delta t)~ \md s'.
\EN
where the terminal condition $v \l(s',b,\tn^-\r)$ is given by \eqref{eq:recursive_interest*}.
The conditional density $g(\cdot; \Delta t)$ is given by the infinite series in \eqref{eq:g_series} (Lemma~\eqref{lemma:1}),
and is defined on $[s_{\min}^{\ddagger}, s_{\max}^{\ddagger}]$.

\end{itemize}
In Definition~\ref{def:glwb} below, we formally define the MV portfolio optimization problem
\begin{definition}[Localized MV portfolio optimization problem]
\label{def:glwb}
The MV portfolio optimization problem with the set of rebalancing times
$\Tcal_M$ defined in \eqref{eq:T_N},  and dynamics \eqref{eq:FX-HHW.1} with the PDF {\myblue{$p(y)$}} given by \eqref{eq:log_norm_pdf} or \eqref{eq:log_exp_pdf}, is defined in $\Omega \times {\myblue{\Tcal_M}} \cup \l\{t_M\r\}$ as follows.

At each $t_{m-1} \in \Tcal_M$, the solution to the MV portfolio optimization problem $v(s, b, t_{m-1})$ {\myblue{given by}} \eqref{eq:recursive_b_comp_dom},
{\myblue{where $v(s, b, t_{m-1}^+)$ satisfies}} (i) the integral \eqref{eq:omega_in} in $\Omega_{\myin}\times\{t_{m-1}^+\}$,
(ii) the boundary conditions \eqref{eq:left_padding}, \eqref{eq:right_padding}, and \eqref{eq:bmax} in
$\l\{\Omega_{s_{\min}}, \Omega_{s_{\max}}, \Omega_{b_{\max}}\r\} \times\{t_{m-1}^+\}$, respectively, and
(iii) subject to the terminal condition \eqref{eq:recursive_b_comp_initial*} in $\Omega \times\{t_M\}$,
{\myblue{with the settlement of interest subject to \eqref{eq:recursive_interest*} in $\Omega \times\{t_m^-\}$}}.

\end{definition}

We introduce a result on uniform continuity of the solution to the MV portfolio optimization.
\begin{proposition}
\label{prp:uniform}
The solution $v(s, b, t_{m})$ to the MV portfolio optimization in Definition~\ref{def:glwb} is uniformly continuous
within each sub-domain $\Omega_{\myin} \times \{t_m\}$, $m = M, \ldots, 0$.
\end{proposition}

\begin{proof}
This proposition can be proved using mathematical induction on $m$. For brevity, we outline key details below.
We first note that the domain $\Omega$ is bounded and $T$ is finite.
We observe that if $v(s, b, t)$ is a uniformly continuous function,
then $\ds \inf_{c\in \mathcal{Z}} \mathcal{M}(c) v(s, b, t)$, where $\mathcal{M}(\cdot)$ defined in \eqref{eq:Operator_M},
is also uniformly continuous \cite[Lemma~2.2]{Guo2009}.
As such, $\min\{v(s, b, t), \inf_{c\in \mathcal{Z}} \mathcal{M}(c) v(s, b, t)\}$ is also uniformly continuous
since $\Omega$ is bounded.
Therefore, it follows that if $v(s, b, t_m^+), m = M-1,\ldots,0$, is uniformly continuous then
the intervention result $v(s, b, t_m)$ obtained in \eqref{eq:recursive_b_comp_dom} is also uniformly
continuous.  Next, if $v(s, b, t_m)$, $m = M,\ldots,1$, is uniformly continuous, then the interest settlement result
$v(s, b, t_m^-)$ defined in \eqref{eq:recursive_interest*} is also uniformly continuous.
The other key step is to show that, if $v(s, b, t_{m}^-)$, $m = M, \ldots, 1$, is uniformly continuous, then
the solution $v(s, b, t_{m-1}^+)$ for $(s, b) \in \Omega_{\myin}$ given by the convolution integral \eqref{eq:omega_in} is also uniformly continuous.
Combining these above three steps with the fact that the initial condition $v(s, b,  t_M)$ given in \eqref{eq:recursive_b_comp_initial*}
is uniformly continuous in $(s, b) \in \Omega$, with $\Omega$ a bounded domain, gives the desired result.
\end{proof}
We conclude this section by emphasizing that
the value function may not be continuous across $s_{\min}$ and $s_{\max}$.
The interior domain $\Omega_{\myin}\times \{t_{m}\}$, $m = M-1, \ldots, 0$,
is the target region where provable pointwise convergence of the proposed numerical method is investigated,
which relies on Proposition~\ref{prp:uniform}.

\section{Numerical methods}
\label{sc:num}
Given the closed-form expressions of $g\l(s-s'; \Delta t \r)$,
the convolution integral \eqref{eq:omega_in} is approximated by a discrete convolution which can be efficiently computed
via FFTs.
For our scheme, the intervals $[s^{\dagger}_{\min}, s_{\min}]$ and $[s_{\max}, s^{\dagger}_{\max}]$
also serve as padding areas for nodes in $\Omega_{\myin}$.
Without loss of generality, for convenience, we assume that  $|s_{\min}|$ and $s_{\max}$ are chosen sufficiently
large {\myblue{with}}
\EQA
\label{eq:w_choice_green_jump_form}
s^{\dagger}_{\min} = s_{\min} - \frac{s_{\max} - s_{\min}}{2},
~~~\text{and}~~~
s^{\dagger}_{\max} =  s_{\max} + \frac{s_{\max} - s_{\min}}{2}.
\ENA
With this in mind, $s_{\min}^{\ddagger}$ and $s_{\max}^{\ddagger}$, defined
in \eqref{eq:extra_dom}, are given by
\[
s_{\min}^{\ddagger} = s^{\dagger}_{\min} - s_{\max} =  -\frac{3}{2}\l(s_{\max} - s_{\min}\r),
~~~\text{and}~~~
s_{\max}^{\ddagger} = s^{\dagger}_{\max} - s_{\min}  = \frac{3}{2}\l(s_{\max} - s_{\min}\r).
\]

\subsection{Discretization}
We discretize MV portfolio optimization problem defined in Defn.~\ref{def:glwb} on the localized domain $\Omega$ as follows.
\begin{itemize}
\setlength\itemsep{0.05em}
\item[(i)]
We denote by $N$ (resp.\ $N^{\dagger}$ and $N^{\ddagger}$ ) the number of intervals of a uniform partition of $[s_{\min}, s_{\max}]$
(resp.\ $[s_{\min}^{\dagger}, s_{\max}^{\dagger}]$ and $[s_{\min}^{\ddagger}, s_{\max}^{\ddagger}]$).
For convenience, we typically choose $N^{\dagger} = 2N$ and $N^{\ddagger} = 3N$ so that only one set of $s$-coordinates is needed.
We use an equally spaced partition in the $s$-direction, denoted by
$\{s_n\}$, where
\begin{align}
\label{eq:grid_w}
    s_n &= \hat{s}_0 + n\Delta s;
    ~~
    s = -N^{\ddagger}/2, \ldots, N^{\ddagger}/2, ~~\text{where }
    \hat{s}_0 = \frac{s_{\min} + s_{\max}}{2}= \frac{s^{\dagger}_{\min} + s^{\dagger}_{\max}}{2} = \frac{s^{\ddagger}_{\min} + s^{\ddagger}_{\max}}{2},
    \nonumber
    \\
    \text{and } \Delta s &= \frac{s_{\max} - s_{\min}}{N} ~=~ \frac{ s^{\dagger}_{\max} - s^{\dagger}_{\min}}{N^{\dagger}}
    ~=~ \frac{ s^{\ddagger}_{\max} - s^{\ddagger}_{\min}}{N^{\ddagger}}.
\end{align}
\item[(ii)]
We use an unequally spaced partition in
the $b$-direction, denoted by $\{b_j\}$, where $j = 0,\ldots, J$,
with $b_0 = b_{\min}$,
$b_J = b_{\max}$,
$\Delta b_{\max} = \max_{0 \leq j \leq J-1} \left(b_{j+1} - b_j\right)$, and
$\Delta b_{\min} = \max_{0 \leq j \leq J-1} \left(b_{j+1} - b_j\right)$.
\end{itemize}
We emphasize that no timestepping is required for the interval $[t_{m-1}^+, t_m^-]$, {\myblue{$t_{m-1} \in \Tcal_M$.}}
As noted earlier, $\Delta t = T/M$ is kept constant.
We assume that there exists a discretization parameter $h>0$ such that
\EQA
\label{eq:dis_parameter}
\Delta s = C_1 h,
\quad
\Delta b_{\max} = C_2 h,
\quad
\Delta b_{\min} = C_2'h,
\ENA
where the positive constants $C_1$, $C_2$, $C_2'$ are independent of $h$.
For convenience, we occasionally use $\x_{n, j}^m \equiv (s_n, b_j, t_m)$ to refer to the reference gridpoint $(s_n, b_j, t_m)$,
$n = -N^{\dagger}/2, \ldots, N^{\dagger}/2$,
$j =  0, \ldots, J$, $m = M, \ldots, 0$.
{\myblue{Nodes ${\bf{x}}_{n, j}^m$ have
(i) $n = -N^{\dagger}/2, \ldots, -N/2$, in $\Omega_{s_{\min}}$,
(ii) $n = -N/2+1, \ldots N/2-1$, in $\Omega_{\myin}$,
(iii) $n = N/2, \ldots N^{\dagger}/2$, in  $\Omega_{s_{\max}}$.
and
(iv) $n = -N^{\ddagger}/2+1 \ldots -N^{\dagger}/2-1$ and
$n = N^{\dagger}/2+1 \ldots N^{\ddagger}/2-1$, in padding sub-domains.
}}

For $t_m\in \Tcal_M$, we denote by $v(s_n, b_j, t)$
the exact solution at the reference node $(s_n, b_j, t)$, where $t = \{t_m^{\pm}, t_m\}$,
and by $v_h(s, b, t)$ the approximate solution at an arbitrary point $(s, b, t)$
obtained using the discretization parameter $h$.
We refer to the approximate solution at the reference node $(s_n, b_j, t)$, where $t = \{t_m^{\pm}, t_m\}$,
as $v_{n, j}^{m\pm} \equiv v_h(s_n, b_j, t_m^{\pm})$ and  $v_{n, j}^{m} \equiv v_h(s_n, b_j, t_m)$.
In the event that we need to evaluate $v_h$ at a point other than a node on the computational gridpoint,
linear interpolation is used.
We define by  $\Zcal_h$ the discrete set of admissible impulse values defined as follows
\EQA
\label{eq:Z_h_dis}
\Zcal_h = \l\{b_0, b_1, \ldots, b_J\r\} \cap \Zcal.
\ENA
{\myblue{where $\mathcal{Z}$ is defined in \eqref{eq:Zcal_def}, and $h$ is the discretization parameter.}}
With $b^+\in \Zcal_h$ being an impulse value (a control),
applying $b^+$ at the reference spatial node $(s_n, b_j)$ results in
\EQ
\label{eq:sbplus}
s^+_n = s^+(s_n, b_j, b^+) \text{ computed by \eqref{eq:S+_B+}} , \quad b_j^+ = b^+(s_n, b_j, b^+) =  b^+.
\EN
%
For the special case $t_M$, as discussed earlier, we only have interest rate payment, but no rebalancing,
and therefore, only $v_{n, j}^{M}$ and $v_{n, j}^{M-}$ are used.

{\myblue{\subsection{Numerical schemes} }}
{\myblue{For convenience, we define $\mathbb{N} = \l\{-N/2+1, \ldots, N/2-1\r\}$, $\mathbb{N}^{\dagger} = \l\{-N^{\dagger}/2, \ldots, N^{\dagger}/2\r\}$ and $\mathbb{J} = \l\{0, \ldots, J\r\}$.
Backwardly, over the time interval $\l[t_{m-1}, t_{m}\r]$, $t_{m-1} \in \Tcal_M$, there are three key components solving the MV optimisation problem, namely
(i) the interest settlement over $\l[t_m^-, t_m\r]$ as given in \eqref{eq:recursive_interest*};
(ii) the time advancement from $t_{m}^-$ to $t_{m-1}^+$,
as captured by \eqref{eq:left_padding}-\eqref{eq:omega_in}, and
(iii) the intervention action over $[t_{m-1}, t_{m-1}^+]$
as given in \eqref{eq:recursive_b_comp_dom}. 
We now propose the numerical schemes for these steps.
}}

For $(s_n, b_j, t_M) \in \Omega \times \{T\}$, we impose the terminal
condition \eqref{eq:recursive_b_comp_initial*} by
\EQ
\label{eq:terminal}
v_{n, j}^{M} =  {\myblue{
\l(W_{\t{liq}}(s_n, b_j) - \f{\gamma}{2} \r)^2,
\quad
n \in \mathbb{N}^{\dagger},
~~
j \in \mathbb{J}.}}
\EN
{\myblue{
For imposing the intervention action \eqref{eq:recursive_b_comp_dom}, we solve the optimization problem
\EQA
\label{eq:Intervention_Operator_glwb}
{\myblue{ v_{n,j}^{m} = \min\l\{ v_{n,j}^{m+},
\min_{b^+ \in \mathcal{Z}_h} v_h(s_n^+, b^+, t_m^+)\r\},}}
\quad
s_n^+ = s^+(s_n, b_j, b^+), \quad
n \in \mathbb{N}^{\dagger}, ~
j \in \mathbb{J}.
\ENA
Here, $v_h(s^+_n, b^+, t_m^+)$ is the approximate solution
to the exact solution  $v(s_n^+, b^+, t_m^+)$, where $b^+ \in \mathcal{Z}_h$
and $s^+_n = s^+(s_n, b_j, b^+)$ is given by \eqref{eq:S+_B+}.
The approximation $v_h(s^+_n, b^+, t_m^+)$ is computed by linear interpolation as follows
\EQA
\label{eq:vtil}
    v_h(s_n^+, b^+, t_m^+) =
    \mathcal{I}\left\{v^{m+} \right\}
    \left(s^+_n, b^+\right),
    \quad
    n \in \mathbb{N}^{\dagger}, ~~
    j \in \mathbb{J}.
\ENA
Here, $\mathcal{I}\left\{v^{m+} \right\}(\cdot)$ is a linear interpolation operator acting on the time-$t_m^+$ discrete solutions
$\l\{s_q, b_p,  v_{q, p}^{m+}\r\}$, $q\in \mathbb{N}^{\dagger}$ and $p \in \mathbb{J}$.
We note that since $b^+ \in \l\{b_0, b_1, \ldots, b_J\r\}$,
\eqref{eq:vtil} boils down to a single dimensional interpolation along the $s$-dimension.
}}
\begin{remark}[Attainability of local minima]
\label{rm:sup_exist}
We determine the infimum of the intervention operator in \eqref{eq:recursive_b_comp} by a linear search over the discrete set of controls $\mathcal{Z}_h$ in \eqref{eq:Z_h_dis}, that is, an exhaustive search through all admissible controls. As mentioned in \cite{dang2014continuous}, using this approach,
we can guarantee obtain the global minimum as $h \rightarrow 0$.
\end{remark}

For the settlement of interest \eqref{eq:recursive_interest*}, linear interpolation/extrapolation is applied to compute $v_{n,j}^{m-}$ as follows.
\EQ
v_{n,j}^{m-} =  \mathcal{I}\left\{v^{m}_n \right\}
    \left(b_je^{R(b_j) \Delta t} \right), \quad
    n \in \mathbb{N}^{\dagger}, ~~
    j \in \mathbb{J}
\label{eq:recursive_interest**}
\EN
Here $\mathcal{I}\left\{v^{m}_n \right\}(\cdot)$ be
linear interpolation/extrapolation operator acting on the time-$t_m$ discrete solutions
{\myblue{
$\l\{b_q, v_{n, q}^{m}\r\}$, $q\in \mathbb{J}$, where $v_{n, q}^{m}$ are given by
\eqref{eq:terminal} at $t_m = T$ and by \eqref{eq:Intervention_Operator_glwb} at $t_m, m = M-1, \ldots,1$.
Note that when $\l(s_n, b\r) \in \Omega_{b_{\max}}$, $\mathcal{I}\left\{v^{m}_n \right\}(b)$ becomes
a linear extrapolation operator which imposes the boundary condition \eqref{eq:bmax}.  That is,
\EQA
\label{eq:v_h_bmax}
v(s_n, b, t_m) =  \l(\f{b}{b_{J}}\r)^2v \l(s_n, \sgn(b) b_{J},t_m\r), \qquad
(s_n, b, t_m) \in \Omega_{b_{\max}} \times \l\{t_m\r\}, \, m = M,\ldots,1.
\ENA

For the time advancement of $(s_n, b_j, t_{m-1}^+) \in \Omega_{s_{\min}} \cup \Omega_{s_{\myin}} \cup \Omega_{s_{\max}} \times \{t_{m-1}^+\}$, $t_{m-1} \in \Tcal_M$. The boundary conditions, for $\Omega_{s_{\min}} \cup \Omega_{s_{\max}} \times \{t_{m-1}^+\}$ as \eqref{eq:left_padding} and \eqref{eq:right_padding}, can be imposed by
\EQA
\label{eq:left_padding*}
v_{n,j}^{(m-1)+} &=& v_{n,j}^{m-},
\qquad\qquad\qquad\quad
n=-N^{\dagger}/2, \ldots, -N/2,
~~
j \in \mathbb{J},
{\text{ and $v_{n,j}^{m-}$ is given in \eqref{eq:recursive_interest**}, }}
\\
\label{eq:right_padding*}
v_{n,j}^{(m-1)+} &=& e^{\l(\sigma^2 + 2\mu + \lambda \kappa_2\r)\Delta t} v_{n,j}^{m-},
\quad
n= N/2, \ldots, N^{\dagger}/2,
~
j \in \mathbb{J},
{\text{ and $v_{n,j}^{m-}$ is given in \eqref{eq:recursive_interest**}}}.
\ENA
In  $\Omega_{\myin}$,}}
we tackle the convolution integral in \eqref{eq:omega_in},  where {\myblue{$j \in \mathbb{J}$}} is fixed.
For simplicity, we adopt the following notational convention: with
{\myblue{$n \in \mathbb{N}$ and
$l \in \mathbb{N}^{\dagger}$,}}
we let ${\myblue{g_{n-l}(\Delta t, \infty)}} = g(s_n - s_{l}; \Delta t, \infty)$, where $g(\cdot)$ is given by the infinite series \eqref{eq:g_series}.
We also denote by {\myblue{$g_{n-l}(\Delta t, K)$}} an approximation to {\myblue{$g_{n-l}(\Delta,\infty)$}} using the first $K$ terms
of the  infinite series \eqref{eq:g_series}.
Applying the composite trapezoidal rule to approximate the convolution integral \eqref{eq:omega_in} gives the approximation
in the form of a discrete convolution as follows
\EQA
\label{eq:scheme}
v_{n,j}^{(m-1)+} = \Delta s {\myblue{\sum_{l=-N^{\dagger}/2}^{N^{\dagger}/2} }} \omega_{l}\,
{\myblue{ g_{n - l}(\Delta t, K) }} \,
v_{l, j}^{m-}, \quad
{\myblue{
n \in \mathbb{N},~
j \in \mathbb{J}.
}}
\ENA
where $v_{l, j}^{m-}$ are given in \eqref{eq:recursive_interest**}
{\myblue{and $\omega_{l} = 1$, $l = -N^{\dagger}/2 +1, \ldots, N^{\dagger}/2-1$, and $\omega_{-N^{\dagger}/2} = \omega_{N^{\dagger}/2} = 1/2$.}}

\begin{remark}[Monotonicity]
\label{rm:mono}
We highlight that the conditional density {\myblue{$g_{n-l}(\Delta t, \i)$}} given by the infinite series \eqref{eq:g_series}
is defined and non-negative for all {\myblue{$n \in \mathbb{N}$ and $l \in \mathbb{N}^{\dagger}$}}
(or, alternatively, for all $s_n \in (s_{\min}, s_{\max})$ and $s_{l} \in [s_{\min}^{\dagger}, s_{\max}^{\dagger}]${\myblue{).}}
Therefore, scheme \eqref{eq:scheme} is monotone.

We highlight that for computational purposes, {\myblue{$g_{n-l}(\Delta t, \infty)$}}, given by the infinite series \eqref{eq:g_series},
is truncated to {\myblue{$g_{n-l}(\Delta t, K)$}}. However, since each term of the series is non-negative, this truncation does not result
in loss of monotonicity, which is a key advantage of the proposed approach.
\end{remark}
As $K \to \infty$, there is no loss of information in the discrete convolution \eqref{eq:scheme}.
For a finite $K$, however, there is an error, namely {\myblue{$\l| g_{n-l}(\Delta t, \infty) - g_{n-l}(\Delta t, K) \r|$}},
due to the use of a truncated Taylor series.
Specifically, this truncation error can be bounded as follows:
{\myblue{
\begin{align}
\l| g_{n-l}(\Delta t, \infty) - g_{n-l}(\Delta t; K) \r| &=
\l| \sum_{k=K+1}^{\i} \f{\l(\lambda\Delta t\r)^k}{k!}
\int_{-\infty}^{\infty}
e^{- \alpha \eta^2 +\l(\beta +s_n - s_{l}\r) \l(i \eta \r) + \theta} \,
\l(\Gamma\l( \eta \r) \r)^k ~ \md \eta \r|
\nonumber\\
&
{\myblue{
~{\buildrel (\text{i}) \over \le}~
\f{\l(\lambda\Delta t\r)^{K+1}}{(K+1)!} \,
g_{n-l}(\Delta t, \infty)
}}
~{\buildrel (\text{ii}) \over \le}~
\f{\l(\lambda\Delta t\r)^{K+1}}{(K+1)!}\,
\f{1}{\s{2 \pi \sigma^2 \Delta t}}.
\label{eq:Kbound}
\end{align}
Here, in (i), $\l| \l(\Gamma\l( \eta \r) \r)^{K+1} \r|
\le \l(\int_{-\infty}^{\infty} p(y)~\l|e^{i \eta y}\r|~dy\r)^{K+1}
\le 1$;
in (ii),  $g_{n-l}(\Delta t, \infty) \le
\f{e^{\theta}}{\s{4\pi \alpha}} \sum\limits_{k=0}^{\infty} \f{\l(\lambda \Delta t\r)^{k}}{k!}
    = \f{1}{\s{2 \pi \sigma^2 \Delta t}}$.  }}
Therefore, from \eqref{eq:Kbound}, as $K \rightarrow \i$, we have $\f{\l(\lambda\Delta t\r)^{K+1}}{(K+1)!} \rightarrow 0$,
resulting in no loss of information. For a given $\epsilon > 0$, we can choose $K$ such that
the error {\myblue{$\l| g_{n-l}(\Delta t, \infty) - g_{n-l}(\Delta t, K) \r| < \epsilon$}}, for all
{\myblue{$n \in \mathbb{N}$ and $l \in \mathbb{N}^{\dagger}$}}.
This can be achieved by enforcing
\EQA
\label{eq:K_Oh}
\f{\l(\lambda\Delta t\r)^{K+1}}{(K+1)!} \le \epsilon \s{2 \pi \sigma^2 \Delta t}.
\ENA
It is straightforward to see that, if $\epsilon = \Ocal(h)$, then $K = \Ocal(\ln(h^{-1}))$, as $h \rightarrow 0$.
For a given $\epsilon$, we denote by $K_{\epsilon}$ be the smallest $K$ values that satisfies \eqref{eq:K_Oh}.
We then have
\EQ
\label{eq:sum_g}
0< g_{n-l}(\Delta t, \infty) - g_{n-l}(\Delta t, K_{\epsilon}) < \epsilon,
\quad
{\myblue{n \in \mathbb{N},~ l \in \mathbb{N}^{\dagger}}}.
\EN
{\myblue{This value $K_{\epsilon}$ can be obtained through a simple iterative procedure, as illustrated in
Algorithm~\ref{alg:Gtilde}.
}}

\subsection{Efficient implementation and algorithms}
\label{sec:fft_alg}
In this section, we discuss an efficient implementation of the scheme presented above using FFT.
For convenience, we define/recall sets of indices: $\mathbb{N}^{\ddagger} = \left\{-N^{\ddagger}/2+1, \ldots, N^{\ddagger}/2-1\right\}$,
$\mathbb{N}^{\dagger} = \left\{-N^{\dagger}/2, \ldots, N^{\dagger}/2\right\}$,
$\mathbb{N} = \left\{-N/2+1, \ldots, N/2-1\right\}$,
$\mathbb{J} = \left\{0, \ldots, J \right\}$, with $N^{\dagger} = 2N$ and {\myblue{$N^{\ddagger} = N + N^{\dagger} = 3N$}}.
For brevity, we adopt the notational convention: for  $n \in \mathbb{N}$ and  $l \in \mathbb{N}^{\dagger}$,
{\myblue{$g_{n - l} \equiv g_{n - l}(\Delta t, K)$}}, where $K$ is chosen by \eqref{eq:K_Oh}.
To effectively compute the discrete convolution in \eqref{eq:scheme} for a fixed $j \in \mathbb{J}$, we rewrite \eqref{eq:scheme}
in a matrix-vector product form as follows
\EQA
\label{eq:scheme_matrix}
\underbrace{\begin{bmatrix}
    v_{-N/2+1,j}^{(m-1)+} \\
    v_{-N/2+2,j}^{(m-1)+} \\
    \vdots \\
    \vdots \\
    v_{N/2-1,j}^{(m-1)+}
\end{bmatrix}}_{~~~\scalebox{1}{$v_{j}^{(m-1)+}$}}
~~~
=
~~~
\Delta s~
\underbrace{\begin{bmatrix}
    g_{N/2+1} & g_{N/2}& \dots & g_{-3N/2+1} \\
    g_{N/2+2} & g_{N/2+1} & \dots & g_{-3N/2+2} \\
    \vdots & \vdots && \vdots \\
    \vdots & \vdots && \vdots \\
    g_{3N/2-1} & g_{3N/2-2} & \dots & g_{-N/2-1}
\end{bmatrix}}_{\quad\quad\scalebox{1.1}{$\l[g_{n - l}\r]_{n \in \mathbb{N}, \, l \in \mathbb{N}^{\dagger}}$}}
\,
~~~
\underbrace{\begin{bmatrix}
    \f{1}{2} v_{-N^{\dagger}/2, j}^{m-}  \\
    v_{-N^{\dagger}/2+1, j}^{m-}  \\
    \vdots \\
    v_{N^{\dagger}/2-1, j}^{m-} \\
    \f{1}{2} v_{N^{\dagger}/2, j}^{m-}
\end{bmatrix}}_{\quad \scalebox{1}{$v_{j}^{m-}$}}.
\ENA
Here, in \eqref{eq:scheme_matrix},  the vector $v_{j}^{(m-1)+} \equiv \l[v_{n,j}^{(m-1)+}\r]_{n \in \mathbb{N}}$ is of size $(N-1)\!\times\! 1$,
the matrix $\l[g_{n - l}\r]_{n \in \mathbb{N}, \, l \in \mathbb{N}^{\dagger}}$ is
of size $(N-1)\!\times\! (2N + 1)$, and the vector $v_{j}^{m-} \equiv \l[v_{n,j}^{m-}\r]_{n \in \mathbb{N}^{\dagger}}$
is of size $(2N + 1)\!\times\! 1$.
It is important to note that $\l[g_{n - l}\r]_{n \in \mathbb{N}, \, l \in \mathbb{N}^{\dagger}}$ is a Toeplitz matrix \cite{BrycDemboJiang_2006} having constant along diagonals. To compute the matrix-vector product in \eqref{eq:scheme_matrix} efficiently using FFT, we take advantage of
a cicular convolution product  described below.
\begin{itemize}
    \item We first expand the non-square
matrix $\l[g_{n - l}\r]_{n \in \mathbb{N}, \, l \in \mathbb{N}^{\dagger}}$ (of size $(N-1)\!\times\! (N^{\dagger} + 1)$)
into a circulant matrix of size $(3N-1)\!\times\! (3N-1)$ denoted by $\tilde{g}$, and is defined as follows
\EQA
\label{eq:cir_matrix}
\tilde{g} =
\l[\begin{array}{c|c}
\tilde{g}'_{-1,0} & \tilde{g}'_{-1,1}\\
\hline
\l[g_{n - l}\r]_{n \in \mathbb{N}, \, l \in \mathbb{N}^{\dagger}} & \tilde{g}'_{0,1}\\
\hline
\tilde{g}'_{1,0} & \tilde{g}'_{1,1}
\end{array}\r]. 
\ENA
Here, $\tilde{g}'_{-1,0}$, $\tilde{g}'_{1,0}$, $\tilde{g}'_{-1,1}$, $\tilde{g}'_{0,1}$ and $\tilde{g}'_{1,1}$ are matrices
of sizes $N\! \times\! (2N+1)$, $N\! \times\! (2N+1)$,
$N\! \times\!(N-2)$, $(N-1)\! \times\! (N-2)$, and $N\!\times\! (N-2)$, respectively,
and are given below
\EQA
\tilde{g}'_{-1,0} &=&
\l[\begin{array}{cccccccc}
    g_{-N/2+1} & g_{-N/2} & \dots & g_{-3N/2+1} & g_{3N/2-1} & g_{3N/2-2} & \dots & g_{N/2} \\
    g_{-N/2+2} & g_{-N/2+1} & \dots & g_{-3N/2+2} & g_{-3N/2+1} & g_{3N/2-1} & \dots & g_{N/2+1}\\
    \vdots & \vdots & & \vdots & \vdots & \vdots & & \vdots \\
    g_{N/2} & g_{N/2-1} & \dots & g_{-N/2} & g_{-N/2-1} & g_{-N/2-2} & \dots & g_{3N/2-1}
\end{array}\r], 
\nonumber\\
\tilde{g}'_{1,0} &=&
\l[\begin{array}{cccccccc}
    g_{-3N/2+1} & g_{3N/2-1} & g_{3N/2-2} & \dots & g_{N/2+1} & g_{N/2} & \dots & g_{-N/2} \\
    g_{-3N/2+2} & g_{-3N/2+1} & g_{3N/2-1}& \dots & g_{N/2+2} & g_{N/2+1} & \dots & g_{-N/2+1}\\
    \vdots & \vdots & \vdots & & \vdots & \vdots & & \vdots \\
    g_{-N/2} & g_{-N/2-1} & g_{-N/2-2} & \dots & g_{-3N/2+1} & g_{3N/2-1}  & \dots & g_{N/2-1}
\end{array}\r], 
\nonumber\\
\tilde{g}'_{-1,1} &=&
\l[\begin{array}{cccc}
    g_{N/2-1} & g_{N/2-2} & \dots & g_{-N/2+2} \\
    g_{N/2} & g_{N/2-1} & \dots & g_{-N/2+3} \\
    \vdots & \vdots & & \vdots \\
    g_{3N/2-2} & g_{3N/2-3} & \dots & g_{N/2+1}
\end{array}\r], 
\nonumber
\\
\tilde{g}'_{0,1} &=&
{\myblue{
\l[\begin{array}{ccccc}
    g_{3N/2-1} & g_{3N/2-2}& \dots & g_{N/2+3} & g_{N/2+2}\\
    g_{-3N/2+1} & g_{3N/2-1} & \dots & g_{N/2+4} & g_{N/2+3} \\
    \vdots & \vdots &  & \vdots & \vdots  \\
	g_{-N/2-2} & g_{-N/2-3} & \dots & g_{-3N/2+2} & g_{-3N/2+1}
\end{array}\r], 
}}
\nonumber\\
\tilde{g}'_{1,1} &=&
\l[\begin{array}{cccc}
    g_{-N/2-1} & g_{-N/2-2}  & \dots & g_{-3N/2+2} \\
    g_{-N/2} & g_{-N/2-1}  & \dots & g_{-3N/2+3}\\
    \vdots & \vdots & & \vdots \\
    g_{N/2-2} & g_{N/2-1} & \dots & g_{-N/2+1}
\end{array}\r]. 
\nonumber
\ENA

\item For  fixed $j \in \mathbb{J}$, we construct $\tilde{v}_{j}^{m-}$ a vector of size $(3N-1)\!\times\!1$ by
augmenting vector $v_{j}^{m-}$, defined in \eqref{eq:scheme_matrix}, with zeros as follows
\begin{align}
\label{eq:v_vec_toeplitz}
\tilde{v}_{j}^{m-} = \l[(v_{j}^{m-})^{\top}, \, 0,\, 0,\ldots, 0\r]^{\top}
=
\l[\f{1}{2}v_{-N^{\dagger}/2, j}^{m-}, \, v_{-N^{\dagger}/2+1, j}^{m-}, \ldots, v_{N^{\dagger}/2-1, j}^{m-}, \, \f{1}{2} v_{N^{\dagger}/2, j}^{m-}, \, 0,\, 0,\ldots, 0\r]^{\top}.
\end{align}
Then, \eqref{eq:scheme_matrix} can be implemented by applying a circulant matrix-vector product
to compute an intermediate vector of discrete solutions $\tilde{v}_{j}^{(m-1)+}$ as follows
\EQA
\label{eq:scheme_cir_matrix}
\tilde{v}_{j}^{(m-1)+}  = \Delta s\, \tilde{g} \, \tilde{v}_{j}^{m-},
\qquad
j \in \mathbb{J}.
\ENA
Here, the circulant matrix $\tilde{g}$ is given by \eqref{eq:cir_matrix}, and the vector $\tilde{v}_{j}^{m-}$ is given by  \eqref{eq:v_vec_toeplitz},
and the intermediate result $\tilde{v}_{j}^{(m-1)+}$ is a vector of size $(3N-1)\!\times\! 1$,
{\myblue{
with $v_{j}^{(m-1)+}$ is the middle $2N-1$ (from the $(N+1)$-th to the $(2N-1)$-th) elements of $\tilde{v}_{j}^{(m-1)+}$.
}}

\item Observing that a circulant matrix-vector product is equal to a circular convolution product, \eqref{eq:scheme_cir_matrix} can further be written {\myblue{as a}} circular convolution product. More specifically, let $\tilde{g}_1$ be the first column of the circulant matrix $\tilde{g}$ defined in \eqref{eq:cir_matrix}, and is  given by
\EQA
\label{eq:vec_toeplitz}
\tilde{g}_1 &=& \big[g_{-N/2+1}, \, g_{-N/2+2}, \ldots, g_{3N/2-1},  \, g_{-3N/2+1}, \, g_{-3N/2+2},\ldots, g_{-N/2}\big]^{\top}.
\ENA
The circular convolution product $z = x\ast y$ is defined componentwise by
\EQA
z_{k'} = \sum_{k=-N^{\ddagger}/2+1}^{N^{\ddagger}/2-1} x_{k'-k+1} \, y_k,
\quad
k' = -N^{\ddagger}/2+1,\ldots,N^{\ddagger}/2-1,
\nonumber
\ENA
where $x$ and $y$ are two sequences with the period $(N^{\ddagger}-1)$ (i.e.\ $x_k = x_{k+(N^{\ddagger}-1)}$ and $y_k = y_{k+(N^{\ddagger}-1)}$,
{\myblue{$k' \in \mathbb{N}^{\ddagger}$)}}.
Then, \eqref{eq:scheme_cir_matrix} can be written as the following circular convolution product
\EQA
\label{eq:scheme_cir_conv}
\tilde{v}_{j}^{(m-1)+} =
\Delta s \, \tilde{g} \,
\tilde{v}_{j}^{m-}
=
\Delta s \, \tilde{g}_1 * \tilde{v}_{j}^{m-},
\qquad
j =  0, \ldots, J.
\ENA

\item The circular convolution product in \eqref{eq:scheme_cir_conv} can be computed efficiently using
FFT and iFFT as follows
\EQA
\label{eq:scheme_FFT}
\tilde{v}_{j}^{(m-1)+} = \Delta s \, \text{FFT}^{-1} \l\{ \text{FFT}(\tilde{v}_{j}^{m-}) \circ \text{FFT}(\tilde{g}_1) \r\},
\qquad
j =  0, \ldots, J.
\ENA

\item Once the vector of intermediate discrete solutions $\tilde{v}_{j}^{(m-1)+} \equiv \tilde{v}_{n,j}^{(m-1)+}$ is computed,
we then obtain the vector of discrete solutions $\l[v_{n,j}^{(m-1)+}\r]_{n \in \mathbb{N}}$ (of size $(2N+1)\!\times\! 1$) for $\Omega_{\myin}$
by discarding values  $\tilde{v}_{n,j}^{(m-1)+}$, $n \in \mathbb{N}^{\ddagger} \setminus \mathbb{N}$.
\end{itemize}

The implementation \eqref{eq:scheme_FFT} suggests that
{\myblue{we compute the weight components of $\tilde{g}_1$ only once, and}}
reuse them for the computation over all time intervals. More specifically, for a given user-tolerance $\epsilon$,
using \eqref{eq:K_Oh}, we can compute a sufficiently large the number of terms $K = K_{\epsilon}$ in the infinite series representation \eqref{eq:g_series} for these weights. Then, using Corollary~\ref{cor:twodis}, these weights for the  case $\xi$
following a normal distribution \cite{MertonJumps1976} or a double-exponential distribution \cite{KouOriginal}
can be computed only once in the Fourier space, as in \eqref{eq:scheme_FFT}, and reused for all time intervals.
The step is described in Algorithm~\ref{alg:Gtilde}.

\begin{algorithm}[htb!]
\caption{
\label{alg:Gtilde}
Computation of weight vector $\tilde{g}_1(\Delta t, K_{\epsilon})$ in the Fourier space; $\epsilon>0$ is an user-defined tolerance.
}
\begin{algorithmic}[1]
%
\STATE
set $k = K_{\epsilon} = 0$;

\STATE
compute
$\text{test} = \f{\l(\lambda\Delta t\r)^{k+1}}{(k+1)! \, \s{2 \pi \sigma^2 \Delta t}}$;
\\
compute
$g_{n - l}(\Delta t, K_{\epsilon}) = g(s_n - s_l; \Delta t, 0)$,  $n \in \mathbb{N}, \, l \in \mathbb{N}^{\dagger}$, given in Corollary~\ref{cor:twodis};

\STATE
construct
the weight vector~$\tilde{g}_1(\Delta t, K_{\epsilon})$ using $g_{n - l}(\Delta t, K_{\epsilon})$ as defined in \eqref{eq:vec_toeplitz};

\WHILE{{\red{$\text{test}  \ge \epsilon$} }}
	\label{alg:while}
	\STATE
    set $k = k+1$, and $K_{\epsilon} = k$;

    \STATE
	{\myblue{
	compute
	$\text{test} = \f{\l(\lambda\Delta t\r)^{k+1}}{(k+1)! \, \s{2 \pi \sigma^2 \Delta t}}$;
	}}
	
    \STATE
     compute
    the increments~$\Delta g_k(s_n - s_l; \Delta t)$,  $n \in \mathbb{N}, \, l \in \mathbb{N}^{\dagger}$, given in Corollary~\ref{cor:twodis};

    \STATE
     compute
	$g_{n - l}(\Delta t, K_{\epsilon}) =  g_{n - l}(\Delta t, K_{\epsilon}) + \Delta g_k(s_n - s_l; \Delta t)$,  $n \in \mathbb{N}, \, l \in \mathbb{N}^{\dagger}$;
	
	\STATE
	construct
	the weight vector~$\tilde{g}_1(\Delta t, K_{\epsilon})$ using $g_{n - l}(\Delta t, K_{\epsilon})$ as defined in \eqref{eq:vec_toeplitz};
	
\ENDWHILE

\STATE
output weight vector $\text{FFT}(\tilde{g}_1)$;
\end{algorithmic}
\end{algorithm}

%
%
%
%
%
%

Putting everything together, in Algorithm~\ref{alg:monotone}, we present a monotone integration algorithm for MV portfolio optimization.
\begin{algorithm}[htb!]
\caption{
\label{alg:monotone}
A monotone numerical integration algorithm for MV portfolio optimization when
$\xi$ follows a normal distribution \cite{MertonJumps1976} or a double-exponential distribution \cite{KouOriginal};
$\epsilon>0$ is a user-tolerance; the embedding parameter $\gamma\in \mathbb{R}$ is a fixed;
}
\begin{algorithmic}[1]

\STATE
compute weight vector $\tilde{g}_1$ using Algorithm~\ref{alg:Gtilde};
\STATE
\label{alg:initial}
initialize $v_{n, j}^{M} = {\myblue{ \l(W_{\t{liq}}(s_n, b_j) - \f{\gamma}{2} \r)^2}}$,
$n= -N^{\dagger}/2, \ldots, N^{\dagger}/2$, $j = 0, \ldots, J$;
\FOR{$m = M, \ldots, 1$}
	\STATE enforce interest rate payment \eqref{eq:recursive_interest**}
to obtain  $v_{n,j}^{m-}$, $n =  -N^{\dagger}/2, \ldots, N^{\dagger}/2$,
                $j =  0, \ldots, J$;

     \STATE
    \label{alg:step2}
    compute vectors of intermediate values $\tilde{v}_{j}^{(m-1)+}$, $j =  0, \ldots, J$ using \eqref{eq:scheme_FFT};
%
%
    \STATE
    \label{alg:step3}
    obtain vectors of discrete solutions $\l[v_{n,j}^{(m-1)+}\r]_{n \in \mathbb{N}}$, $j =  0, \ldots, J$ by discarding
    all values $\tilde{v}_{n,j}^{(m-1)+}$ (Line \eqref{alg:step2}) where $n \in \mathbb{N}^{\ddagger} \setminus \mathbb{N}$ ;
    \hfill $\Omega_{\myin}$

	
	\STATE
    compute
    \label{alg:step4}
    $v_{n,j}^{(m-1)+}$,
    $n = -N^{\dagger}/2, \ldots, -N/2$, $j =  0, \ldots, J$,
    using \eqref{eq:left_padding*};
     \hfill $\Omega_{s_{\min}}$

     \STATE
    compute
    \label{alg:step5}
	$v_{n,j}^{(m-1)+}$,
    $n = N/2, \ldots, N^{\dagger}/2$, $j =  0, \ldots, J$
	using \eqref{eq:right_padding*};
     \hfill $\Omega_{s_{\max}}$

	{\myblue{
	\STATE
        \label{alg:opti}
        solve the optimization problem \eqref{eq:Intervention_Operator_glwb}
        to obtain  $v_{n,j}^{m-1}$, $n \in \mathbb{N}^{\dagger}$, $j \in \mathbb{J}$;
        \\
        save the optimal impulse value   $c_{n,j}^{m, *}$;
	}}

	

\ENDFOR

\end{algorithmic}
\end{algorithm}

\begin{remark}[Complexity]
\label{rm:complexity}
Algorithm~\ref{alg:monotone} involves, for $m = M\ldots,1$, the key steps as follows.
\begin{itemize}
\item Compute $v_{n,j}^{(m-1)+}$, $n \in \mathbb{N}^{\ddagger}$, $j \in \mathbb{J}$ via
FFT algorithm.
The complexity of this step is
$\Ocal\l(J N^{\ddagger} \log_2 N^{\ddagger}\r) = \Ocal\l( 1/h^2 \cdot \log_2(1/h) \r)$,
where we take into account \eqref{eq:dis_parameter}.

\item We use exhaustive search through all admissible controls in $\Zcal_h$ to obtain global minimum.
Each optimization problem is solved by evaluating the objective function $\Ocal(1/h)$ times. There are $\Ocal(1/h^2)$ nodes, and $\Ocal(1)$ timesteps giving a total complexity $\Ocal(1/h^3)$. This is an order reduction compared to complexity of finite difference methods, which typically is  $\Ocal(1/h^4)$ for discrete rebalancing (see \cite{DangForsyth2014}[Section~6.1].)

\end{itemize}
\end{remark}

\subsection{Construction of efficient frontier}
We know discuss construction of efficient frontier.
To this end, we define the auxiliary function $u(s,b,t_m) = E_{\Ccal_m^*}^{x, \tn} \l[W_T\r]$, where
$\Ccal_m^*$, as defined in \eqref{eq:optimal_ctrl}, is the optimal control
for the problem $PCMV_{\Delta t} (\tn;\gamma)$ obtained by solving the localized problem
in Definition~\ref{def:glwb}. Similar to \cite{DangForsyth2014, PvSDangForsyth2018_MQV, PvSDangForsyth2018_TCMV},
we now present a localized problem for
{\myblue{$u(x^m) = u(s,b,t_m)$, with $x^m = (s,b,t_m)$ and $t_m \in \Tcal_M \cup \l\{T\r\}$,}}
in the sub-domains \eqref{eq:sub_domain_whole} as below
{\myblue{
\begin{linenomath}
\begin{subequations}
\label{eq:linear}
\begin{empheq}[left={\empheqlbrace}]{alignat=6}
&u \l(x^M\r) ~&=&~ W_{\t{liq}}(s, b)-\varepsilon,
 \qquad \qquad \qquad \qquad \qquad~ x^M \in \Omega \times \l\{T\r\},
\label{eq:linear_T}
\\
&u(x^m) ~&=&~  \mathcal{M}(c_m^*) \, u\l(x^{m+}\r),  \qquad \qquad \qquad \qquad \quad x^m \in \Omega \times \Tcal_M ,
\label{eq:linear_c}
\\
&u(x^{m-}) ~&=&~ u \l(s,be^{R(b) \Delta t},\tn\r), \qquad \qquad \qquad\qquad x^m \in \Omega \times \l\{t_m\r\}, \quad m = M, \ldots, 1,
\label{eq:linear_i}
\\
&u(x^m) ~&=&~
{\red{
\f{\l|b\r|}{b_{\max}}u \l(s, \sgn(b) b_{\max},t_m\r),
\qquad \qquad\quad x^m \in \Omega_{b_{\max}} \times \l\{t_m\r\}, \, m = M,\ldots,1,
}}
\label{eq:linear_bmax}
\\
&u(x^{(m-1)+}) ~&=&~  \l\{
\begin{small}
\begin{aligned}
&\int_{s_{\min}^{\dagger}}^{s_{\max}^{\dagger}}
u \l(s',b,\tn^-\r) g(s- s'; \Delta t)~ \md s',
&x^{(m-1)+} \in \Omega_{\myin} \times \l\{t_{m-1}^+\r\},
\, t_{m-1} \in \Tcal_M,
\\
&u(x^{m-}) \, e^{\mu \Delta t},
&x^{(m-1)+} \in \Omega_{s_{\max}} \times \l\{t_{m-1}^+\r\},
\, t_{m-1} \in \Tcal_M,
\\
&u(x^{m-}),
&x^{(m-1)+} \in \Omega_{s_{\min}} \times \l\{t_{m-1}^+\r\},
\, t_{m-1} \in \Tcal_M.
\end{aligned}
\end{small}
\r\}
\label{eq:linear_all}
\end{empheq}
\end{subequations}
\end{linenomath}
}}
Here, in \eqref{eq:linear_c}, $c_m^*$ is the optimal impulse value obtained from solving the value function problem
\eqref{eq:recursive_b_comp_dom}; \eqref{eq:linear_i} is due to the settlement (payment or receipt)
of interest due for the time interval $[t_{m-1}, t_{m}]$, $m = M, \ldots, 1$;
{\myblue{\eqref{eq:linear_bmax}-\eqref{eq:linear_all} are equations for spatial sub-domains $\Omega_{b_{\max}}$, $\Omega_{\myin}$, $\Omega_{s_{\max}}$ and $\Omega_{s_{\min}}$.}}
The localized problem \eqref{eq:linear} can be solved numerically in a straightforward manner.
In  particular, at a reference gridpoint $(s_n, b_j)$,  the optimal impulse value $c_m^*$ in \eqref{eq:linear_c}
becomes $c_{n, j}^{m,*}$ which is the optimal impulse value obtained from Line~\eqref{alg:opti} of Algorithm~\ref{alg:monotone}.
We emphasize the convention that it may be non-optimal to rebalance, in which case,
the convention is $c_{n, j}^{m,*} = b_j$.
Furthermore, the convolution integral in \eqref{eq:linear_all} can be approximated
using a scheme similar to \eqref{eq:scheme}.  For brevity,
{\myblue{
we only provide the proof of numerical scheme for $\Omega_{s_{\max}}$ in Appendix~\ref{sec:v_u_s_max}, and omit details of the other schemes for \eqref{eq:linear}.
}}

We assume that
{\myblue{given the initial state $x = (s,b)$ at time $t_0$ and}}
the positive discretization parameter $h$, the efficient frontier (EF), denote by $\mathcal{Y}_h$,
can be traced out using the embedding parameter $\gamma \in \mathbb{R}$ as below
\EQ
\label{eq:y}
\mathcal{Y}_h = \bigcup\limits_{\gamma \in \mathbb{R}} \l(\s{\l(Var_{\Ccal_0^*}^{x,t_0} \l[W_T\r]\r)_h} , \, \l(E_{\Ccal_0^*}^{x, t_0} \l[W_T\r]\r)_h\r)_{\gamma}.
\EN
Here, $(\cdot)_h$ refers to a discretization approximation to the expression in the brackets. Specifically, for fixed $\gamma$, we let
\vspace{-2ex}
\EQA
V_0 \equiv  {\myblue{v(s,b,t_0)}} = E_{\Ccal_0^*}^{x, t_0} \l[\l(W_T-\f{\gamma}{2}\r)^2\r]
\quad {\text{and}}
\quad
U_0 \equiv  {\myblue{u(s,b,t_0)}} = E_{\Ccal_0^*}^{x, t_0} \l[W_T\r].
\ENA
Then $\l(Var_{\Ccal_0^*}^{x,t_0} \l[W_T\r]\r)_h$ and $\l(E_{\Ccal_0^*}^{x, t_0} \l[W_T\r]\r)_h$ corresponding to $\gamma$ in \eqref{eq:y} are computed as follows
\EQA
\l(Var_{\Ccal_0^*}^{x,t_0} \l[W_T\r]\r)_h = V_0 + \gamma U_0 - \f{\gamma^2}{4} - \l(U_0\r)^2
\quad {\text{and}}
\quad
\l(E_{\Ccal_0^*}^{x, t_0} \l[W_T\r]\r)_h = U_0.
\ENA

\section{Pointwise convergence}
\label{sc:conv}
In this section, we establish pointwise convergence of the proposed numerical integration method.
We start by verifying three properties: $\ell_\infty$-stability, monotonicity,
and consistency (with respect to the integral formulation \eqref{eq:omega_in}).
We recall that the infinite series $g_{n-l}(\Delta t, \infty)$ is approximated by
$g_{n-l}(\Delta t, K_{\epsilon})$, where $\epsilon>0$ is an user-defined tolerance,
and we have the error bound $g_{n-l}(\Delta t, \infty) - g_{n-l}(\Delta t, K_{\epsilon}) < \epsilon$,
as noted in \eqref{eq:sum_g}.

It is straightforward to see that the {\myblue{proposed}} scheme is monotone since all the weights
$g_{n-l}$ are positive. Therefore, we will primarily focus on $\ell_\infty$-stability and consistency of the scheme.
We will then show that convergence of our scheme is ensured if $K_{\epsilon} \to \infty$ as $h \to 0$,
or equivalently, $\epsilon \to 0$ as $h \to 0$.

For subsequent use,  we present a remark about {\myblue{$g_{n-l}(\Delta t; K_{\epsilon})$, $n \in \mathbb{N}$, $l \in \mathbb{N}^{\dagger}$}}.
\begin{remark}
\label{rm:sum_p}
{\myblue{Recalling that}} $g(s, s'; \Delta t) \equiv g(s, s'; \Delta t, \infty)$ is a (conditional) probability density function,
for a fixed $s_n \in [s_{\min}, s_{\max}]$, we have
$\ds
\int_{\mathbb{R}} g(s_n, s; \Delta  t, \infty )~ds = 1$,
hence
$\int_{s_{\min}^{\dagger}}^{s_{\max}^{\dagger}} g(s_n, s; \Delta  t, \infty)~ds \le 1$.
{\myblue{Furthermore}}, applying quadrature rule to approximate
$\ds {\myblue{ \int_{s_{\min}^{\dagger}}^{s_{\max}^{\dagger}} }} g(s_n, s; \Delta  t , \infty)~ds$
gives rise to an approximation error, denoted by $\epsilon_g$, defined as follows
\[
\epsilon_g := \bigg| {\red{\Delta s \sum_{l = -N^{\dagger}/2}^{N^{\dagger}/2} \omega_l }}~g_{n-l}(\Delta  t , \infty ) - \int
_{s_{\min}^{\dagger}}^{s_{\max}^{\dagger}} g(s_n, s; \Delta  t, \infty )~ds \bigg|.
\]
It is straightforward to see that $ \epsilon_g \to 0$ as $N^{\dagger} \to \infty$, i.e.\ as $h \to 0$.
Using the above results, recalling the weights $\omega_l$, {\myblue{$l \in \mathbb{N}^{\dagger}$}}, are positive,
and the error bound \eqref{eq:sum_g},
{\myblue{we have}}
\EQA
\label{eq:qua_error}
0 ~\leq~ \Delta s {\myblue{\sum_{l = -N^{\dagger}/2}^{N^{\dagger}/2}}} \omega_l~g_{n-l}(\Delta  t, K_{\epsilon})
~\leq~ \Delta s {\myblue{\sum_{l = -N^{\dagger}/2}^{N^{\dagger}/2}}} \omega_l~g_{n-l}(\Delta  t,\infty)
~\leq~ 1 + \epsilon_g
< e^{\epsilon_g}.
\ENA
\end{remark}
%

\subsection{Stability}
Our scheme consists of the following equations: \eqref{eq:terminal} for $\Omega\times \{T\}$, \eqref{eq:left_padding*} for $\Omega_{s_{\min}}$,
\eqref{eq:right_padding*} for $\Omega_{s_{\max}}$,
and finally \eqref{eq:scheme} for $\Omega_{\myin}$.
We start by verifying $\ell_\infty$-stability of our scheme.
\begin{lemma}[$\ell_\infty$-stability]
\label{lemma:stability}
Suppose the discretization parameter $h$ satisfies \eqref{eq:dis_parameter}.
If linear interpolation is used for the intervention action \eqref{eq:Intervention_Operator_glwb},
then the scheme \eqref{eq:terminal}, \eqref{eq:left_padding*}, \eqref{eq:right_padding*}, and \eqref{eq:scheme} satisfies
the bound
$\ds \sup_{h > 0} \left\| v^{m} \right\|_{\infty} < \infty$
for all $m = M, \ldots, 0$, as the discretization parameter $h \to 0$.
Here, we have $\left\| v^{m} \right\|_{\infty} = \max_{n, j} |v_{n, j}^{m}|$,
{\myblue{$n \in \mathbb{N}^{\dagger}$ and
$j \in \mathbb{J}$}}.
\end{lemma}

\begin{proof}[Proof of Lemma~\ref{lemma:stability}]
First, we note that, for any fixed $h >0$, as given by \eqref{eq:terminal},
and for a finite $\gamma$, we have $\left\| v^{M} \right\|_{\infty} < \infty$, since $\Omega$ is a bounded domain.
Therefore, we have $\sup_{h > 0} \left\| v^{M} \right\|_{\infty} < \infty$.
Motivated by this observation, to demonstrate $\ell_\infty$-stability of our scheme, we will show that, for a fixed $h > 0$, at any $(s_n, b_j,
t_m)$, $m = M, \ldots, 0$, we have
\EQ
\label{eq:key}
|v_{n, j}^{m}| < e^{(M-m)\l(\epsilon_g + \l( {\red{2r_{\max}}} + \sigma^2 + 2\mu + \lambda \kappa_2\r)\Delta t\r)}\left\| v^{M} \right\|_{\infty}{\myblue{,}}
\EN
where (i) $\epsilon_g$ is the error of the quadrature rule discussed in Remark~\ref{rm:sum_p},
(ii) $r_{\max} = \max\l\{r_b, r_{\iota}\r\}$, and
(iii) $\kappa_2 = \Ebb \l[\l(e^{\xi}-1\r)^2\r]$. In \eqref{eq:key}, the term $e^{(M-m) {\red{2r_{\max}}}\Delta t}$
is a result of the evaluation of $v_{n,j}^{m-}$ using
\eqref{eq:v_h_bmax} for nodes near $\pm b_{\max}$.
For the rest of the proof, we will show the key inequality \eqref{eq:key}
when $h > 0$ is fixed. The proof follows from a straightforward maximum analysis,
since $\Omega$ is a bounded domain. For brevity, we outline only key steps
of an induction proof below.

We use induction $m$, $m = M-1, \ldots, 0$, to show the bound \eqref{eq:key}
for $\Omega_{s_{\min}}\cup \Omega_{\myin}\cup \Omega_{s_{\max}}$.
For the base case, $m = M-1$ and thus \eqref{eq:key} becomes
\begin{align}
\label{eq:base_case}
|v_{n, j}^{M-1}| < e^{\epsilon_g  + ({\red{2r_{\max}}} + \sigma^2 + 2\mu + \lambda \kappa_2)\Delta t}\left\| v^{M} \right\|_{\infty},
\quad
{\myblue{ n \in \mathbb{N}^{\dagger} \text{ and }
j \in \mathbb{J} }}
.
\end{align}
For the settlement of interest rate for all  $\Omega_{s_{\min}}\cup \Omega_{\myin}\cup \Omega_{s_{\max}}$,  as reflected by \eqref{eq:recursive_interest**}, we have $|v_{n, j}^{M-}|<  e^{ {\red{2r_{\max}}} \Delta t}|v_{n, j}^{M}|$,
{\myblue{$n \in \mathbb{N}^{\dagger}$ and $j \in \mathbb{J}$}}.
Since $|v_{n, j}^{M}| \le \left\| v^{M} \right\|_{\infty}$,
it follows that
\EQ
\label{eq:itner}
|v_{n, j}^{M-}|<  e^{{\red{2r_{\max}}} \Delta t}\left\| v^{M}\right\|_{\infty}.
\EN
We now turn to $\ell_\infty$-stability of \eqref{eq:right_padding*} (for $\Omega_{s_{\max}}$).
From \eqref{eq:right_padding*}, we note that for
{\myblue{$n \in \big\{N/2, \ldots, N^{\dagger}/2\big\}$}} and {\myblue{$j \in \mathbb{J}$}},
\EQ
\label{eq:right_padding**}
|v_{n, j}^{(M-1)+}| = e^{\Delta t(\sigma^2 + 2\mu + \lambda \kappa_2)} |v_{n, j}^{M-}|
\overset{\eqref{eq:itner}}{\le} e^{\Delta t( {\red{2r_{\max}}} + \sigma^2 + 2\mu + \lambda \kappa_2)} \left\| v^{M}\right\|_{\infty}
\le e^{\epsilon_g  + \Delta t( {\red{2r_{\max}}} + \sigma^2 + 2\mu + \lambda \kappa_2)} \left\| v^{M}\right\|_{\infty}
\EN
noting $e^{\epsilon_g} \ge 1$.
Using \eqref{eq:itner}, it is trivial that \eqref{eq:left_padding*} (for $\Omega_{s_{\min}}$) satisfies
\EQ
\label{eq:left_padding**}
|v_{n, j}^{(M-1)+}| \le e^{\epsilon_g + \Delta t({\red{2r_{\max}}} + \sigma^2 + 2\mu + \lambda \kappa_2)} \left\| v^{M}\right\|_{\infty},
n \in  \big\{-N^{\dagger}/2, \ldots, -N/2\big\}, ~ {\myblue{j \in \mathbb{J}}}.
\EN
Now, we focus on the timestepping scheme \eqref{eq:scheme} (for $\Omega_{\myin}$).
For {\myblue{$n \in \mathbb{N}$ and $j \in \mathbb{J}$}},
we have
\begin{linenomath}
\begin{align}
|v_{n,j}^{(M-1)+}| \le \Delta s
{\myblue{\sum_{l=-N^{\dagger}/2}^{N^{\dagger}/2}}}
\omega_{l}\, g_{n - l}(\Delta t, K_{\epsilon}) \,
|v_{l, j}^{M-}|
&\overset{\eqref{eq:itner}}{\le}
e^{{\red{2r_{\max}}} \Delta t }\left\| v^{M}\right\|_{\infty}\bigg(
\Delta s
{\myblue{\sum_{l=-N^{\dagger}/2}^{N^{\dagger}/2}}}
\omega_{l}\, g_{n - l}(\Delta t, K_{\epsilon})\bigg) \,
\\
&\overset{\text{(i)}}{\le}
e^{\epsilon_g + \l({\red{2r_{\max}}}  +  \sigma^2 + 2\mu + \lambda \kappa_2\r)\Delta t }\left\| v^{M}\right\|_{\infty}.
\label{eq:vl}
\end{align}
\end{linenomath}
Here, (i) is due to \eqref{eq:qua_error}
and \eqref{eq:right_padding**} and \eqref{eq:left_padding**}.

Finally, given  \eqref{eq:vl}, we bound the intervention result $|v_{n, j}^{(M-1)}|$,
{\myblue{$n \in \mathbb{N}^{\dagger}$  and  $j \in \mathbb{J}$}},
obtained from  \eqref{eq:Intervention_Operator_glwb}.
Since linear interpolation is used, the weights for interpolation are non-negative.
In addition, due to \eqref{eq:right_padding**}, \eqref{eq:left_padding**}, and \eqref{eq:vl},
the numerical solutions at nodes used for interpolation, namely $|v_{l, j}^{(M-1)+}|$,
{\myblue{$l \in  \mathbb{N}^{\dagger}$}}, are also bounded by
\[
|v_{l, j}^{(M-1)+}| \le e^{\epsilon_g + \l({\red{2r_{\max}}}  +  \sigma^2 + 2\mu + \lambda \kappa_2\r)\Delta t }\left\| v^{M}\right\|_{\infty}.
\]
Therefore, by monotonicity of linear interpolation, which is preserved by the $\sup(\cdot)$ operator in \eqref{eq:Intervention_Operator_glwb},
$|v_{n, j}^{(M-1)}|$, {\myblue{$n \in \mathbb{N}^{\dagger}$ and $j \in \mathbb{J}$}},
satisfy \eqref{eq:base_case}. We have proved the base case \eqref{eq:base_case}.
Similar arguments can be used to show the induction step. This concludes the proof.
\end{proof}

\subsection{Consistency}
In this subsection, we mathematically demonstrate the pointwise consistency  of
the proposed scheme with respect to the MV optimization in Definition~\ref{def:glwb}.
Since it is straightforward that \eqref{eq:terminal} is consistent with the terminal condition \eqref{eq:recursive_b_comp_initial*}
($\Omega\times \{T\}$), we primarily focus  on the consistency of the scheme on $\Omega\times \{t_{m-1}\}$, $m = M, \ldots, 1$.

We start by introducing notational convention. We use $\x = (s, b) \in \Omega$ and
$\x^m \equiv (s, b, t_m) \in \Omega \times \{t_m\}$, $m = M, \ldots, 0$.  In addition, for brevity,
we use $v^m(\x)$ instead of $v(s, b, t_m)$, $m = M, \ldots, 0$.
We now write the MV portfolio optimization in Definition~\ref{def:glwb}
and the proposed scheme in forms amendable for analysis.
Recalling $s^+(s, b, c)$ defined in \eqref{eq:S+_B+},
over each time interval $[t_{m-1}, t_m]$, where $m = M, \ldots, 1$, we write
the MV portfolio optimization in Definition~\ref{def:glwb} via an operator $\mathcal{D}(\cdot)$ as follows
\begin{align}
\label{eq:con_mean}
%
 v^{m-1}(s, b)  = \mathcal{D}\big({\bf{x}}^{m-1}, v^{m}\big) &\coloneqq
  {\myblue{ \min\left\{v(s, b, t_{m-1}^+), \,
 \inf_{c \in \mathcal{Z}} \mathcal{M}(c)~v(s, b ,t_{m-1}^+)\r\}}}
 \nonumber
 \\
& =
 {\myblue{ \min\left\{v(s, b, t_{m-1}^+), \,
 \inf_{c \in \mathcal{Z}} v(s^+(s, b, c), c, t_{m-1}^+)\right\}}}.
\end{align}
Here,  $v(s^+(s, b, c), c, t_{m-1}^+)$ is given by
\begin{linenomath}
\begin{subequations}\label{eq:scheme_v_all*}
\begin{empheq}[left={
v(s^+(s, b, c), c, t_{m-1}^+)
= \empheqlbrace}]{alignat=6}
&
v \l(s^+(s, be^{R(b) \Delta t}, c),c,\tn\r)
&&
&& \quad\quad (s, b) \in \Omega_{s_{\min}} ,
\label{eq:scheme_v_all_1*}
\\
&
\ds\int_{s_{\min}^{\dagger}}^{s_{\max}^{\dagger}} v(s^+(s', be^{R(b) \Delta t}, c), c, t_m)\, g(s - s';\Delta  t )~ds'
&&
&&\quad\quad (s, b) \in  \Omega_{\myin},
\label{eq:scheme_v*}
\\
&
e^{\l(\sigma^2 + 2\mu + \lambda \kappa_2\r)\Delta t} v \l(s^+(s, be^{R(b) \Delta t}, c),c,\tn\r)
&&
&&\quad\quad (s, b) \in \Omega_{s_{\max}}.
\label{eq:scheme_v_all_3*}
\end{empheq}
\end{subequations}
\end{linenomath}
Similarly, the term $v(s, b, t_{m-1}^+)$  in \eqref{eq:con_mean} is defined as follows
\begin{linenomath}
\begin{subequations}
\begin{empheq}[left={
v(s, b, t_{m-1}^+)
= \empheqlbrace}]{alignat=6}
&
v \l(s, be^{R(b) \Delta t}, \tn\r)
&&
&& \quad\quad (s, b) \in \Omega_{s_{\min}} ,
\nonumber
\\
&
\ds\int_{s_{\min}^{\dagger}}^{s_{\max}^{\dagger}} v(s',  be^{R(b) \Delta t}, t_m)\, g(s - s';\Delta  t )~ds'
&&
&&\quad\quad (s, b) \in  \Omega_{\myin},
\nonumber
\\
&
e^{\l(\sigma^2 + 2\mu + \lambda \kappa_2\r)\Delta t} v \l(s, be^{R(b) \Delta t}, \tn\r)
&&
&&\quad\quad (s, b) \in \Omega_{s_{\max}}.
\nonumber
\end{empheq}
\end{subequations}
\end{linenomath}

Next, we write the proposed scheme
at  $(\x_{n, j}^{m-1}) = \l(s_n, b_j,  t_{m-1}\r) \in \Omega \times t_{m-1}$, $m = M, \ldots, 1$,
in an equivalent form via an operator $\mathcal{D}_h(\cdot)$ as follows
\begin{align}
\label{eq:scheme_1}
v_{n, j}^{m-1}  &= \mathcal{D}_h\l(\x_{n,j}^{m-1}, \l\{v_{l, j}^{m} \r\}_{l=-N^{\dagger}/2}^{N^{\dagger}/2}\r)
\coloneqq
 {\myblue{\min\l\{v_{n, j}^{(m-1)+},  \min_{b^+ \in \mathcal{Z}} v_h\l(s^+(s_n, b_j, b^+), \, b^+,  \,t_{m-1}^+\r)\r\},}}
\end{align}
where
$v_h\l(s^+(s_n, b_j, b^+), \, b^+, \, t_{m-1}^+\r) = \mathcal{C}_h\l(\x_{n,j}^{(m-1)+}, \l\{v_{l, j}^{m} \r\}_{l=-N^{\dagger}/2}^{N^{\dagger}/2-1}; b^+\r) = \ldots$
\begin{linenomath}
\begin{subequations}\label{eq:scheme_v_all}
\begin{empheq}[left={= \empheqlbrace}]{alignat=6}
& v_h\l(s^+(s_n, b_je^{R(b_j) \Delta t}, b^+), \, b^+, \, t_m\r)
&&~ n=-N^{\dagger}/2, \ldots, -N/2,
\label{eq:scheme_v_all_1}
\\
&\ds\Delta s \sum_{l=-N^{\dagger}/2}^{N^{\dagger}/2} \omega_{l}\, g_{n - l}(\Delta t;K_{\epsilon}) \,~
v_h\l(s^+(s_l, b_je^{R(b_j) \Delta t}, b^+), \, b^+, \, t_m\r)
&&~ n = -N/2+1, \ldots, N/2-1,
\label{eq:scheme_v}
\\
&
v_h\l(s^+(s_n, b_je^{R(b_j) \Delta t}, b^+), \, b^+, \, t_m\r)
e^{\l(\sigma^2 + 2\mu + \lambda \kappa_2\r)\Delta t}
&&~ n = N/2, \ldots, N^{\dagger}/2.
\label{eq:scheme_v_all_3}
\end{empheq}
\end{subequations}
\end{linenomath}
Similarly, the term $v_{n, j}^{(m-1)+}$ in \eqref{eq:scheme_1} is given by
\begin{linenomath}
\begin{subequations}
\begin{empheq}[left={v_{n, j}^{(m-1)+}= \empheqlbrace}]{alignat=6}
& v_h\l(s_n, \, b_je^{R(b_j) \Delta t}, \, t_m\r)
&&~ n=-N^{\dagger}/2, \ldots, -N/2,
\nonumber
\\
&\ds\Delta s \sum_{l=-N^{\dagger}/2}^{N^{\dagger}/2} \omega_{l}\, g_{n - l}(\Delta t;K_{\epsilon}) \,~
v_h\l(s_l, \,  b_je^{R(b_j) \Delta t}, \, t_m\r)
&&~ n = -N/2+1, \ldots, N/2-1,
\nonumber
\\
&
e^{\l(\sigma^2 + 2\mu + \lambda \kappa_2\r)\Delta t} v_h\l(s_n, b_je^{R(b_j) \Delta t}\, t_m\r)
&&~ n = N/2, \ldots, N^{\dagger}/2.
\nonumber
\end{empheq}
\end{subequations}
\end{linenomath}

We now introduce a lemma on local consistency of the proposed scheme.
\begin{lemma}[Local consistency]
\label{lemma:pointwise_consistency}
Suppose that (i)  the discretization parameter $h$ satisfies \eqref{eq:dis_parameter},
(ii) linear interpolation is used for the intervention action
\eqref{eq:Intervention_Operator_glwb}.
For any smooth test function $\phi \in \mathcal{C}^{\infty}(\Omega \cup \Omega_{b_{\max}} \times [0,T])$,
with $\phi_{n, j}^{m} \equiv \phi({\bf{x}}_{n, j}^{m})$
and ${\bf{x}}_{n, j}^{m-1} \in \Omega_{\myin} \times \{t_{m-1}\}$, $m = M, \ldots, 1$,
and for a sufficiently small $h, \chi$, we have
\EQA
\label{eq:pointwise_consistency}
\mathcal{D}_h \l(\x_{n, j}^{m-1},
\l\{\phi_{l, j}^{m} + \chi \r\}_{l=-N^{\dagger}/2}^{N^{\dagger}/2-1}\r)
=
\mathcal{D}\left({\bf{x}}_{n, j}^{m-1},~\phi^{m}\right)
+ \mathcal{E}\l({\bf{x}}_{n, j}^{m-1}, \epsilon, h\r)
+ \mathcal{O}\l(\chi+ h\r).
\ENA
Here, $\mathcal{E}({\bf{x}}_{n, j}^{(m-1)+}, \epsilon, h) \to 0 $ as $\epsilon, h \to 0$.
The operators $\mathcal{D}\left( \cdot \right)$ and $\mathcal{D}_h(\cdot)$ are
defined in \eqref{eq:con_mean} and \eqref{eq:scheme_1}, respectively,
noting that $\mathcal{D}_h(\cdot)$ depends on $\mathcal{C}_h(\cdot)$ given in \eqref{eq:scheme_v_all}.
\end{lemma}
\begin{proof}[Proof of Lemma~\ref{lemma:pointwise_consistency}]
We first consider the operator $\mathcal{C}_h\l(\cdot\r)$ defined in \eqref{eq:scheme_v_all}.
For the case \eqref{eq:scheme_v} of \eqref{eq:scheme_v_all},
$\mathcal{C}_h\l(\x_{n,j}^{m-1}, \l\{\phi_{l, j}^{m} + \chi \r\}_{l=-N^{\dagger}/2}^{N^{\dagger}/2}; b^+\r)$
becomes
\begin{align}
\mathcal{C}_h\l(\cdot\r) &\overset{\text{(i)}}{=}
\Delta s \sum_{l=-N^{\dagger}/2}^{N^{\dagger}/2} \omega_{l}\, g_{n - l}(\Delta t;K_{\epsilon}) \,
\l(\phi_h\l(s^+(s_l, b_je^{R(b_j) \Delta t}, b^+), \, b^+ , \, t_m\r) + \chi\right)
\nonumber
\\
&\overset{\text{(ii)}}{=} \Delta s \sum_{l=-N^{\dagger}/2}^{N^{\dagger}/2} \omega_{l}\, g_{n - l}(\Delta t;K_{\epsilon}) \,
\l(\phi\l(s^+(s_l, b_je^{R(b_j) \Delta t}, b^+), \, b^+, \, t_m\r) + \mathcal{O}(h^2) + \chi\r)
\nonumber
\\
 &\overset{\text{(iii)}}{=}
\Delta s \sum_{l=-N^{\dagger}/2}^{N^{\dagger}/2} \omega_{l}\, g_{n - l}(\Delta t;\infty) \,
\phi\l(s^+(s_l, b_je^{R(b_j) \Delta t}, b^+), \, b^+, \, t_m\r) + \errorfprime +   \mathcal{O}\l(h^2 + \chi\r)
\nonumber
\\
&\ds\overset{\text{(iv)}}{=} \underbrace{\int_{s^{\dagger}_{\min}}^{s^{\dagger}_{\max}}
\phi\l(s^+(s', b_je^{R(b_j) \Delta t}, b^+), \, b^+, \, t_m\r)
g(s_n- s'; \Delta t, \infty)
~ds'}_{\text{$\LARGE{=~\phi\l(s^+(s_n, b_j, b^+), \, b^+, \, t_{m-1}^+\r)}$}}
 + \errorfprime + \errorcprime + \mathcal{O}\l(h^2 + \chi\r)
\label{eq:OO}
\end{align}
Here, (i) and (ii) are due to the facts that, for linear interpolation,
the constant $\chi$ can be completely separated  from interpolated values,
and the interpolation error is of size $\mathcal{O}\l((\Delta s)^2 + (\Delta b_{\max})^2\r) = \mathcal{O}(h^2)$
for sufficiently small $h$;
(iii) is due to \eqref{eq:qua_error} and $\chi$ being sufficiently small.
The errors $\errorfprime$ and $\errorcprime$  in (iii)
and (iv) are  respectively described below.
\begin{itemize}
 \item In (iii), $\errorfprime \equiv \errorfprime\l(\x_{n, j}^{m-1}, \epsilon\r)$ is the error arising from truncating the infinite series
$g_{n-l}(\cdot, \Delta t, \infty)$, defined in \eqref{eq:g_series}, to
$g_{n-l}(\cdot, \Delta t, K_{\epsilon})$. Taking into account the fact that function
$\phi$ is continuous and hence bounded on the closed domain $\Omega\cup \Omega_{b_{\max}}  \times [0,T])$,
together with \eqref{eq:sum_g} and \eqref{eq:qua_error}, we have $\mid\errorfprime\mid \le C' \epsilon e^{\epsilon_g}$, where $C'>0$ is a bounded constant independently of $\epsilon$.  

\item In (iv),  $\errorcprime \equiv \errorcprime\l(\x_{n, j}^{m-1}, h\r)$ is the error arising from
the simple lhs numerical integration rule
\begin{align}
\label{eq:green_lhs}
\errorcprime &=  \Delta s
\sum_{l=-N^{\dagger}/2}^{N^{\dagger}/2}
    \omega_l
    g_{n-l}(\Delta t, \infty)\,
    \phi\l(s^+(s_l, b_je^{R(b_j) \Delta t}, b^+), \, b^+, \, t_m\r)
\nonumber
    \\
&\qquad \qquad - \int_{s^{\dagger}_{\min}}^{s^{\dagger}_{\max}}
\phi\l(s^+(s', b_je^{R(b_j) \Delta t}, b^+), \, b^+, \, t_m\r)~
g(s_n- s'; \Delta t, \infty)
~ds'.
\end{align}
Due to the continuity and the boundedness of the integrand, we have $\errorcprime \to 0$ as
$h \to 0$.
\end{itemize}
Since $c = b^+$, $\phi$ is smooth, and  $\mathcal{Z}$ is compact,
for $\phi\l(s^+(s_n, b_j, b^+), \, b^+, \, t_{m-1}^+\r)$ given in
\eqref{eq:OO}, we have
\begin{linenomath}
\postdisplaypenalty=0
\begin{align}
\inf_{b^+ \in \mathcal{Z}_h} \phi\l(s^+(s_n, b_j, b^+), \, b^+, \, t_{m-1}^+\r)
&=\inf_{c \in \mathcal{Z}} \phi\l(s^+(s_n, b_j, c), \, c, \, t_{m-1}^+\r) + \mathcal{O}(h)
\nonumber
\\
&=
\inf_{c \in \mathcal{Z}} \mathcal{M}(c) \phi(s_n, b_j, t_{m-1}^+) + \mathcal{O}(h).
\label{eq:II}
\end{align}
\end{linenomath}
{\myblue{Next, similar to \eqref{eq:scheme_v}, the term $v_{n, j}^{(m-1)+}$ in \eqref{eq:con_mean} corresponding to $n = -N/2+1, \ldots, N/2-1$
is written as follows $\ds v_{n, j}^{(m-1)+} = \Delta s \sum_{l=-N^{\dagger}/2}^{N^{\dagger}/2} \omega_{l}\, g_{n - l}(\cdot, \Delta t;K_{\epsilon}) \,~
\l(\phi_h\l(s_l, b_je^{R(b_j) \Delta t}, \, t_m\r)+ \chi \r)= \ldots$
\begin{linenomath}
\postdisplaypenalty=0
\begin{align}
\label{eq:OOO}
\ldots&\overset{\eqref{eq:OO}}{=}
\int_{s^{\dagger}_{\min}}^{s^{\dagger}_{\max}}
\phi(s', b_je^{R(b_j) \Delta t}, t_m)~ g(s_n- s'; \Delta t, \infty)
~ds' + \errorfprime + \errorcprime + \mathcal{O}\l(h^2 + \chi\r).
\nonumber
\\
&= \phi\l(s_n, b_j, t_{m-1}^+\r)
 + \errorfprime + \errorcprime + \mathcal{O}\l(h^2 + \chi\r).
\end{align}
\end{linenomath}
}}
Therefore, using \eqref{eq:OO}, \eqref{eq:OOO} and \eqref{eq:II},
and letting $\mathcal{E}\l({\bf{x}}_{n, j}^{m-1}, \epsilon, h\r) = \errorfprime + \errorcprime$,
for $n = -N/2+1, \ldots, N/2-1$, the operator $\mathcal{D}_h \l(\x_{n, j}^{m-1},
\l\{\phi_{l, j}^{m} + \chi \r\}_{l=-N^{\dagger}/2}^{N^{\dagger}/2}\r)$
in \eqref{eq:pointwise_consistency} can be written as
\begin{linenomath}
\postdisplaypenalty=0
\begin{align}
\mathcal{D}_h(\cdot) &=
{\myblue{\min\l\{\phi\l(s_n, b_j, t_{m-1}^+\r),
\inf_{c \in \mathcal{Z}} \mathcal{M}(c)\,  \phi(s_n, b_j, t_{m-1}^+)\r\}}}
+ \mathcal{O}(h + \chi)
+ \mathcal{E}\l({\bf{x}}_{n, j}^{m-1}, \epsilon, h\r)
\nonumber
\\
&=
\mathcal{D}\l(\x_{n, j}^{m-1},
\l\{\phi_{l, j}^{m}\r\}_{l=-N^{\dagger}/2}^{N^{\dagger}/2}\r)
+ \mathcal{O}(h + \chi)
+
\mathcal{E}\l({\bf{x}}_{n, j}^{m-1}, \epsilon, h\r),
\label{eq:D}
\end{align}
\end{linenomath}
which is \eqref{eq:pointwise_consistency} for ${\bf{x}}_{n, j}^{m-1} \in \Omega_{\myin} \times \{t_{m-1}\}$, $m = M-1, \ldots, 1$.

For the cases \eqref{eq:scheme_v_all_1} and \eqref{eq:scheme_v_all_3},  $\mathcal{C}_h\l(\x_{n,j}^{m-1}, b^+, \l\{\phi_{l, j}^{m} + \chi \r\}_{l=-N^{\dagger}/2}^{N^{\dagger}/2}\r)$ can be respectively written into
\EQS
\label{eq:12}
\left\{\phi\l(s^+(s_n, b_je^{R(b_j) \Delta t}, b^+), \, b^+, \, t_m\r) \text{ and }
e^{\l(\sigma^2 + 2\mu + \lambda \kappa_2\r)\Delta t} \phi\l(s^+(s_n, b_je^{R(b_j) \Delta t}, b^+), \, b^+, \, t_m\r)\right\}
+ \mathcal{O}(\chi + h^2),
\ENS
where arguments similar to those used for (i)-(ii) of \eqref{eq:OO} are used.
Then using \eqref{eq:II} on \eqref{eq:12}, following \eqref{eq:OO} and \eqref{eq:D}, we obtain
\eqref{eq:pointwise_consistency} for ${\bf{x}}_{n, j}^{m-1} \in \l(\Omega_{s_{\min}} \cup \Omega_{s_{\max}}\r) \times \{t_{m-1}\}$, $m = M-1, \ldots, 1$.
\end{proof}

\subsection{Main convergence theorem}
Given the $\ell$-stability and consistency of the proposed numerical scheme established in
Lemmas~\ref{lemma:stability} and \ref{lemma:pointwise_consistency},
as well as together with its monotonicity, we now mathematically demonstrate
the  pointwise convergence of the scheme in $\Omega_{\myin}\times\{t_{m-1}\}$, $m = M, \ldots, 1$, as $h \to 0$.
Here, as noted earlier, we assume that $\epsilon \to 0$ as $h \to 0$.
We first need to recall/introduce relevant notation.

We denote by $\Omega^{h}$ the computational grid parameterized by $h$, noting that $\Omega^{h} \to \Omega$ as $h \to 0$.
We also have the respective $\Omega_{\myin}^h$. In general, a generic gridpoint in $\Omega^{h}_{\myin} \times \{t_m\}$,
$m = M, \ldots, 0$, is denoted by $\x_h^m = (\x_h, t_m)$, whereas an arbitrary point
in $\Omega_{\myin} \times \{t_m\}$ is denoted by $\x^m =  (\x, t_m)$.
Numerical solutions at $(\x_h, t_{m-1})$, $m = M, \ldots, 1$, is denoted by {\myblue{$v_h^{m-1}(\x_h; v_h^{m})$}},
where it is emphasized that $v_h^{m}$, which is the time-$t_m$ numerical solution at gridpoints is used
for the computation of {\myblue{$v_h^{m-1}$}}.
The exact solution at an arbitrary point in $\x^{m-1} = (\x, t_{m-1}) \in \Omega_{\myin}\times\{t_{m-1}\}$,
$m = M,\ldots, 1$, is denoted by {\myblue{$v^{m-1}(\x; v^{m})$}}, where it is emphasized that $v^{m}$, which is the time-$t_m$
exact solution in $\Omega$ is used.
More specifically, $v_h^{m-1}(\x_h; v_h^{m})$ and  $v^{m-1}(\x; v^{m})$ are defined via operators
$\mathcal{D}_h\left(\cdot \right)$ and $\mathcal{D}(\cdot)$ as follows
{\myblue{
\EQ
\label{eq:HG}
v_h^{m-1}\big(\x_h; v_h^{m} \big) \coloneqq \mathcal{D}_h\big(\x_h^{m-1}; \{v_{l, j}^{m}\} \big),
\qquad
v^{m-1}\big(\x; v^{m} \big) \coloneqq \mathcal{D}\big(\x^{m-1}; v^{m}\big),
\quad m = M, \ldots, 1.
\EN
}}
Here, our convention is that {\myblue{$v_h\big(\x^{M-1}; v_h^M\big) = v_h\big(\x^{M-1}; v^M\big)$}}.


The pointwise convergence of the proposed scheme is stated in the main theorem below.
\begin{theorem}[Pointwise convergence]
\label{thm:convergence}
Suppose that all the conditions for Lemma~\ref{lemma:stability} and \ref{lemma:pointwise_consistency}
are satisfied. Under the assumption that the infinite series truncation tolerance $\epsilon \to 0$ as $h \to 0$,
scheme~\eqref{eq:scheme_1} converges pointwise in $\Omega_{\myin} \times \{t_{m-1}\}$, $m \in \{M, \ldots, 1\}$,
to the unique bounded solution of the MV portfolio optimization in Definition~\ref{def:glwb}, i.e.\
for any $m \in\{M, \ldots, 1\}$, we have
\EQ
{\myblue{v^{m-1}\l({\bf{x}}; v^{m}\r)}} ~=~\lim_{\subalign{h &\to 0 \\ {\bf{x}}_h &\to {\bf{x}}}}  v_h^{m-1}\l(\x_h; v_h^m \r),
\quad
\text{ for  }  {\bf{x}}_h \in \Omega_{\myin}^h, \quad {\bf{x}}  \in \Omega_{\myin}.
\label{eq:main}
\EN
\end{theorem}

\begin{proof}[Proof of Theorem~\ref{thm:convergence}]
By Proposition~\ref{prp:uniform}, there exists (bounded) $\phi \in \mathcal{C}^{\infty}(\Omega \times [0,T])$ such that, for any $h > 0$,
\EQ
\label{eq:UC}
v \le \phi \le v + h, \quad  \text{in } \Omega \times \{t_m\}, \quad m = M, \ldots, 0.
\EN
We then define
\begin{linenomath}
\postdisplaypenalty=01
\begin{align*}
{\myblue{
v_h^{m-1}\big(\x_h; \phi^{m} \big) \coloneqq \mathcal{D}_h\big(\x_h^{m-1}; \{\phi_{l, j}^{m+}\} \big),
\quad
v^{m-1}\big(\x; \phi^{m}\big) \coloneqq \mathcal{D}\big(\x^{m-1}; \phi^{m}\big),
}}
\end{align*}
\end{linenomath}
noting our convention that $\phi_{l, j}^{m} = {\myblue{\phi(\x_{l, j}^m)}}$.
To show \eqref{eq:main}, we will prove by mathematical induction on $m$ the following result:
for any $m \in \{M, \ldots, 1\}$, and for sequence $\{\x_h\}_{h > 0}$ such that $\x_h \to \x$ as $h \to 0$,
\EQ
\label{eq:main2}
\left|{\myblue{ v_h^{m-1}\l(\x_h; v^{m}\r) - v^{m-1}\l({\bf{x}}; v^m\r) }}\right| \le \chi_h^{m-1},
\quad \chi_h^{m-1} \text{ is bounded $\forall h >0$ and } \chi_h^{m-1} \to 0 \text { as } h \to 0.
\EN
In the following proof, we let $K_1$, $K_2$, and $K_3$ be generic positive constants independent of $h$
and $\epsilon$, which may take different values from line to line.

\medskip
\noindent \underline{Base case $m = M$:} by \eqref{eq:UC}, we can write $v^M \le \phi^M \le v^M + h$. Therefore,
by monotonicity of the scheme and \eqref{eq:qua_error}, we have
\begin{linenomath}
\postdisplaypenalty=01
\begin{align}
\label{eq:uh}
 v_h^{M-1}\l(\x_h; v^M\r) \leq v_h^{M-1}\l(\x_h; \phi^M\r) \leq v_h^{M-1}\l(\x_h; v^M+h\r) \leq v_h^{M-1}\l(\x_h; v^M\r) + K_1 h.
\end{align}
\end{linenomath}
Using \eqref{eq:uh} and the triangle inequality gives
\begin{linenomath}
\postdisplaypenalty=01
\begin{align}
&\left| v_h^{M-1}\l(\x_h; v^M\r) - v^{M-1}\l({\bf{x}}; v^M\r) \right|
\overset{\eqref{eq:uh}}{\leq} \left|v_h^{M-1}\big(\x_h; \phi^M\big) - v^{M-1}({\bf{x}}; \phi^M) \right| + K_1 h
\nonumber
\\
&\overset{(i)}{\leq} \left| v_h^{M-1}\l({\bf{x}}_h; \phi^M\r) - v^{M-1}\l({\bf{x}}_h; \phi^M\r) \right| + \left| v^{M-1}\l({\bf{x}}_h;
\phi^M\r)     - v^{M-1}\l({\bf{x}}; \phi^M\r) \right|
+ K_1h.
\label{eq:bound_2}
\end{align}
\end{linenomath}
By local consistency established in Lemma~\ref{lemma:pointwise_consistency}, we have
\EQA
 \label{eq:converge_2*}
 v_h^{M-1}\l({\bf{x}}_h; \phi^M\r) - v^{M-1}\l({\bf{x}}_h; \phi^M\r) = \mathcal{E}({\bf{x}}_h^{M-1}, \epsilon, h) +  \mathcal{O}(h).
\ENA
Due to smoothness of $\phi(\cdot)$ and regularity of $g(\cdot)$, we have
\EQA
\label{eq:converge_1*}
\l| v^{M-1}\l(\x_h; \phi^M\r) - v^{M-1}\l({\bf{x}}; \phi^M\r)\r| \le {\myblue{K_1 \|\x_h - {\bf{x}}\|}}.
\ENA
Therefore, using \eqref{eq:bound_2}, \eqref{eq:converge_2*}, \eqref{eq:converge_1*}, we can show that
\begin{linenomath}
\postdisplaypenalty=01
\begin{align}
&\left| v_h^{M-1}\l(\x_h; v^M\r) - v^{M-1}\l({\bf{x}}; v^M\r) \right|
\le \chi_h^{M-1},
\label{eq:base}
\\
&\qquad  \qquad \chi_h^{M-1} = K_1 h
 + \mathcal{O}(h) + |\mathcal{E}({\bf{x}}_h^{M-1}, \epsilon, h)|
 + K_2 \|\x_h - {\bf{x}}\| \longrightarrow 0, \text{ as } h \to 0,
\nonumber
\end{align}
\end{linenomath}
noting $\x_h \to {\bf{x}}$ as $h \to 0$, and $\chi_h^{M-1}$ is bounded for all $h>0$.

\smallskip
\noindent \underline{Induction hypothesis:} assume that, for some {\myblue{$m \in \{M, \ldots, 2\}$}},  we have
\EQ
\label{eq:IH}
\left| v_h^{m-1}\l(\x_h; v_h^m\r) - v^{m-1}\l({\bf{x}}; v^m\r) \right| \le \chi_h^{m-1}, \text{ where $\chi_h^{m-1}$
is bounded, $\chi_h^{m-1} \to 0$ as $h \to 0$}.
\EN

\smallskip
\noindent \underline{Induction step:}
By the triangle inequality, we have
$\left| v_h^{m-2}\l(\x_h; v_h^{m-1}\r) - v^{m-2}\l({\bf{x}}; v^{m-1}\r) \right| \le \ldots$
\begin{linenomath}
\postdisplaypenalty=01
\begin{align}
\ldots\le
\left| v_h^{m-2}\l(\x_h; v_h^{m-1}\r) - {\red{ v_h^{m-2} }}\l(\x_h; v^{m-1}\r)\r|
+
\left|{\red{ v_h^{m-2} }}\l(\x_h; v^{m-1}\r) - v^{m-2}\l({{\bf{x}}}; v^{m-1}\r) \right|.
\label{eq:ind}
\end{align}
\end{linenomath}
Now, we examine the first term \eqref{eq:ind}. By the induction hypothesis \eqref{eq:IH}, $|v_h^{m-1} - v^{m-1}|\le \chi_h^{m-1}$, where $\chi_h^{m-1} \to 0$ as $h \to 0$. Therefore, the first term in \eqref{eq:ind} can be bounded by
\begin{linenomath}
\postdisplaypenalty=01
\begin{align}
\left| v_h^{m-2}\l(\x_h; v_h^{m-1}\r) - {\red{ v_h^{m-2} }}\l(\x_h; v^{m-1}\r)\r|
\overset{\text{(i)}}{\le}
\chi_h'=
{\myblue{ \mathcal{O}(h + \chi_h^{m-1}) }}
+ |\mathcal{E}(\x_h^{m-2}, \epsilon, h)|
\to 0 \text{ as } h \to 0.
\label{eq:is}
\end{align}
\end{linenomath}
Here, (i) follows from the local consistency of the numerical scheme established in Lemma~\ref{lemma:pointwise_consistency}.
Next, we focus on the second term. Using the same arguments for the base case $m = M$ (see \eqref{eq:base}, with $M$ being replaced by $m$), the second term in \eqref{eq:ind} can be bounded by $\chi_h''$, where  $\chi_h'' \to 0$ as $h \to 0$.
Here, we note that {\myblue{$v_h^{m-1}$}} is bounded for all $h > 0$ by Lemma~\ref{lemma:stability} on stability.
Combining this with \eqref{eq:is}, we have $\left| v_h^{m-2}\l(\x_h; {\myblue{ v_h^{m-1} }} \r) - v^{m-2}\l({\bf{x}}; v^{m-1}\r) \right| \le \chi_h^{m-2}$,
where $\chi_h^{m-2}$ is bounded for all $h>0$, and $\chi_h^{m-2} \to 0$ as $h \to 0$. This concludes the  proof.
\end{proof}

\begin{remark}[Convergence on an infinite domain]
It is possible to develop a numerical scheme that converge to the solution
of the theoretical formulation \eqref{eq:recursive_comp}, in particular
to the convolution integral \eqref{eq:recursive_c_comp_int} which is {\myblue{posed
on}} an infinite domain.
This can be achieved by making a requirement  on the discretization parameter $h$ (in addition to the assumption in \eqref{eq:dis_parameter})
as follows:
\EQ
\label{eq:max_s}
s_{\max} - s_{\min} = C_3'/h, \text{ where $C_3'>0$ is independent of $h$}.
\EN
As such, as $h \to 0$, we have $|s_{\min}|, s_{\max} \to \infty$ (implying $|s_{\min}^{\dagger}|, s_{\max}^{\dagger},
|s_{\min}^{\ddagger}|, s_{\max}^{\ddagger}  \to \infty$ as well.
It is straightforward to ensure \eqref{eq:dis_parameter} and \eqref{eq:max_s} simultaneously as $h$ is being refined.
For example, with $C_3' = 1$, we can quadruple $N^{\dagger}$ and sextuple $N^{\ddagger}$ as
$h$ is being halved. Nonetheless, with $|s_{\min}|$ and $s_{\max}$ chosen sufficiently
large as in our extensive experiments, numerical solutions in the interior ($\Omega_{\myin}$)
virtually do not get affected by the boundary conditions.

\end{remark}

\section{Continuously observed impulse control MV portfolio optimization}
\label{sc:conti}
Recall that $\Delta t = t_{m+1} - t_m$ is the rebalancing time interval.
In this section, we intuitively demonstrate that as $\Delta t \to 0$,
of the proposed numerical scheme converge to the viscosity solution \cite{crandall-ishii-lions:1992, barles-souganidis:1991}
of an impulse formulation of the continuously rebalanced MV portfolio optimization
in \cite{DangForsyth2014}. A rigorous analysis of convergence to the viscosity solution
of this impulse formulation is the topic of a paper in progress.

The impulse formulation proposed in \cite{DangForsyth2014} takes the form of an Hamilton-Jacobi-Bellman
Quasi-Variational Inequality (HJB-QVI) as follows (boundary conditions omitted for brevity)
\begin{linenomath}
\begin{subequations}\label{eq:hjb_impulse}
\begin{empheq}[left={\mathcal{F}(v) \coloneqq \empheqlbrace}]{alignat=6}
& \max\l\{ - v_t - \mathcal{L} v  - \mathcal{J}v - R(b) b v_b, \, v - \inf_{c \in \mathcal{Z}}  v(s^+(s, b, c), c, t)\r\} = 0
&&
&&~(s, b, t) \in \mathcal{N} \times [0, T),
\label{eq:hjb_impulse_qvi}
\\
&v \l(s,b,t\r) ~=~ 
\l({\myblue{W_{\t{liq}}(s,b)}} - \f{\gamma}{2} \r)^2,
&&&&~ t = T,
\label{eq:eq:hjb_impulse_term}
\end{empheq}
\end{subequations}
\end{linenomath}
where $\mathcal{L}(\cdot)$ and $\mathcal{J}(\cdot)$ respectively are the differential and jump operators defined as follows
\begin{linenomath}
\postdisplaypenalty=01
\begin{align*}
\mathcal{L} v  \coloneqq \frac{\sigma^2}{2} v_{ss} + \l(\mu - \lambda \kappa - \frac{\sigma^2}{2}\r) v_s  - {\red{\lambda v}},
\qquad
\mathcal{J} v = \int_{-\infty}^{\infty} v(s + y, b, t) p(y)\, \md y.
\end{align*}
\end{linenomath}
A $\ell$-stable and consistent finite difference numerical scheme for the HJB-QVI \eqref{eq:hjb_impulse}
is presented in \cite{DangForsyth2014} in which monotonicity is ensured via a positive coefficient method.
Therefore, convergence of this finite different scheme to the viscosity solution of the HJB-QVI is guaranteed
\cite{crandall-ishii-lions:1992, barles-souganidis:1991}.

To intuitively see that the proposed scheme  \eqref{eq:scheme_1} is consistent in the viscosity sense with the impulse formulation
\eqref{eq:hjb_impulse} in $\Omega_{\myin} \times \{t_{m-1}\}$, {\myblue{$m = M, \ldots, 1$}},
we write \eqref{eq:scheme_1} for ${\myblue{ {\bf{x}}_{n, j}^{m-1} }} \in \Omega_{\myin} \times \{t_{m-1}\}$
via  $\mathcal{G}_h \l(\x_{n, j}^{m-1}, \l\{v_{l, j}^{m}\r\}_{l=-N^{\dagger}/2}^{N^{\dagger}/2}\r)$,
where
\EQ
\label{eq:QQ}
0 = \mathcal{G}_h \l(\cdot\r)
\coloneqq
\max\l\{\underbrace{\frac{v_{n, j}^{m-1} - v_{n, j}^{(m-1)+}}{\Delta t}}_{A_1}, \,
\underbrace{v_{n, j}^{m-1} - \min_{b^+ \in \mathcal{Z}} v_h\l(s^+(s_n, b_j, b^+), \, b^+,  \,t_{m-1}^+\r)}_{A_2}\r\},
\EN
where $v_{n, j}^{(m-1)+}$ and $v_h\l(s^+(s_n, b_j, b^+), \, b^+,  \,t_{m-1}^+\r)$ are respectively given by
\begin{align*}
v_{n, j}^{(m-1)+ }&= \Delta s \sum_{l=-N^{\dagger}/2}^{N^{\dagger}/2} \omega_{l}\, g_{n - l}(\Delta t;K_{\epsilon})
\,~
v_h\l(s_l, b_je^{R(b_j) \Delta t}, t_m\r),
\\
v_h\l(s^+(s_n, b_j, b^+), \, b^+,  \,t_{m-1}^+\r) &= \Delta s \sum_{l=-N^{\dagger}/2}^{N^{\dagger}/2} \omega_{l}\, g_{n - l}(\Delta t;K_{\epsilon}) \,~
v_h\l(s^+(s_l, b_je^{R(b_j) \Delta t}, b^+), \, b^+, \, t_m\r).
\end{align*}
For a smooth test function $\phi$, and a constant $\chi$, term $A_1$ in \eqref{eq:QQ}
of $\mathcal{G}_h \l(\x_{n, j}^{m-1}, \l\{\phi_{l, j}^{m} + \chi \r\}_{l=-N^{\dagger}/2}^{N^{\dagger}/2}\r)$
are
\begin{align}
\label{eq:term2}
\underbrace{\frac{1}{\Delta t}\l(\phi_{n, j}^{m-1} -  \Delta s \sum_{l=-N^{\dagger}/2}^{N^{\dagger}/2} \omega_{l}\, g_{n - l}(\Delta t;K_{\epsilon})
\,~
\phi_h\l(s_l, b_je^{R(b_j) \Delta t}, \, t_m\r)\r)}_{B_1} \,
 + \underbrace{\frac{\chi}{\Delta t}\l( 1 - \Delta s \sum_{l=-N^{\dagger}/2}^{N^{\dagger}/2} \omega_{l}\, g_{n - l}(\Delta t;K_{\epsilon})\r)}_{B_2}.
\end{align}
Now, we first examine the term $B_1$. Noting that
\begin{linenomath}
\postdisplaypenalty=01
\begin{align*}
\phi_h\l(s_l, b_je^{R(b_j) \Delta t}, \, t_m\r) &= \phi_{l, j}^m  + R(b_j)b_j\, (\phi_b)_{l, j}^m\, \Delta t + \mathcal{O}(h^2),
\\
s^+(s_l, b_je^{R(b_j) \Delta t}, b^+) &= s^+(s_l, b_j, b^+) + \mathcal{O}(h),
\nonumber
\end{align*}
\end{linenomath}
we have term-$B_1$ of \eqref{eq:term2} $=\ds \frac{1}{\Delta t}\l(\phi_{n, j}^{m-1} -  \Delta s \sum_{l=-N^{\dagger}/2}^{N^{\dagger}/2} \omega_{l}\, g_{n - l}(\Delta t;K_{\epsilon})
\,~
\phi_h\l(s_l, b_je^{R(b_j) \Delta t}, \, t_m\r)\r) = \ldots$
\begin{linenomath}
\postdisplaypenalty=01
\begin{align}
\label{eq:stay_1}
\ldots
=
\frac{1}{\Delta t}\l(\phi_{n, j}^{m-1} -  \Delta s \sum_{l=-N^{\dagger}/2}^{N^{\dagger}/2} \omega_{l}\, g_{n - l}(\Delta t;K_{\epsilon})
\phi_{l, j}^m\r)
-
\Delta s \sum_{l=-N^{\dagger}/2}^{N^{\dagger}/2} \omega_{l}\, g_{n - l}(\Delta t;K_{\epsilon})
R(b_j)b_j (\phi_b)_{l, j}^m
+ \mathcal{O}(h).
\end{align}
\end{linenomath}
Using similar techniques as in  \cite{online23}[Lemma 5.4], noting the it is possible to show that
\EQ
\label{eq:danglu}
\Delta s \sum_{l=-N^{\dagger}/2}^{N^{\dagger}/2} \omega_{l}\, g_{n - l}(\Delta t;K_{\epsilon})
\phi_{l, j}^m = \phi_{l, j}^m + \Delta t \l[ \mathcal{L} \phi + \mathcal{J} \phi\r]_{n, j}^m + \mathcal{O}(h^2) +
{\myblue{\mathcal{O}\l(\epsilon\r)}}.
\EN
{\myblue{By choosing $\epsilon  = \mathcal{O}(h^2)$}},  from \eqref{eq:danglu}, noting $\phi_{n, j}^{m-1} - \phi_{l, j}^m = -(\phi_t)_{l, j}^m \Delta t + \mathcal{O}\l((\Delta t)^2\r)$,
the first term of \eqref{eq:stay_1} can be written as
\EQ
\label{eq:danglu_1}
\l[-\phi_t - \mathcal{L} \phi - \mathcal{J} \phi\r]_{n, j}^m + \mathcal{O}\l(h\r).
\EN
The second term of \eqref{eq:stay_1} can be simplified  as \cite{online23}[Lemma 5.4]
\EQ
\label{eq:danglu_2}
\Delta s \sum_{l=-N^{\dagger}/2}^{N^{\dagger}/2-1} \omega_{l}\, g_{n - l}(\Delta t;K_{\epsilon})
R(b_j)b_j (\phi_b)_{l, j}^m = R(b_j)b_j (\phi_b)_{n, j}^m +  \mathcal{O}\l(h\r).
\EN
Putting \eqref{eq:danglu_1}-\eqref{eq:danglu_2} into \eqref{eq:stay_1},
noting that term-$B_2$  in \eqref{eq:term2} has the form $\mathcal{O}(\chi)$,
we have
\EQ
\label{eq:danglu_3}
\text{term } A_1 \text{ in } \eqref{eq:QQ} = \l[-\phi_t - \mathcal{L} \phi - \mathcal{J} \phi - R(b)b\phi_b\r]_{n, j}^m + \mathcal{O}\l(h\r) + {\myblue{\mathcal{O}\l(\chi\r)}}.
\EN
Term $A_2$ in \eqref{eq:QQ} can be handled using similar steps as in \eqref{eq:OO}-\eqref{eq:II}.
Thus,  \eqref{eq:QQ} becomes
\EQS
0 =
\max\l\{ \big[-\phi_t - \mathcal{L} \phi - \mathcal{J} \phi - R(b)b\phi_b\big]_{n, j}^{m-1}, \,
\big[\phi - \inf_{c \in \mathcal{Z}} \mathcal{M}(c) \phi \big]_{n, j}^{m-1}\r\}
+ \mathcal{E}\l({\bf{x}}_{n, j}^{(m-1)+}, h\r)
+ \mathcal{O}\l(\chi+ h\r).
\ENS
This show local consistency of the proposed scheme to \eqref{eq:hjb_impulse_qvi}, that is,
\EQ
\mathcal{G}_h \l(\x_{n, j}^{(m-1)}, \l\{\phi_{l, j}^{m} + \chi \r\}_{l=-N^{\dagger}/2}^{N^{\dagger}/2-1}\r)
= \mathcal{F}\l[\phi_{l, j}^{m}\r] +
+ \mathcal{E}\l({\bf{x}}_{n, j}^{(m-1)+}, h\r)
+ \mathcal{O}\l(\chi+ h\r).
\EN
Together with $\ell$-stability and monotonicity of the proposed scheme,
it is possible to utilize a Barles-Souganidis analysis \cite{barles-souganidis:1991}
to show convergence to the viscosity solution of the impulse formulation \eqref{eq:hjb_impulse}
as $h \to 0$. We leave this for our future work.

\section{Numerical examples}
\label{sec:num_test}

\subsection{Empirical data and calibration}
In order to parameterize the underlying asset dynamics, the same calibration
data and techniques are used as detailed in \cite{DangForsyth2016,ForsythVetzal2016}.
We briefly summarize the empirical data sources. The risky asset data
is based on daily total return data (including dividends and other
distributions) for the period 1926-2014 from the CRSP's VWD index\footnote{Calculations were based on data from the Historical Indexes 2015\textcopyright,
Center for Research in Security Prices (CRSP), The University of Chicago
Booth School of Business. Wharton Research Data Services was used
in preparing this article. This service and the data available thereon
constitute valuable intellectual property and trade secrets of WRDS
and/or its third party suppliers. }, which is a capitalization-weighted index of all domestic stocks
on major US exchanges. The risk-free rate is based on 3-month US T-bill
rates\footnote{Data has been obtained from \texttt{See http://research.stlouisfed.org/fred2/series/TB3MS}. }
over the period 1934-2014, and has been augmented with the NBER's
short-term government bond yield data\footnote{Obtained from the National Bureau of Economic Research (NBER) website,
\\
\texttt{http://www.nber.org/databases/macrohistory/contents/chapter13.html}. } for 1926-1933 to incorporate the impact of the 1929 stock market
crash. Prior to calculations, all time series were inflation-adjusted
using data from the US Bureau of Labor Statistics\footnote{The annual average CPI-U index, which is based on inflation data for
urban consumers, were used - see \texttt{http://www.bls.gov.cpi}
.}.

In terms of calibration techniques, the calibration of the jump models
is based on the thresholding technique of \cite{ContTankovBOOK,ContMancini2011}
using the approach of \cite{DangForsyth2016,ForsythVetzal2016} which,
in contrast to maximum likelihood estimation of jump model parameters,
avoids problems such as ill-posedness and multiple local maxima\footnote{If $\Delta\hat{X}_{i}$ denotes the $i$th inflation-adjusted, detrended
log return in the historical risky asset index time series, a jump
is identified in period $i$ if $\left|\Delta\hat{X}_{i}\right|>\alpha\hat{\sigma}\sqrt{\Delta t}$,
where $\hat{\sigma}$ is an estimate of the diffusive volatility,
$\Delta t$ is the time period over which the log return has been
calculated, and $\alpha$ is a threshold parameter used to identify
a jump. For both the Merton and Kou models, the parameters in Table
\ref{tab:process_params} is based on a value of $\alpha=3$
, which means that a jump is only identified in the historical time
series if the absolute value of the inflation-adjusted, detrended
log return in that period exceeds 3 standard deviations of the ``geometric
Brownian motion change'', definitely a highly unlikely event. }. In the case of GBM, standard maximum likelihood techniques are used.
The calibrated parameters are provided in Table~\ref{tab:process_params} (reproduced from \cite{PvSDangForsyth2018_TCMV, PvSDangForsyth2018_MQV}[Table~5.1]).

\noindent
\vspace{+1em}
\begin{minipage}{0.38\textwidth}
\center
\begin{tabular}{crr}
\hline
Ref. level  	& $s$-grid ($N$) 		& $b$-grid ($J$) 	\\ \hline
0					& 128				& 25			\\
1					& 256				& 50			\\
2 					& 512				& 100			\\
3   				& 1024				& 200			\\
4   				& 2048				& 400			\\
\hline
\end{tabular}
\captionof{table}{Grid refinement levels for convergence analysis;
$N^{\dagger} = 2N$ and $N^{\ddagger} = 3N$. }
\label{tab:grid_ref_dis}
\end{minipage}
\hfill
\begin{minipage}{0.62\textwidth}
\vspace*{+0.25cm}
\begin{tabular}{lcc}
\hline
Parameters 									& Merton 	& Kou \\ \hline
$\mu$ (drift) 								& 0.0817	& 0.0874\\
$\sigma$ (diffusion volatility) 			& 0.1453	& 0.1452\\
$\lambda$ (jump intensity) 					& 0.3483	& 0.3483\\
$\widetilde{\mu}$(log jump multiplier mean)& -0.0700	& n/a\\
$\widetilde{\sigma}$(log jump multiplier stdev) & 0.1924	& n/a\\
$q_1$ (probability of an up-jump) 			& n/a		& 0.2903\\
$\eta_1$ (exponential parameter up-jump)	& n/a		& 4.7941\\
$\eta_2$ (exponential parameter down-jump)	& n/a		& 5.4349\\
$r_b$ (borrowing interest rate) 			& 0.00623	& 0.00623\\
$r_{\iota}$ (lending interest rate) 		& 0.00623	& 0.00623\\ \hline
\end{tabular}
\captionof{table}{Calibrated risky and risk-free asset process parameters.
Reproduced from \cite{PvSDangForsyth2018_TCMV, PvSDangForsyth2018_MQV}[Table~5.1].}
\label{tab:process_params}
\end{minipage}
For all experiments, unless otherwise noted, we use $b_{\max} = 1000$, and with the initial wealth being $w_0 = 10$. We set $s_{\min} = -10+\ln(w_0)$, $s_{\max} = 5+\ln(w_0)$, $s_{-\i} = s_{\min}^{\dagger} = -17.5+\ln(w_0)$, and $s_{\max}^{\dagger} = 12.5+\ln(w_0)$, so that $s_{\min}^{\ddagger} = -22.5$ and $s_{\max}^{\ddagger} = 22.5$.

For the user-defined tolerance $\epsilon$ used for the truncation of the infinite series representation of $g(\cdot)$ in \eqref{eq:K_Oh}, we use $\epsilon = 10^{-20}$, which can be satisfied in discretely rebalancing examples when $K_{\epsilon} = 15$
(i.e.\ the number terms in the truncated series of $g(\cdot)$ is 15).

\subsection{Validation examples}
Since for PCMV portfolio optimization under a solvency
condition (no bankruptcy allowed)  and a maximum leverage condition does not admit known
analytical solution, we rely on existing numerical methods, namely (i) finite difference \cite{DangForsyth2014, DangForsythVetzal2017}
and (ii) Monte Carlo (MC) simulation to verify results.
For brevity, we will refer to the proposed monotone integration method as the ``MI'' method,
and to the finite difference method of \cite{DangForsyth2014, DangForsythVetzal2017} as  the ``FD'' method.

As noted earlier, to the best  of our knowledge, the FD methods proposed in \cite{DangForsyth2014, DangForsythVetzal2017} are the only
existing FD methods for MV optimization under jump-diffusion dynamics with investment constraints.
In discrete rebalancing setting, FD methods typically involve solving, between two consecutive rebalancing times,
a Partial Integro Differential Equation (PIDE), where the amount invested in the risky asset $z = e^s$ is the independent variable.
These FD methods achieve monotonicity in time-advancement through a positive coefficient finite difference discretization method (for the partial derivatives), which is combined with implicit timestepping. Optimal strategies are obtained by solving an optimization problem across rebalancing times. Despite their effectiveness, finite difference methods present significant computational challenges in multi-period settings with very long maturities, as encountered in DC plans. In particular, they necessitate time-stepping between rebalancing dates (i.e., control monitoring dates), which often occur annually. This time-stepping requirement introduces errors and substantially increase the computational cost of FD  methods (as noted earlier in Remark~\ref{rm:complexity}.
In the numerical experiments, the FD results are obtained on the same computational domain as those obtained by the MI method
with the number of partition points in the $z$- and $b$-grids being $2048$ and $400$, respectively, with 256 timesteps between yearly rebalancing times.
(This approximates to a daily timestepping.)

Validation against MC simulation is proceeded in two steps. In Step 1, we solve the PCMV problem using the MI method on a relatively fine computational grid: the number of partition points in the $s$- and $b$-grids are $N = 1024$ and $J = 200$, respectively.
During this step, the optimal controls are stored for each discrete state value and
rebalancing time $t_m\in \mathcal{T}_M$. In Step 2, we carry out Monte Carlo simulations with $10^6$ paths from $t = 0$ to $t = T$ following these stored numerical optimal strategies for asset allocation across each $t_m\in \mathcal{T}_M$, using linear interpolation, if necessary, to determine the controls for a given state value.
For Step 2, an Euler's timestepping method is used for timestepping between consecutive rebalancing times
and we use a total of $180$ timesteps.

\subsubsection{Discrete rebalancing}
For discretely rebalancing experiments, we use $T = \{20, 30\}$ (years), with  $\Delta t = 1$ year (yearly rebalancing). The the details of the mesh size/timestep refinement level used are given in Table~\ref{tab:grid_ref_dis}.

Table~\ref{tab:PCMV_Kou_20_30} presents the numerically computed $E_{\Ccal_0}^{x_0,t_0} \l[W_T\r]$ and $Std_{\Ccal_0}^{x,t_0} \l[W_T\r]$
under the Kou model obtained for different refinement levels with  $\gamma = 100$ and $T = 20$ and $T = 30$ (years).
To provide an estimate of the convergence rate of the proposed MI method, we compute the ``Change'' as the difference in values from the coarser grid and the ``Ratio'' as the ratio of changes between successive grids.
For validation purposes, $E_{\Ccal_0}^{x_0,t_0} \l[W_T\r]$ and $Std_{\Ccal_0}^{x,t_0} \l[W_T\r]$ obtained by FD method,
as well as those obtained by MC methods, together with 95\% confidence intervals (CIs), are also provided.
As evident from Table~\ref{tab:PCMV_Kou_20_30}, means and standard deviations obtained by the MI method
exhibit excellent agreement with those obtained by the FD method and MC simulation.

Results obtained by MC simulation agree with those obtained by our numerical method.
Results for the Merton jump case when $T = 20$ and $T = 30$ (years) are presented in Table~\ref{tab:PCMV_Merton}
and similar observations can be made.
\begin{table}[htb!]
\begin{center}
\caption{PCMV under the Kou model with parameters in the Table~\ref{tab:process_params}; $\gamma = 100$;
solvency constraint applied; maximum leverage constraint applied  with $q_{\max} = 1.5$; $T = \{20, 30\}$ years.}
\label{tab:PCMV_Kou_20_30}
\begin{tabular}{|c|lc|ccc|ccc|}
\hline
$T$  &Method		& Ref.\ level		& $E_{\Ccal_0}^{x_0,t_0} \l[W_T\r]$			 & Change					& Ratio	
								& $Std_{\Ccal_0}^{x,t_0} \l[W_T\r]$			 & Change					& Ratio		\\
\hline
&				&$0$			& $35.7768$		& 							 &				
								& $16.2451$		& 							 &							
				\\
&				&$1$			& $35.9828$		& $0.2060$					 &				
								& $15.7801$		& $0.4650$					 &					
				\\
&MI				&$2$			& $36.1051$		& $0.1223$					 & 	$1.7$
								& $15.6192$		& $0.1609$					 &	$2.9$
				\\
&				&$3$			& $36.1893$		& $0.0842$					 &	$1.5$		
								& $15.5995$		& $0.0197$					 &	$8.2$		
				\\
20&				&$4$			& $36.2182$		& $0.0289$					 &	$2.9$		
								& $15.5942$		& $0.0053$					 &	$3.7$	
				\\
\cline{2-9}
(years)&MC	& &$36.1994$ && &$15.5965$&& \\
&$95\%$ CI	& &\multicolumn{3}{l|}{~~$[36.1688, 36.2299]$} &\multicolumn{3}{r|}{$ $} \\
\cline{2-9}
 &FD	&  &$36.2208$ &\multicolumn{2}{r|}{$ $} &$15.5938$&\multicolumn{2}{r|}{$ $} \\
\hline
\hline
&				&$0$			& $41.9856$	& 							&				
								& $14.1207$	& 							&						
				\\
&				&$1$			& $42.1626$		& $0.1770$				&				
								& $13.2359$		& $0.8848$				&						
				\\
&MI				&$2$			& $42.2776$		& $0.1150$				& 	 $1.5$
								& $12.9544$		& $0.2815$				&	 $3.1$
				\\
&				&$3$			& $42.3462$		& $0.0686$				&	 $1.7$		
								& $12.8772$		& $0.0772$				&	 $3.6$		
				\\
30&				&$4$			& $42.3834$		& $0.0372$				&	 $1.8$		
								& $12.8569$		& $0.0203$				&	 $3.8$
				\\
\cline{2-9}
(years)&MC	& &$42.3827$ && &$12.8512$&&
\\
&$95\%$ CI	& &\multicolumn{3}{l|}{~~$[42.3575, 42.4079]$} &\multicolumn{3}{r|}{$ $}
\\
\cline{2-9}
&FD	& &$42.3850$ &\multicolumn{2}{r|}{$ $} &$12.8520$&\multicolumn{2}{r|}{$ $} \\
\hline
\end{tabular}
\end{center}
\end{table}

In Figure~\ref{fig:EF} presents we present efficient frontiers
for the Merton model (Figure~\ref{fig:EF}~(a)) and for the Kou model (Figure~\ref{fig:EF}~(b))
obtained by the MI's methods with refinement level ($N = 1024$ and $J = 200$).
We observe that efficient frontiers produced by  the MI's method agree well with those obtained by the FD method.
\begin{figure}[htb!]
	\centering
	\subfigure[Merton, $T = 20$, $\Delta t = 1$]{
		\label{fig:EF_merton}
		\includegraphics[width = 0.48\textwidth] {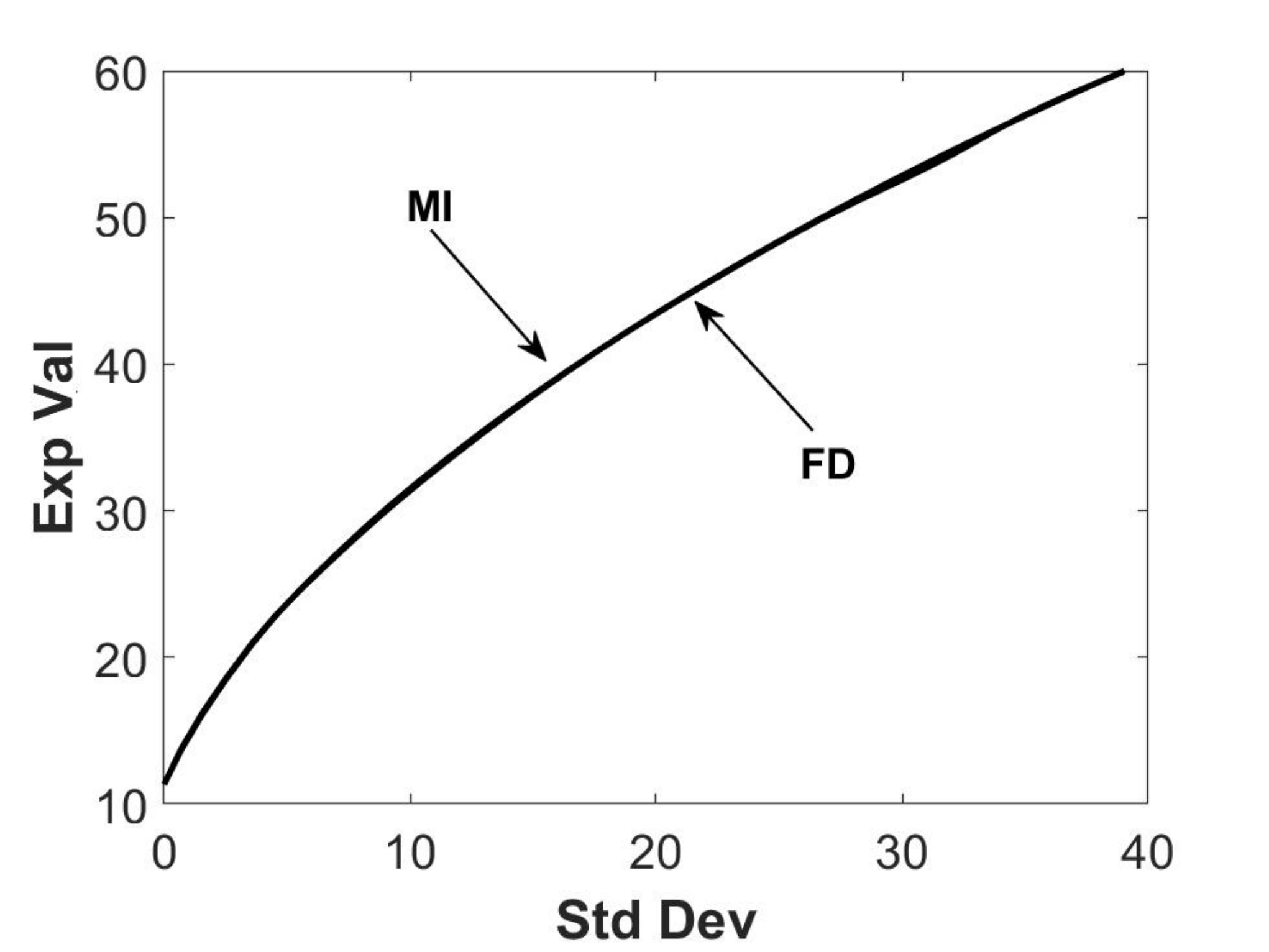}}
	\subfigure[Kou, $T = 20$, $\Delta t = 1$]{
		\label{fig:EF_kou}
		\includegraphics[width = 0.48\textwidth] {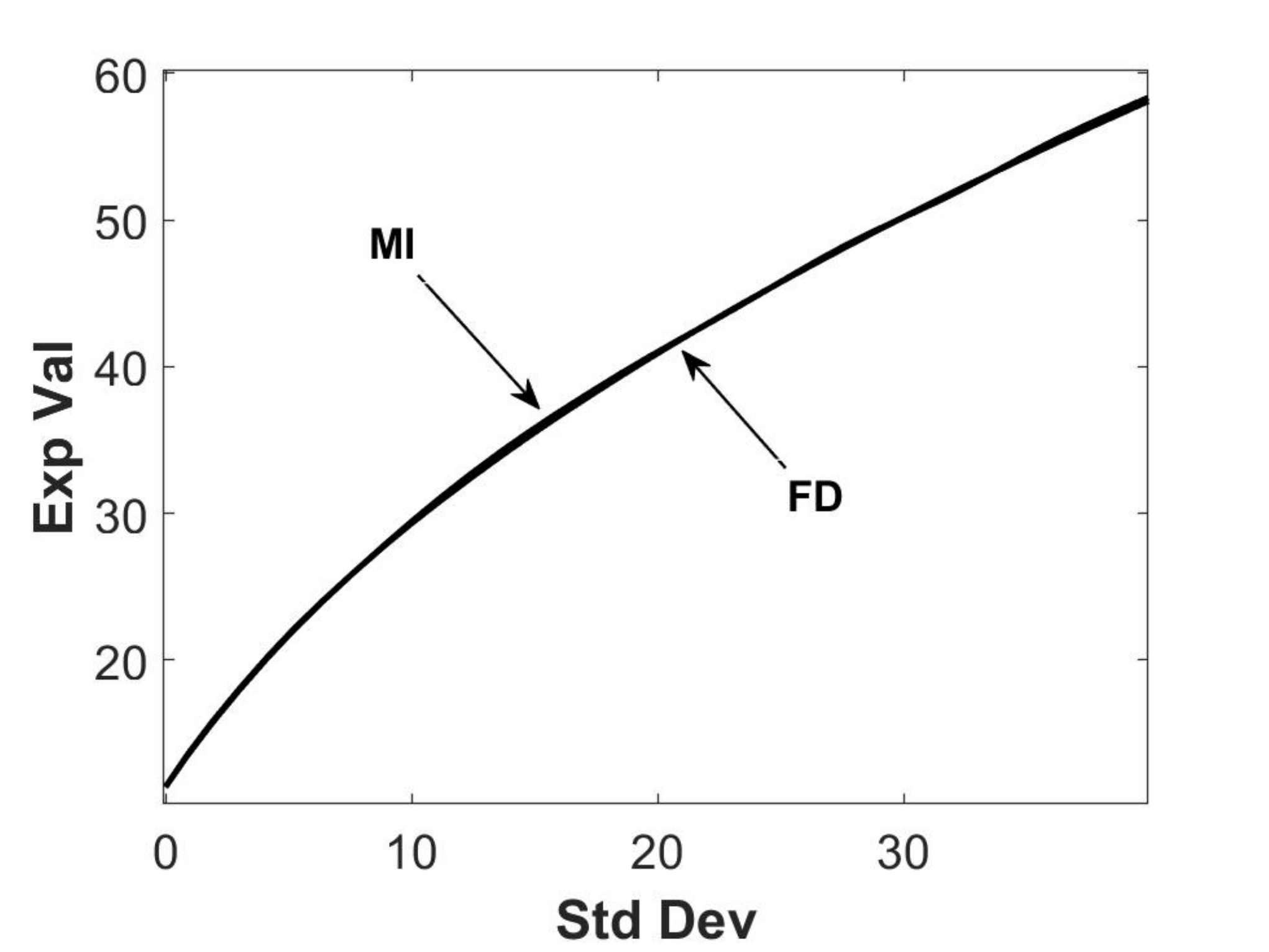}}
	\caption{Efficient frontier; parameters in the Table~\ref{tab:process_params}; $T = 20$; $\Delta t = 1$;
solvency constraint applied; maximum leverage constraint applied with $q_{\max} = 1.5$; refinement level $3$.}
	\label{fig:EF}
\end{figure}

\begin{table}[htb!]
\begin{center}
\caption{PCMV under the Merton model with parameters in the Table~\ref{tab:process_params}; $\gamma = 250$; liquidate risky asset when insolvency; $q_{\max} = 3.0$; $T = \{20, 30\}$ years.}
\label{tab:PCMV_Merton}
\begin{tabular}{|c|lc|ccc|ccc|}
\hline
&Method		& Ref.\ level		& $E_{\Ccal_0}^{x_0,t_0} \l[W_T\r]$			 & Change					& Ratio	
								& $Std_{\Ccal_0}^{x,t_0} \l[W_T\r]$			 & Change					& Ratio		\\
\hline
&				&$0$			& $72.0225$		& 							 &				
								& $46.6043$		& 							 &				
				\\			
&				&$1$			& $73.2063$		& $1.1838$					 &				
								& $46.4594$		& $0.1449$					 &		
				\\
&MI				&$2$			& $73.7386$		& $0.5323$					 & 	$2.2$
								& $46.5059$		& $0.0465$					 &	$3.1$
				\\
&				&$3$			& $73.9277$		& $0.1891$					 &	$2.8$		
								& $46.5172$		& $0.0113$					 &	$4.1$	
				\\
20&				&$4$			& $73.9916$		& $0.0639$					 &	$3.0$	
								& $46.5247$		& $0.0075$					 &	$1.5$
				\\
\cline{2-9}
(years)&MC	& &$73.9041$ && &$46.5249$&& \\
&$95\%$ CI	& &\multicolumn{3}{l|}{~~$[73.8129, 73.9953]$} &\multicolumn{3}{r|}{$ $} \\
\cline{2-9}
 &FD	& & 73.9518&\multicolumn{2}{r|}{$ $} &46.5288&\multicolumn{2}{r|}{$ $} \\
\hline
\hline
&				&$0$			& $91.7970$	& 							&				
								& $42.3493$	& 							&						
				\\
&				&$1$			& $92.9707$	& $1.1737$					&				
								& $41.9321$	& $0.4172$					&						
				\\
&MI				&$2$			& $93.4524$	& $0.4817$					& 	 $2.4$
								& $41.7774$	& $0.1547$					&	 $2.7$
				\\
&				&$3$			& $93.6167$	& $0.1643$					&	 $2.9$	
								& $41.7263$	& $0.0511$					&	 $3.0$	
				\\
30&				&$4$			& $93.6844$	& $0.0677$					& 	 $2.4$			
								& $41.7197$	& $0.0066$					& 	 $7.7$	
				\\
\cline{2-9}
(years)&MC	& &$93.6696$ && &$41.7753$&&
\\
&$95\%$ CI	& &\multicolumn{3}{l|}{~~$[93.5877, 93.7515]$} &\multicolumn{3}{r|}{$ $}
\\
\cline{2-9}
&FD	& & 93.6812&\multicolumn{2}{r|}{$ $} &41.7168&\multicolumn{2}{r|}{$ $} \\
\hline
\end{tabular}
\end{center}
\end{table}

\newpage
\subsubsection{Continuous rebalancing}
While continuous rebalancing is not the primary focus of this paper, we believe it is valuable to include numerical results for the continuous rebalancing setting in order to validate the method. For experiments in the continuous rebalancing setting, we use $T = 10$ (years), with  $\Delta t = T/M$ year. The details of the mesh size/timestep refinement level used are given in Table~\ref{tab:grid_ref_cont}.
The model parameters, according to \cite{DangForsyth2014}, are given in Table~\ref{tab:process_params_cont}.

In these experiments, we also consider the scenario where the dynamics of the risky asset follow a Geometric Brownian Motion (GBM) model. In cases where pure diffusion is desired, such as in a GBM model, the occurrence of jumps can be eliminated by setting $\lambda = 0$.
Table~\ref{tab:PCMV_GBM} displays the numerical results for $E_{\Ccal_0}^{x_0,t_0} \l[W_T\r]$ (expected value) and $Std_{\Ccal_0}^{x,t_0} \l[W_T\r]$ (standard deviation) obtained at various refinement levels. These results correspond to the case where $\gamma = 80$ and $T = 10$ years, assuming a GBM model.
For validation purposes, we also include the expected values and standard deviations obtained using the FD method proposed by \cite{DangForsyth2014}. The results are presented alongside the values obtained through the proposed MI method in Table~\ref{tab:PCMV_GBM}. Notably, it is evident from the table that the FD and MI results show excellent agreement.

\noindent
\vspace{+1em}
\begin{minipage}{0.5\textwidth}
\center
\begin{tabular}{cccc}
\hline
Ref. level  		& Timesteps $M$		& $s$-grid $N$ 		& $b$-grid $J$ 	 \\ \hline
0					& 10				& 128				& 25			 \\
1					& 20				& 256				& 50			 \\
2 					& 40				& 512				& 100			 \\
3   				& 80				& 1024				& 200			 \\
4   				& 160				& 2048				& 400			 \\
\hline
\end{tabular}
\captionof{table}{Grid and timestep refinement levels for convergence analysis;
$N^{\dagger} = 2N$ and $N^{\ddagger} = 3N$.}
\label{tab:grid_ref_cont}
\end{minipage}
\hfill
\begin{minipage}{0.5\textwidth}
\begin{tabular}{lcc}
\hline
Parameters 									& GBM 	& Merton \\ \hline
$\mu$ (drift) 								& 0.15	& 0.0795487\\
$\sigma$ (diffusion volatility) 			& 0.15	& 0.1765\\
$\lambda$ (jump intensity) 					& n/a	& 0.0585046\\
$\widetilde{\mu}$(log jump multiplier mean)& n/a	& -0.788325\\
$\widetilde{\sigma}$(log jump multiplier stdev) & n/a	& 0.450500\\
$r_b$ (borrowing interest rate) 			& 0.04	& 0.0445\\
$r_{\iota}$ (lending interest rate) 		& 0.04	& 0.0445\\ \hline
\end{tabular}
\captionof{table}{Calibrated risky and risk-free asset process parameters.
Reproduced from Table 7.2  in \cite{dang2014continuous}.}
\label{tab:process_params_cont}
\end{minipage}

\begin{table}[htb!]
\begin{center}
\caption{PCMV under the GBM model (no jumps) with parameters in
Table~\ref{tab:process_params_cont}; $\gamma = 80$; liquidate risky asset when insolvency; $q_{\max} = \i$; $T = 10$ years.}
\label{tab:PCMV_GBM}
\begin{tabular}{|ccccc|ccc|}
\hline
Method		& Ref.\ level		& $E_{\Ccal_0}^{x_0,t_0} \l[W_T\r]$			 & Change					& Ratio	
								& $Std_{\Ccal_0}^{x,t_0} \l[W_T\r]$			 & Change					& Ratio		\\
\hline
				&$0$			& $37.4438$		& 							 &				
								& $6.2582$		& 							 &				
				\\
				&$1$			& $38.3201$		& $0.8763$					 &				
								& $5.4393$		& $0.8189$					 &			
				\\
MI				&$2$			& $38.4084$		& $0.0883$					 & 	$9.9$
								& $5.1285$		& $0.3108$					 &	$2.6$
				\\
				&$3$			& $38.4689$		& $0.0605$					 &	$1.5$		
								& $5.0394$		& $0.0891$					 &	$3.5$		
				\\
				&$4$			& $38.4820$		& $0.0131$					 &	$4.6$	
								& $4.9999$		& $0.0395$					 &	$2.3$
				\\ \cline{1-8}
FD	& &38.4789 & \multicolumn{2}{l|}{$ $} &5.0888& \multicolumn{2}{l}{$ $} \\ \hline
\end{tabular}
\end{center}
\end{table}

\begin{table}[htb!]
\begin{center}
\caption{PCMV under the Merton model with parameters in the Table~\ref{tab:process_params_cont}; $\gamma = 200$; liquidate risky asset when insolvency; $q_{\max} = 2.0$; $T = 10$ years.}
\label{tab:PCMV_Merton_cont}
\begin{tabular}{|ccccc|ccc|}
\hline
Method		& Ref.\ level		& $E_{\Ccal_0}^{x_0,t_0} \l[W_T\r]$			 & Change					& Ratio	
								& $Std_{\Ccal_0}^{x,t_0} \l[W_T\r]$			 & Change					& Ratio		\\
\hline
				&$0$			& $24.1538$		& 							 &				
								& $21.5101$		& 							 &						
				\\
				&$1$			& $24.4528$		& $0.2990$					 &				
								& $22.0738$		& $0.5637$					 &				
				\\
MI				&$2$			& $24.5468$		& $0.0940$					 & 	$3.2$
								& $22.1659$		& $0.0921$					 &	$6.1$
				\\
				&$3$			& $24.5855$		& $0.0387$					 &	$2.4$		
								& $22.1907$		& $0.0248$					 &	$3.7$	
				\\
				&$4$			& $24.6042$ 	& $0.0187$ 					 & $2.1$
								& $22.2005$ 	& $0.0098$ 					 & $2.5$
				\\ \cline{1-8}
FD	& &24.6032 & \multicolumn{2}{l|}{$ $} &22.2024& \multicolumn{2}{l}{$ $} \\ \hline
\end{tabular}
\end{center}
\end{table}

\section{Conclusion}
\label{sc:conclude}
In this study, we present a highly efficient, straightforward-to-implement, and monotone numerical integration method for MV portfolio optimization. The model considered in this paper addresses a practical context that includes a variety of investment constraints, as well as jump-diffusion dynamics that govern the price processes of risky assets. Our method employs an infinite series representation of the transition density, wherein all series terms are strictly positive and explicitly computable. This approach enables us to approximate the convolution integral for time-advancement over each rebalancing interval via a monotone integration scheme. The scheme uses a composite quadrature rule, simplifying the computation significantly. Furthermore, we introduce an efficient implementation of the proposed monotone integration scheme using FFTs, exploiting the structure of Toeplitz matrices. The pointwise convergence of this scheme, as the discretization parameter approaches zero, is rigorously established. Numerical experiments affirm the accuracy of our approach, aligning with benchmark results obtained through the FD method and Monte Carlo simulation, as demonstrated in \cite{DangForsyth2014}. Notably, our proposed method offers superior efficiency
compared to existing FD methods, owing to its computational complexity being an order of magnitude lower.

Further work includes investigation of self-exciting jumps for MV optimization, together with a convergence analysis as $\Delta t \to 0$
to a continuously observed impulse control formulation for MV optimization taking the form of a HJB Quasi-Variational Inequality.

{\fontsize{9.65}{10.3}\rm

\ifx\allfiles\undefined
\documentclass[11pt, sort]{article}
\usepackage{hyperref, bm}
\usepackage{tabularx}
\usepackage{comment}
\linespread{1.1}
%
\usepackage{soul}
\usepackage{geometry}
\newcommand{\hili}[1]{\colorbox{yellow}{#1}}
 \geometry{
 left=1.8cm,
 right=1.8cm,
 top=1.9cm,
 bottom=1.9cm,
 }
 \usepackage{mathtools}
\newcommand{\x}{{\bf{x}}}
\AtBeginDocument{
\addtolength{\abovedisplayskip}{-0.9ex}
\addtolength{\abovedisplayshortskip}{-0.9ex}
\addtolength{\belowdisplayskip}{-0.9ex}
\addtolength{\belowdisplayshortskip}{-0.9ex}
}
\makeatletter
\def\thm@space@setup{\thm@preskip=3pt
\thm@postskip=3pt}
\makeatother
\usepackage{titlesec}
\titlespacing*{\section}
  {0pt}{0.7\baselineskip}{0.2\baselineskip}
\titlespacing*{\subsection}
  {0pt}{0.7\baselineskip}{0.2\baselineskip}

\usepackage{graphicx}
\usepackage{amsmath, amsthm, amssymb, amsfonts, latexsym}
\usepackage{array}
\usepackage{caption}
\usepackage{color}
\newcommand{\hilight}[1]{\colorbox{yellow}{#1}}
\setlength{\captionmargin}{20pt}
\renewcommand{\captionfont}{\small\itshape}
\renewcommand{\captionlabelfont}{\small\scshape}
\usepackage{subfigure}
\renewcommand{\subcapsize}{\scriptsize}
\renewcommand{\subcaplabelfont}{\normalfont}
\renewcommand{\subcapfont}{\itshape}
\renewcommand{\subfigcapmargin}{30pt}
\renewcommand{\subfigcapskip}{0pt}
\renewcommand{\subfigbottomskip}{0pt}

\usepackage{epsfig}
\usepackage{fancybox}
\usepackage{array}
\usepackage{shadow}

\usepackage{mathrsfs}
\usepackage{longtable}
\usepackage{calc}
\usepackage{multirow}
\usepackage{hhline}
\usepackage{ifthen}
\usepackage{algorithmic}
\usepackage{algorithm}
\usepackage{color}

\usepackage[running, mathlines]{lineno}

\DeclareMathOperator{\sgn}{sgn}

\newcommand*\patchAmsMathEnvironmentForLineno[1]{%
  \expandafter\let\csname old#1\expandafter\endcsname\csname #1\endcsname
  \expandafter\let\csname oldend#1\expandafter\endcsname\csname end#1\endcsname
  \renewenvironment{#1}%
     {\linenomath\csname old#1\endcsname}%
     {\csname oldend#1\endcsname\endlinenomath}}%
\newcommand*\patchBothAmsMathEnvironmentsForLineno[1]{%
  \patchAmsMathEnvironmentForLineno{#1}%
  \patchAmsMathEnvironmentForLineno{#1*}}%
\AtBeginDocument{%
\patchBothAmsMathEnvironmentsForLineno{equation}%
\patchBothAmsMathEnvironmentsForLineno{align}%
\patchBothAmsMathEnvironmentsForLineno{flalign}%
\patchBothAmsMathEnvironmentsForLineno{alignat}%
\patchBothAmsMathEnvironmentsForLineno{gather}%
\patchBothAmsMathEnvironmentsForLineno{multline}%
}

\numberwithin{equation}{section}
\numberwithin{table}{section}
\numberwithin{figure}{section}


\newtheorem{rmk}{\scshape Remark}[section]
\newtheorem{definition}{Definition}
\newtheorem{theorem}{Theorem}
\newtheorem{lemma}{Lemma}
\newtheorem{remark}{Remark}
\newtheorem{assumption}{Assumption}
\newtheorem{condition}{Condition}
\newtheorem{property}{Property}
\newtheorem{proposition}{Proposition}
\newtheorem{corollary}{Corollary}

\numberwithin{definition}{section}
\numberwithin{theorem}{section}
\numberwithin{lemma}{section}
\numberwithin{remark}{section}
\numberwithin{assumption}{section}
\numberwithin{condition}{section}
\numberwithin{property}{section}
\numberwithin{proposition}{section}
\numberwithin{corollary}{section}
\numberwithin{algorithm}{section}
\usepackage{empheq}

\newcommand{\e}{\mathrm{e}}
\newcommand{\blue}{\color{black}}
\newcommand{\myblue}{\color{blue}}

\newcommand{\g}{\color{black}}
\newcommand{\red}{\color{red}}
\newcommand{\EQ}{\begin{equation}}
\newcommand{\EN}{\end{equation}}
\newcommand{\EQS}{\begin{equation*}}
\newcommand{\ENS}{\end{equation*}}
\newcommand{\EQA}{\begin{eqnarray}}
\newcommand{\ENA}{\end{eqnarray}}
\newcommand{\EQAS}{\begin{eqnarray*}}
\newcommand{\ENAS}{\end{eqnarray*}}
\newcommand{\md}{\mathrm{d}}
\newcommand{\Ebb}{\mathbb{E}}
\newcommand{\Ibb}{\mathbb{I}}
\newcommand{\Qbb}{\mathbb{Q}}

\newcommand{\tn}{t_m}
\newcommand{\tnn}{t_{m+1}}
\newcommand{\tnp}{t_{m-1}}

\newcommand{\zn}{z_n}
\newcommand{\znn}{z_{n+1}}
\newcommand{\znp}{z_{n-1}}
\newcommand{\zp}{z^{\prime}}
\newcommand{\Var}{\mathcal{V}}
\newcommand{\var}{\nu}
\newcommand{\varn}{\nu_n}
\newcommand{\varnn}{\nu_{n+1}}
\newcommand{\varnp}{\nu_{n-1}}
\newcommand{\Lnv}{\mathcal{Q}}
\newcommand{\lnv}{\varsigma}
\newcommand{\lnvn}{\varsigma_n}
\newcommand{\lnvnn}{\varsigma_{n+1}}
\newcommand{\lnvnp}{\varsigma_{n-1}}
\newcommand{\lnvp}{\varsigma^{\prime}}

\newcommand{\ds}{\displaystyle}
\def\r{\right}
\def\l{\left}
\def\f{\frac}
\def\b{\textbf}
\def\s{\sqrt}
\def\i{\infty}
\def\t{\text}

\newcommand{\Acal}{\mathcal{A}}
\newcommand{\Bcal}{\mathcal{B}}
\newcommand{\Ccal}{\mathcal{C}}
\newcommand{\Fcal}{\mathcal{F}}
\newcommand{\Gcal}{\mathcal{G}}
\newcommand{\Hcal}{\mathcal{H}}
\newcommand{\Jcal}{\mathcal{J}}
\newcommand{\Lcal}{\mathcal{L}}
\newcommand{\Mcal}{\mathcal{M}}
\newcommand{\Ncal}{\mathcal{N}}
\newcommand{\Ocal}{\mathcal{O}}
\newcommand{\Pcal}{\mathcal{P}}
\newcommand{\Qcal}{\mathcal{Q}}
\newcommand{\Rcal}{\mathcal{R}}
\newcommand{\Scal}{\mathcal{S}}
\newcommand{\Tcal}{\mathcal{T}}
\newcommand{\Zcal}{\mathcal{Z}}

\newcommand{\sinc}{\text{sinc}}
\newcommand{\Rit}{\mathit{R}}
\newcommand{\winf}{s_{\scalebox{0.6}{-$\infty$}}}
\newcommand{\wpinf}{s_{\scalebox{0.6}{$\infty$}}}
\newcommand{\epinf}{e^{\wpinf}}
\newcommand{\einf}{e^{\winf}}

\def\r{\right}
\def\l{\left}
\newcommand{\Oinf}{\Omega^{\scalebox{0.5}{$\infty$}}}
\newcommand{\myin}{\scalebox{0.7}{\text{in}}}

\AtBeginDocument{
\addtolength{\abovedisplayskip}{-0.2ex}
\addtolength{\abovedisplayshortskip}{-0.2ex}
\addtolength{\belowdisplayskip}{-0.2ex}
\addtolength{\belowdisplayshortskip}{-0.2ex}
}

\usepackage{titlesec}
\titlespacing*{\section}
  {0pt}{0.5\baselineskip}{0.2\baselineskip}
\titlespacing*{\subsection}
  {0pt}{0.5\baselineskip}{0.2\baselineskip}
\titlespacing*{\subsubsection}
  {0pt}{0.5\baselineskip}{0.2\baselineskip}

\newcommand{\errorf}{\scalebox{0.9}{$\mathcal{E}_{f}$}}
\newcommand{\errorfprime}{\scalebox{0.9}{$\mathcal{E}_{f}$}}
\newcommand{\errorp}{\scalebox{0.9}{$\mathcal{E}_{o}$}}
\newcommand{\errorpprime}{\scalebox{0.9}{$\mathcal{E}_{o}$}}
\newcommand{\errorg}{\scalebox{0.9}{$\mathcal{E}_{\hat{g}}$}}
\newcommand{\errorgprime}{\scalebox{0.9}{$\mathcal{E}_{\hat{g}}$}}
\newcommand{\errorb}{\scalebox{0.9}{$\mathcal{E}_{b}$}}
\newcommand{\errorbprime}{\scalebox{0.9}{$\mathcal{E}_{b}$}}
\newcommand{\errorc}{\scalebox{0.9}{$\mathcal{E}_{c}$}}
\newcommand{\errorcprime}{\scalebox{0.9}{$\mathcal{E}_{c}$}}

\newcommand{\errorq}{\scalebox{0.9}{$\mathcal{E}_{q}$}}
\newcommand{\errorqprime}{\scalebox{0.9}{$\mathcal{E}_{q}\varsigma$}}

\begin{document}
\else
\section*{Appendices}
\fi
\appendix
\section{Proof of Lemma~\ref{lemma:1}}
\label{sec:app_g_lemma}
Recalling the inverse Fourier transform $\mathfrak{F}^{-1}[\cdot]$ in \eqref{eq:small_g}
and the closed-form expression for $G(\eta; \Delta t)$ in \eqref{eq:big_F}, we have
\EQA
\label{eq:g_ift}
{\myblue{
g(s; \Delta t) ~=~
}}
\f{1}{2\pi} \, \int_{-\infty}^{\infty}
e^{- \alpha \eta^2 +
\l(\beta +s\r) \l(i \eta \r) + \theta} \, e^{\lambda \Gamma\l( \eta \r) \Delta t} ~ \md \eta,
\ENA
where $\alpha$, $\beta$ and $\theta$ are given in \eqref{eq:g_series}, and $\Gamma\l( \eta \r) = \int_{-\infty}^{\infty} p(y)~e^{i \eta y}~dy$.
Following the approach developed {\myblue{in \cite{dangJacksonSues2016, dang_luis_2015, berthe2019shannon} }}, we expand the term $e^{\lambda \Gamma(\eta) \Delta t}$ in \eqref{eq:g_ift} in a Taylor series,  noting that
\EQA
\l(\Gamma(\eta)\r)^k &=& \l(\int_{-\i}^{\i} p(y) \exp(i \eta y) \md y\r)^k
= \int_{-\infty}^\infty \ldots \int_{-\infty}^\infty
            \prod_{\ell = 1}^{k} p(y_{\ell}) \exp\l(i\eta Y_k\r)
            \md y_1 \md y_2 \ldots \md y_k
\ENA
where $p(y)$ is the probability density of $\xi$, and $Y_k = \sum_{\ell=1}^{k} y_{\ell}$, with $Y_0 = 0$, and
for $k = 0$, $\l(\Gamma(\eta)\r)^0 = 1$. Then, we have
{\myblue{
\EQA
g(s; \Delta t) &=&
\label{eq:g_proof_int}
\f{1}{2\pi}
\sum_{k=0}^{\i} \f{\l(\lambda\Delta t\r)^k}{k!}
\int_{-\infty}^{\infty}
e^{- \alpha \eta^2 +\l(\beta +s\r) \l(i \eta \r) + \theta} \,
\l(\Gamma\l( \eta \r) \r)^k ~ \md \eta,
\\
\label{eq:g_proof_sum}
&{\buildrel (\text{i}) \over =}& 
\f{1}{\s{4 \pi \alpha}}
\sum_{k=0}^{\infty} \f{\l(\lambda \Delta t\r)^{k}}{k!}
\int_{-\infty}^\infty \ldots \int_{-\infty}^\infty
            	\exp\l(\theta  -
	\f{\l(\beta + s +  Y_k\r)^2}
	{4 \alpha}\r)
    \l(\prod_{\ell = 1}^{k} p(y_\ell)\r) \md y_1 \ldots \md y_k,
\ENA
where the first term of the series corresponds to $k = 0$ and is equal to $\f{1}{\s{4\pi \alpha}}\exp\l(\theta - \f{\l(\beta + s \r)^2}{4 \alpha}\r)$.
Here, in (i), we use the Fubini's theorem and the well known result $\int_{-\i}^{\i}e^{-a\phi^2 - b\phi}\md\phi = \s{\f{\pi}{a}} e^{b^2/4a}$.
}}

\section{$\boldsymbol{\Omega_{s_{\max}}}$: boundary expressions}
\label{sec:v_u_s_max}
Recalling the sub-domain definitions in \eqref{eq:sub_domain_whole}, we observe that $\Omega_{s_{\max}}$ is the boundary where $s \rightarrow \i$.
For fixed $b$, noting the terminal condition \eqref{eq:recursive_b_comp_initial*}, we assume that $v(s \rightarrow \i, b, t) \approx A_1(t) e^{2s}$ for some unknown function $A_1(t)$. 
Using the infinite series representation of $g(s-s'; \Delta t)$ given Lemma~\ref{lemma:1} (proof in {\myblue{ Appendix~\ref{sec:app_g_lemma}) }}
and the integral~\eqref{eq:recursive_c_comp_int},
we have $v(s, b, t_{m-1}^+) =  \int_{-\infty}^{\infty} A_1(t_{m}^-) e^{2s'} g(s-s'; \Delta t)~ \md s' ~ \ldots$
\EQA
\label{eq:v_zmax_untegral}
\ldots &~=~&  \int_{-\infty}^{\infty}
\f{A_1(t_{m}^-)}{\s{4\pi\alpha}} \sum_{k=1}^{\i} \f{\l(\lambda\Delta t\r)^k}{k!}
\int_{-\infty}^\infty \ldots \int_{-\infty}^\infty
            	\exp\l(\theta - \f{\l(\beta + s - s' + Y_k\r)^2}{4\alpha} + 2s'\r)
    \prod_{\ell = 1}^{k} p(y_{\ell})
\md y_1 \ldots \md y_k~\md s'
\nonumber\\
&& \qquad \qquad \qquad \qquad \qquad \qquad
+ \int_{-\infty}^{\infty}
\f{A_1(t_{m}^-)}{\s{4\pi\alpha}}
            	\exp\l(\theta - \f{\l(\beta + s - s'\r)^2}{4\alpha} + 2s'\r) \md s',
\nonumber\\
&~=~&
A_1(t_{m}^-) \exp\l(\theta + 4\alpha + 2\l(\beta + s\r)\r)
\sum_{k=1}^{\i} \f{\l(\lambda\Delta t\r)^k}{k!}
\int_{-\infty}^\infty \ldots \int_{-\infty}^\infty
\exp\l(2Y_k\r)
\prod_{\ell = 1}^{k} p(y_{\ell})
\md y_1 \ldots \md y_k
\nonumber\\
&& \qquad \qquad \qquad \qquad \qquad \qquad
+ A_1(t_{m}^-) \exp\l(\theta + 4\alpha + 2\l(\beta + s\r)\r),
\nonumber\\
&~=~&  A_1(t_{m}^-) \exp\l\{\l(2\mu + \sigma^2 -2\lambda \kappa -\lambda\r)\Delta t + 2s\r\} \,
\sum_{k=0}^{\i} \f{\l(\lambda\Delta t\r)^k}{k!}
\l(\int_{-\infty}^\infty e^{2y}p(y) \md y \r)^k,
\nonumber\\
&~=~&  A_1(t_{m}^-) e^{2s} \,
\exp\l\{\l(2\mu + \sigma^2 -2\lambda \kappa -\lambda\r)\Delta t\r\} \,
\exp\l\{\lambda\l(\kappa_2 + 2\kappa + 1\r)\Delta t\r\},
\nonumber\\
&~=~&  v(s, b,t_m^-) \,
e^{\l(\sigma^2 + 2\mu + \lambda \kappa_2\r)\Delta t},
\nonumber
\ENA
where we use
$\alpha = \f{\sigma^2}{2} \, \Delta t$, $\beta = \l(\mu - \lambda \kappa - \f{\sigma^2}{2}\r) \, \Delta t$ and $\theta = -\lambda \Delta t$.

Similarly, for each fixed $B(t)=b$, we assume the {\myblue{auxiliary linear value function}} $u(s \rightarrow \i, b, t) \approx A_2(t) e^s$, for some unknown function $A_2(t)$.
we have ${\myblue{u(s, b, t_{m-1}^+)}} =  \int_{-\infty}^{\infty} A_2(t_{m}^-) e^{s'} g(s-s'; \Delta t)~ \md s' ~ \ldots$
\EQA
\label{eq:u_zmax_untegral}
\ldots &~=~&  \int_{-\infty}^{\infty}
\f{A_2(t_{m}^-)}{\s{4\pi\alpha}} \sum_{k=1}^{\i} \f{\l(\lambda\Delta t\r)^k}{k!}
\int_{-\infty}^\infty \ldots \int_{-\infty}^\infty
\exp\l(\theta - \f{\l(\beta + s - s' + Y_k\r)^2}{4\alpha} + s'\r)
\prod_{\ell = 1}^{k} p(y_{\ell})
\md y_1 \ldots \md y_k~ \md s'
\nonumber\\
&& \qquad \qquad \qquad \qquad \qquad \qquad
+ \int_{-\infty}^{\infty}
\f{A_2(t_{m}^-)}{\s{4\pi\alpha}}
            	\exp\l(\theta - \f{\l(\beta + s - s'\r)^2}{4\alpha} + s'\r) \md s',
\nonumber\\
&~=~&
A_2(t_{m}^-) \exp\l(\theta + \alpha + \beta + s\r)
\sum_{k=1}^{\i} \f{\l(\lambda\Delta t\r)^k}{k!}
\int_{-\infty}^\infty \ldots \int_{-\infty}^\infty
\exp\l(Y_k\r)
\prod_{\ell = 1}^{k} p(y_{\ell})
\md y_1 \ldots \md y_k
\nonumber\\
&& \qquad \qquad \qquad \qquad \qquad \qquad
+ A_2(t_{m}^-) \exp\l(\theta + \alpha + \beta + s\r),
\nonumber\\
&~=~&  A_2(t_{m}^-) \exp\l\{\l(\mu -\lambda \kappa -\lambda\r)\Delta t + s\r\} \,
\sum_{k=0}^{\i} \f{\l(\lambda\Delta t\r)^k}{k!}
\l(\int_{-\infty}^\infty e^{y}p(y) \md y \r)^k,
\nonumber\\
&~=~&  A_2(t_{m}^-) e^{s} \,
\exp\l\{\l(\mu -\lambda \kappa -\lambda\r)\Delta t\r\} \,
\exp\l\{\lambda \kappa \Delta t + \lambda\Delta t\r\},
\nonumber\\
&~=~&  {\myblue{u(s, b,t_m^-)}} \, e^{\mu \, \Delta t}.
\nonumber
\ENA

{\myblue{
\section{Proof of $\boldsymbol{g\l(s; \Delta t \r)}$ for $\boldsymbol{\xi\sim \text{Asym-Double-Exponential}(q_1,\eta_1,\eta_2)}$}
}}
\label{sec:app_g}
In this case, according to Lemma~\ref{lemma:1}, we have $g(s; \Delta t) = \ldots$
\begin{align}
\ldots
&=
\f{\exp\l(\theta - \f{\l(\beta + s \r)^2}{4 \alpha}\r)}{\s{4\pi \alpha}} +
\f{e^{\theta}}{\s{4 \pi \alpha}}
\sum_{k=1}^{\infty} \f{\l(\lambda \Delta t\r)^{k}}{k!}
\underbrace{\int_{-\infty}^\infty \ldots \int_{-\infty}^\infty
            	\exp\l( -
	\f{\l(\beta + s +  Y_k\r)^2}{4 \alpha}\r)
    \prod_{\ell = 1}^{k} p(y_\ell) 
    \md y_1 \ldots \md y_k}_{E_k}
\nonumber\\
&=
\f{\exp\l(\theta - \f{\l(\beta + s \r)^2}{4 \alpha}\r)}{\s{4\pi \alpha}} +
\f{e^{\theta}}{\s{4 \pi \alpha}}
\sum_{k=1}^{\infty} \f{\l(\lambda \Delta t\r)^{k}}{k!}~E_k.
\label{eq:g_Ek}
\end{align}
Here, the term $E_k$ in \eqref{eq:g_Ek} is clearly non-negative and can be computed as
\EQA
\label{eq:E_k_def}
\displaystyle
E_k =
\int_{-\infty}^{\infty}
\exp\l(- \f{\l(\beta + s +  y\r)^2}{4 \alpha}\r) \,
p_{\hat{\xi}_k}(y) \md y,
\ENA
where $p_{\hat{\xi}_k}(y)$ is the PDF of the random variable $\hat{\xi}_k = \sum_{\ell=1}^{k} \xi_{\ell}$, for fixed $k$.
To find $p_{\hat{\xi}_k}(y)$, the key step is the decomposition of $\hat{\xi}_k = \sum_{\ell=1}^{k} \xi_{\ell}$ into sums of i.i.d exponential random variables \cite{kou01}. More specifically, we have
\EQA
\label{eq:X_J}
\hat{\xi}_k = \sum_{\ell=1}^{k} \xi_{\ell} &{\buildrel dist.\ \over =}&
\begin{cases}
\hat{\xi}_{\ell}^+  = ~~\sum_{i=1}^{\ell} \varepsilon_i^+ & \mbox{with probability } Q_1^{k,\ell},
\quad \ell = 1,\ldots,k \\
\hat{\xi}_{\ell}^-
= - \sum_{i=1}^{\ell} \varepsilon_i^- & \mbox{with probability } Q_2^{k,\ell},
\quad \ell = 1,\ldots,k
\end{cases} \ .
\ENA
Here,
$Q_1^{k,\ell}$ and $Q_2^{k,\ell}$ are given in \eqref{eq:PQ},
and $\varepsilon_i^+$ and $\varepsilon_i^-$ are i.i.d.\ exponential
variables with rates $\eta_1$ and $\eta_2$, respectively.
The PDF for each of the cases in \eqref{eq:X_J} respectively are
\begin{equation}
\begin{aligned}
\label{eq:f_w_n}
p_{\hat{\xi}_{\ell}^+}(y) &= \frac{\e^{- \eta_1 y} \, y^{\ell-1} \, \eta_1^{\ell}}{(\ell-1)!}
\quad \text{for  } \hat{\xi}_{\ell}^+,
\quad
\text{and}
\quad
p_{\hat{\xi}_{\ell}^-}(y) &= \frac{\e^{\eta_2 y} \, (-y)^{\ell-1} \, \eta_2^{\ell} }{(\ell-1)!}
\quad \text{for } \hat{\xi}_{\ell}^-  \ .
\end{aligned}
\end{equation}
Taking into account \eqref{eq:X_J}-\eqref{eq:f_w_n}, \eqref{eq:E_k_def} becomes
\begin{align}
\label{eq:E_J_compute}
E_k
&= \sum_{\ell=1}^k \, Q_1^{k,\ell} \underbrace{\int_{0}^{\infty}
\exp\l( - \f{\l(\beta + s  + y \r)^2}{4\alpha} \r) \,
p_{\hat{\xi}_{\ell}^+}(y) \, \mathrm{d}y}_{E_{1,\ell}}
+ \sum_{\ell=1}^k \, Q_2^{k,\ell} \underbrace{\int_{-\infty}^{0}
\exp\l( - \f{\l(\beta + s + y \r)^2}{4\alpha} \r) \,
p_{\hat{\xi}_{\ell}^-}(y) \, \mathrm{d}y}_{E_{2,\ell}}.
\end{align}
Considering the term $E_{1,\ell}$,
\begin{equation*}
\begin{aligned}
E_{1,\ell}
= \int_0^\infty
\exp\l( - \f{\l(\beta + s + y \r)^2}{4\alpha} \r) \,
\frac{\e^{-\eta_1 y} \, y^{\ell-1} \, {\eta_1}^{\ell}}{(\ell-1)!} \, \mathrm{d}y
= {\eta_1}^{\ell} \, \int_0^\infty \frac{1}{(\ell-1)!} \,
\exp\l( - \f{\l(\beta + s + y \r)^2}{4\alpha} -
\eta_1 y \r) \,
y^{\ell-1} \, \mathrm{d}y \ .
\end{aligned}
\end{equation*}
Making the change of variable
$y_1 = \f{\beta + s + y}{\s{2\alpha}}$,
\begin{equation*}
\begin{aligned}
E_{1,\ell}
&= {\eta_1}^{\ell} \, \int_{\f{\beta + s }{\sqrt{2\alpha}}}^\infty
\frac{1}{(\ell-1)!} \,
\e^{-\frac{1}{2} y_1^2 - \eta_1 \sqrt{2\alpha} y_1} \,
\e^{\eta_1 \, \l(\beta+ s \r)} \,
\l(\sqrt{2\alpha} \, y_1 - \beta - s  \r)^{\ell-1} \,
\sqrt{2\alpha} \, \mathrm{d}y_1\\
&= \l(\eta_1 \, \sqrt{2\alpha}\r)^{\ell} \,
\e^{\eta_1 \, \l(\beta+ s  \r)} \,
\int_{\f{\beta + s}{\sqrt{2\alpha}}}^\infty \frac{1}{(\ell-1)!} \,
\e^{-\frac{1}{2} y_1^2 - \eta_1 \sqrt{2\alpha} y_1} \,
\l(y_1 - \f{\beta+s}{\sqrt{2\alpha}}\r)^{\ell-1} \, \mathrm{d}y_1 \ .
\end{aligned}
\end{equation*}
Making the change of variable $y_2 = y_1 + \eta_1 \sqrt{2\alpha}$,
\begin{equation}
\label{eq:a1k}
\begin{aligned}
E_{1,\ell}
&= \l(\eta_1 \, \sqrt{2\alpha}\r)^{\ell} \,
\e^{\eta_1 \, \l(\beta+s  \r) + \eta_1^2 \alpha} \,
\int_{\f{\beta + s }{\sqrt{2\alpha}} + \eta_1 \sqrt{2\alpha}}^\infty \frac{1}{(\ell-1)!} \,
\l(y_2 - \l(\eta_1 \sqrt{2\alpha} + \f{\beta+s }{\sqrt{2\alpha}}\r)\r)^{\ell-1} \,
\e^{-\frac{1}{2} y_2^2} \,\mathrm{d}y_2 \\
&= \l(\eta_1 \, \sqrt{2\alpha}\r)^{\ell} \,
\e^{\eta_1 \, \l(\beta+s  \r) + \eta_1^2 \alpha} \,
\mathrm{Hh}_{\ell-1}\l(\eta_1 \sqrt{2\alpha} + \f{\beta+s }{\sqrt{2\alpha}}\r)\ ,
\end{aligned}
\end{equation}
where $\mathrm{Hh}_{\ell}$ is defined in \eqref{eq:Hhk}.
Similarly for the term $E_{2,\ell}$,
\begin{equation}
\label{eq:a2k}
\begin{aligned}
E_{2,\ell}
&= \int_{-\infty}^0 \,
\exp\l( - \f{\l(\beta + s  + y \r)^2}{4\alpha} \r) \,
\frac{\e^{\eta_2 y} \, \l(-y\r)^{\ell-1} \, \eta_2^{\ell} }{(\ell-1)!} \, \mathrm{d}y\\
&= \l(\eta_2 \, \sqrt{2\alpha}\r)^{\ell} \,
\e^{-\eta_2 \, \l(\beta+s  \r)} \,
\int_{-\infty}^{\f{\beta + s }{\sqrt{2\alpha}}} \frac{1}{(\ell-1)!} \,
\e^{-\frac{1}{2} y_1^2 + \eta_2 \, \sqrt{2\alpha} \, y_1} \,
\l(-y_1 + \f{\beta + s }{\sqrt{2\alpha}} \r)^{\ell-1} \, \mathrm{d}y_1 \\
&= \l(\eta_2 \, \sqrt{2\alpha}\r)^{\ell} \,
\e^{-\eta_2 \, \l(\beta+s  \r) + \eta_2^2 \alpha} \,
\int_{-\f{\beta + s }{\sqrt{2\alpha}} + \eta_2 \sqrt{2\alpha}}^{\infty} \frac{1}{(\ell-1)!} \,
\l(y_2 - \l(\eta_2 \sqrt{2\alpha} - \f{\beta + s }{\sqrt{2\alpha}}\r) \r)^{\ell-1}
\e^{-\frac{1}{2} y_2^2} \,\mathrm{d}y_2 \\
&= \l(\eta_2 \, \sqrt{2\alpha}\r)^{\ell} \,
\e^{-\eta_2 \, \l(\beta+s  \r) + \eta_2^2 \alpha} \,
\mathrm{Hh}_{\ell-1}\l(\eta_2 \sqrt{2\alpha} - \f{\beta + s }{\sqrt{2\alpha}} \r),
\end{aligned}
\end{equation}
where $\mathrm{Hh}_{\ell}$ is defined in \eqref{eq:Hhk}.
Using \eqref{eq:E_J_compute}, \eqref{eq:a1k} and \eqref{eq:a2k} together with further simplifications gives us \eqref{eq:g_EJ_double_exp}.

\ifx\allfiles\undefined
\bibliographystyle{mynatbib}
\setlength{\bibsep}{1.0ex}
\bibliography{paper2_bib}

\begin{thebibliography}{10}

\bibitem{AbramowitzStegun1972}
M.~Abramowitz and I.~A. Stegun.
\newblock {\em {Handbook of mathematical functions}}.
\newblock {Dover, New York}, 1972.

\bibitem{alonso-garcca_wood_ziveyi_2018}
J.~Alonso-Garcca, O.~Wood, and J.~Ziveyi.
\newblock Pricing and hedging guaranteed minimum withdrawal benefits under a
  general {L}{\'e}vy framework using the {COS} method.
\newblock {\em Quantitative Finance}, 18:1049--1075, 2018.

\bibitem{barles-souganidis:1991}
G.~Barles and P.E. Souganidis.
\newblock Convergence of approximation schemes for fully nonlinear equations.
\newblock {\em Asymptotic Analysis}, 4:271--283, 1991.

\bibitem{BasakChabakauri2010}
S.~Basak and G.~Chabakauri.
\newblock Dynamic mean-variance asset allocation.
\newblock {\em Review of Financial Studies}, 23:2970--3016, 2010.

\bibitem{berthe2019shannon}
Edouard Berthe, Duy-Minh Dang, and Luis Ortiz-Gracia.
\newblock A shannon wavelet method for pricing foreign exchange options under
  the heston multi-factor cir model.
\newblock {\em Applied Numerical Mathematics}, 136:1--22, 2019.

\bibitem{BjorkMurgoci2014}
T.~{Bj{\"o}rk} and A.~Murgoci.
\newblock A theory of {M}arkovian time-inconsistent stochastic control in
  discrete time.
\newblock {\em Finance and Stochastics}, (18):545--592, 2014.

\bibitem{Bokanowski2018}
O.~Bokanowski, A.~Picarelli, and C.~Reisinger.
\newblock High-order filtered schemes for time-dependent second order {HJB}
  equations.
\newblock {\em ESAIM Mathematical Modelling and Numerical Analysis}, 52:69--97,
  2018.

\bibitem{BourgeronEtAl2018}
T.~Bourgeron, E.~Lezmi, and T.~Roncalli.
\newblock Robust asset allocation for robo-advisors.
\newblock {\em Working paper}, 2018.

\bibitem{BrycDemboJiang_2006}
Wlodzimierz Bryc, Amir Dembo, and Tiefeng Jiang.
\newblock Spectral measure of large random hankel, markov and toeplitz
  matrices.
\newblock {\em Annals of probability}, 34(1):1--38, 2006.

\bibitem{butler2021data}
Andrew Butler and Roy~H Kwon.
\newblock Data-driven integration of regularized mean-variance portfolios.
\newblock {\em arXiv preprint arXiv:2112.07016}, 2021.

\bibitem{ChenLiGuo2013}
Z.~Chen, G.~Li, and J.~Guo.
\newblock Optimal investment policy in the time consistent mean--variance
  formulation.
\newblock {\em Insurance: Mathematics and Economics}, 52(2):145--156, March
  2013.

\bibitem{CongOosterlee2016}
F.~Cong and C.W. Oosterlee.
\newblock On pre-commitment aspects of a time-consistent strategy for a
  mean-variance investor.
\newblock {\em Journal of Economic Dynamics and Control}, 70:178--193, 2016.

\bibitem{cong2017robust}
F~Cong and CW~Oosterlee.
\newblock On robust multi-period pre-commitment and time-consistent
  mean-variance portfolio optimization.
\newblock {\em International Journal of Theoretical and Applied Finance},
  20(07):1750049, 2017.

\bibitem{ContMancini2011}
R.~Cont and C.~Mancini.
\newblock Nonparametric tests for pathwise properties of semi-martingales.
\newblock {\em Bernoulli}, (17):781--813, 2011.

\bibitem{ContTankovBOOK}
R.~Cont and P.~Tankov.
\newblock {\em Financial modelling with jump processes}.
\newblock Chapman and Hall / CRC Press, 2004.

\bibitem{crandall-ishii-lions:1992}
M.~G. Crandall, H.~Ishii, and P.~L. Lions.
\newblock User's guide to viscosity solutions of second order partial
  differential equations.
\newblock {\em Bulletin of the American Mathematical Society}, 27:1--67, 1992.

\bibitem{dangJacksonSues2016}
D.~M. Dang, K.~R. Jackson, and S.~Sues.
\newblock A dimension and variance reduction {Monte Carlo} method for pricing
  and hedging options under jump-diffusion models.
\newblock {\em Applied Mathematical Finance}, 24:175--215, 2017.

\bibitem{dang_luis_2015}
D.~M. Dang and Luis Ortiz-Gracia.
\newblock {A dimension reduction Shannon-wavelet based method for option
  pricing}.
\newblock {\em Journal of Scientific Computing}, 2017.
\newblock {https://doi.org/10.1007/s10915-017-0556-y}.

\bibitem{DangForsyth2014}
D.M. Dang and P.A. Forsyth.
\newblock Continuous time mean-variance optimal portfolio allocation under jump
  diffusion: A numerical impulse control approach.
\newblock {\em Numerical Methods for Partial Differential Equations},
  30:664--698, 2014.

\bibitem{DangForsyth2016}
D.M. Dang and P.A. Forsyth.
\newblock Better than pre-commitment mean-variance portfolio allocation
  strategies: A semi-self-financing {H}amilton--{J}acobi--{B}ellman equation
  approach.
\newblock {\em European Journal of Operational Research}, (250):827--841, 2016.

\bibitem{DangForsythVetzal2017}
D.M. Dang, P.A. Forsyth, and K.R. Vetzal.
\newblock The 4 percent strategy revisited: a pre-commitment mean-variance
  optimal approach to wealth management.
\newblock {\em Quantitative Finance}, 17(3):335--351, 2017.

\bibitem{dang2014continuous}
Duy-Minh Dang and Peter~A Forsyth.
\newblock Continuous time mean-variance optimal portfolio allocation under jump
  diffusion: An numerical impulse control approach.
\newblock {\em Numerical Methods for Partial Differential Equations},
  30(2):664--698, 2014.

\bibitem{dang2016convergence}
Duy-Minh Dang, Peter~A Forsyth, and Yuying Li.
\newblock Convergence of the embedded mean-variance optimal points with
  discrete sampling.
\newblock {\em Numerische Mathematik}, 132(2):271--302, 2016.

\bibitem{EltonGruberEtAlBOOK}
E.J. Elton, M.J. Gruber, S.J. Brown, and W.N. Goetzmann.
\newblock {\em Modern portfolio theory and investment analysis}.
\newblock Wiley, 9th edition, 2014.

\bibitem{Fang2008}
F.~Fang and C.W. Oosterlee.
\newblock A novel pricing method for {E}uropean options based on
  {F}ourier-{C}osine series expansions.
\newblock {\em SIAM Journal on Scientific Computing}, 31:826--848, 2008.

\bibitem{ForsythLabahn2017}
P.~A. Forsyth and G.~Labahn.
\newblock {$\epsilon$-monotone Fourier methods for optimal stochastic control
  in finance}.
\newblock 2017.
\newblock Working paper, School of Computer Science, University of Waterloo.

\bibitem{ForsythVetzal2016}
P.A. Forsyth and K.R. Vetzal.
\newblock Dynamic mean variance asset allocation: Tests for robustness.
\newblock {\em International Journal of Financial Engineering}, 4:2, 2017.
\newblock 1750021 (electronic).

\bibitem{ForsythVetzal2019}
P.A. Forsyth and K.R. Vetzal.
\newblock Optimal asset allocation for retirement saving: Deterministic vs.
  time consistent adaptive strategies.
\newblock {\em Applied Mathematical Finance}, 26(1):1--37, 2019.

\bibitem{ForsythVetzalWestmacott2019}
P.A. Forsyth, K.R. Vetzal, and G.~Westmacott.
\newblock Management of portfolio depletion risk through optimal life cycle
  asset allocation.
\newblock {\em North American Actuarial Journal}, 23(3):447--468, 2019.

\bibitem{FC1980}
F.~N. Fritsch and R.~E. Carlson.
\newblock Monotone piecewise cubic interpolation.
\newblock {\em SIAM Journal on Numerical Analysis}, 17:238--246, 1980.

\bibitem{Guo2009}
X.~Guo and G.~Wu.
\newblock Smooth fit principle for impulse control of multidimensional
  diffusion processes.
\newblock {\em SIAM Journal on Control and Optimization}, 48(2):594--617, 2009.

\bibitem{HojgaardVigna2007}
B.~Hojgaard and E.~Vigna.
\newblock Mean-variance portfolio selection and efficient frontier for defined
  contribution pension schemes.
\newblock {\em Research Report Series, Department of Mathematical Sciences,
  Aalborg University}, R-2007-13, 2007.

\bibitem{Huang2015}
Y.T. Huang and Y.K. Kwok.
\newblock Regression-based {M}onte {C}arlo methods for stochastic control
  models: variable annuities with lifelong guarantees.
\newblock {\em Quantitative Finance}, 16(6):905--928, 2016.

\bibitem{Huang2018}
Y.T. Huang, P.~Zeng, and Y.K. Kwok.
\newblock Optimal initiation of guaranteed lifelong withdrawal benefit with
  dynamic withdrawals.
\newblock {\em SIAM Journal on Financial Mathematics}, 8:804--840, 2017.

\bibitem{kou01}
S.~G. Kou.
\newblock A jump diffusion model for option pricing.
\newblock {\em {M}anagement {S}cience}, 48:1086--1101, August 2002.

\bibitem{KouOriginal}
S.G. Kou.
\newblock A jump-diffusion model for option pricing.
\newblock {\em Management Science}, 48(8):1086--1101, 2002.

\bibitem{LiNg2000}
D.~Li and W.-L. Ng.
\newblock Optimal dynamic portfolio selection: multi period mean variance
  formulation.
\newblock {\em Mathematical Finance}, 10:387--406, 2000.

\bibitem{li2000optimal}
D.~Li and W.L. Ng.
\newblock {Optimal Dynamic Portfolio Selection: Multiperiod Mean-Variance
  Formulation}.
\newblock {\em Mathematical Finance}, 10(3):387--406, 2000.

\bibitem{LiForsyth2019}
Y.~Li and P.A. Forsyth.
\newblock A data-driven neural network approach to optimal asset allocation for
  target based defined contribution pension plans.
\newblock {\em Insurance: Mathematics and Economics}, (86):189--204, 2019.

\bibitem{LiangBaiGuo2014}
X.~Liang, L.~Bai, and J.~Guo.
\newblock Optimal time-consistent portfolio and contribution selection for
  defined benefit pension schemes under mean-variance criterion.
\newblock {\em ANZIAM}, (56):66--90, 2014.

\bibitem{LinQian2016}
X.~Lin and Y.~Qian.
\newblock Time-consistent mean-variance reinsurance-investment strategy for
  insurers under cev model.
\newblock {\em Scandinavian Actuarial Journal}, (7):646--671, 2016.

\bibitem{LuDang2023}
Y.~Lu and D.M. Dang.
\newblock A pointwise convergent numerical integration method for {Guaranteed
  Lifelong Withdrawal Benefits} under stochastic volatility.

\bibitem{LuDang2022}
Y.~Lu and D.M. Dang.
\newblock A semi-{L}agrangian {$\epsilon$}-monotone {F}ourier method for
  continuous withdrawal {GMWBs} under jump-diffusion with stochastic interest
  rate.

\bibitem{online23}
Y.~Lu, D.M. Dang, P.A. Forsyth, and G.~Labahn.
\newblock An {$\epsilon$}-monotone {F}ourier method for {Guaranteed Minimum
  Withdrawal Benefit (GMWB)} as a continuous impulse control problem.

\bibitem{Pavel2015}
X.~Luo and P.V. Shevchenko.
\newblock {Valuation of variable annuities with guaranteed minimum withdrawal
  and death benefits via stochastic control optimization}.
\newblock {\em Insurance: Mathematics and Economics}, 62(3):5--15, 2015.

\bibitem{MaForsyth2016}
K.~Ma and P.~A. Forsyth.
\newblock Numerical solution of the {Hamilton-Jacobi-Bellman} formulation for
  continuous time mean variance asset allocation under stochastic volatility.
\newblock {\em Journal of Computational Finance}, 20(01):1--37, 2016.

\bibitem{Markowitz1952}
H.~Markowitz.
\newblock Portfolio selection.
\newblock {\em The Journal of Finance}, 7(1):77--91, March 1952.

\bibitem{MenoncinVigna2013}
F.~Menoncin and E.~Vigna.
\newblock Mean-variance target-based optimisation in {DC} plan with stochastic
  interest rate.
\newblock {\em Working paper, Collegio Carlo Alberto}, (337), 2013.

\bibitem{MertonJumps1976}
R.C. Merton.
\newblock Option pricing when underlying stock returns are discontinuous.
\newblock {\em Journal of Financial Economics}, 3:125--144, 1976.

\bibitem{MichaudBOOK}
R.O. Michaud and R.O. Michaud.
\newblock {\em Efficient asset management: A practical guide to stock portfolio
  optimization and asset allocation}.
\newblock Oxford University Press, 2 edition, 2008.

\bibitem{ni2022optimal}
Chendi Ni, Yuying Li, Peter Forsyth, and Ray Carroll.
\newblock Optimal asset allocation for outperforming a stochastic benchmark
  target.
\newblock {\em Quantitative Finance}, 22(9):1595--1626, 2022.

\bibitem{Nkeki2014}
C.I. Nkeki.
\newblock Stochastic funding of a defined contribution pension plan with
  proportional administrative costs and taxation under mean-variance
  optimization approach.
\newblock {\em Statistics, optimization and information computing},
  (2):323--338, 2014.

\bibitem{Oberman2006}
A.M. Oberman.
\newblock {Convergent difference schemes for degenerate elliptic and parabolic
  equations: Hamilton--Jacobi Equations and free boundary problems}.
\newblock {\em SIAM Journal Numerical Analysis}, 44(2):879--895, 2006.

\bibitem{PerrinRoncalli2019}
S.~Perrin and T.~Roncalli.
\newblock {\em Machine Learning Optimization Algorithms and Portfolio
  Allocation}, chapter~8, pages 261--328.
\newblock Wiley Online Library, 2020.

\bibitem{Pham}
H.~Pham.
\newblock On some recent aspects of stochastic control and their applications.
\newblock {\em Probability Surveys}, 2:506--549, 2005.

\bibitem{pooley2003}
D.M. Pooley, P.A. Forsyth, and K.R. Vetzal.
\newblock Numerical convergence properties of option pricing {PDE}s with
  uncertain volatility.
\newblock {\em IMA Journal of Numerical Analysis}, 23:241--267, 2003.

\bibitem{puter94}
M.~Puterman.
\newblock {\em {Markov Decison Processes: Discrete Stochastic Dynamic
  Programming}}.
\newblock {Wiley, New York}, 1994.

\bibitem{RamezaniZeng2007}
C.~Ramezani and Y.~Zeng.
\newblock Maximum likelihood estimation of the double exponential
  jump-diffusion process.
\newblock {\em Annals of Finance}, 3(4):487--507, 2007.

\bibitem{RF2016}
C.~Reisinger and P.A. Forsyth.
\newblock {Piecewise constant Policy approximations to Hamilton-Jacobi-Bellman
  equations}.
\newblock {\em Applied Numerical Mathematics}, 103:27--47, 2016.

\bibitem{Ruijter2013}
M.J. Ruijter, C.W. Oosterlee, and R.F.T. Aalbers.
\newblock On the {F}ourier cosine series expansion ({COS}) method for
  stochastic control problems.
\newblock {\em Numerical Linear Algebra with Applications}, 20:598--625, 2013.

\bibitem{Sato2019}
Y.~Sato.
\newblock Model-free reinforcement learning for financial portfolios: A brief
  survey.
\newblock {\em Working paper}, 2019.

\bibitem{Pavel2016}
P.V. Shevchenko and X.~Luo.
\newblock {A unified pricing of variable annuity guarantees under the optimal
  stochastic control framework}.
\newblock {\em Risks}, 4(3):1--31, 2016.

\bibitem{StrubLiCui2019}
M.~Strub, D.~Li, and X.~Cui.
\newblock An enhanced mean-variance framework for robo-advising applications.
\newblock SSRN 3302111, 2019.

\bibitem{SunLiZeng2016}
J.~Sun, Z.~Li, and Y.~Zeng.
\newblock Precommitment and equilibrium investment strategies for defined
  contribution pension plans under a jump--diffusion model.
\newblock {\em Insurance: Mathematics and Economics}, (67):158--172, 2016.

\bibitem{PvSDangForsyth2018_TCMV}
P.~M. {Van Staden}, D.M. Dang, and P.A. Forsyth.
\newblock Time-consistent mean-variance portfolio optimization: a numerical
  impulse control approach.
\newblock {\em Insurance: Mathematics and Economics}, 83(C):9--28, 2018.

\bibitem{PvSDangForsyth2018_MQV}
P.~M. {Van Staden}, D.M. Dang, and P.A. Forsyth.
\newblock Mean-quadratic variation portfolio optimization: A desirable
  alternative to time-consistent mean-variance optimization?
\newblock {\em SIAM Journal on Financial Mathematics}, 10(3):815--856, 2019.

\bibitem{PvSDangForsyth2019_Distributions}
P.~M. {Van Staden}, D.M. Dang, and P.A. Forsyth.
\newblock On the distribution of terminal wealth under dynamic mean-variance
  optimal investment strategies.
\newblock {\em SIAM Journal on Financial Mathematics}, 12(2):566--603, 2021.

\bibitem{PvSDangForsyth2019_Robust}
P.~M. {Van Staden}, D.M. Dang, and P.A. Forsyth.
\newblock The surprising robustness of dynamic mean-variance portfolio
  optimization to model misspecification errors.
\newblock {\em European Journal of Operational Research}, 289:774--792, 2021.

\bibitem{van2021practical}
Pieter~M Van~Staden, Duy-Minh Dang, and Peter~A Forsyth.
\newblock Practical investment consequences of the scalarization parameter
  formulation in dynamic mean--variance portfolio optimization.
\newblock {\em International Journal of Theoretical and Applied Finance},
  24(05):2150029, 2021.

\bibitem{Vigna_efficiency2014}
E.~Vigna.
\newblock On efficiency of mean-variance based portfolio selection in defined
  contribution pension schemes.
\newblock {\em Quantitative Finance}, 14(2):237--258, 2014.

\bibitem{Vigna2016TC}
E.~Vigna.
\newblock On time consistency for mean-variance portfolio selection.
\newblock {\em International Journal of Theoretical and Applied Finance},
  23(6), 2020.

\bibitem{vigna2022tail}
Elena Vigna.
\newblock Tail optimality and preferences consistency for intertemporal
  optimization problems.
\newblock {\em SIAM Journal on Financial Mathematics}, 13(1):295--320, 2022.

\bibitem{wang08}
J.~Wang and P.A. Forsyth.
\newblock Maximal use of central differencing for {Hamilton-Jacobi-Bellman
  PDEs} in finance.
\newblock {\em SIAM Journal on Numerical Analysis}, 46:1580--1601, 2008.

\bibitem{WangForsyth2011}
J.~Wang and P.A. Forsyth.
\newblock Continuous time mean variance asset allocation: A time-consistent
  strategy.
\newblock {\em European Journal of Operational Research}, 209(2):184--201,
  2011.

\bibitem{WangChen2018}
L.~Wang and Z.~Chen.
\newblock Nash equilibrium strategy for a {DC} pension plan with
  state-dependent risk aversion: A multiperiod mean-variance framework.
\newblock {\em Discrete Dynamics in Nature and Society}, (1-17), 2018.

\bibitem{WangChen2019}
L.~Wang and Z.~Chen.
\newblock Stochastic game theoretic formulation for a multi-period {DC} pension
  plan with state-dependent risk aversion.
\newblock {\em Mathematics}, 7(108):1--16, 2019.

\bibitem{warin2016}
Xavier Warin.
\newblock Some non-monotone schemes for time dependent
  {H}amilton-{Ja}cobi-{Be}llman equations in stochastic control.
\newblock {\em Journal of Scientific Computing}, 66(3):1122--1147, 2016.

\bibitem{WeiWang2017}
J.~Wei and T.~Wang.
\newblock Time-consistent mean-variance asset-liability management with random
  coefficients.
\newblock {\em Insurance: Mathematics and Economics}, (77):84--96, 2017.

\bibitem{WuZeng2015}
H.~Wu and Y.~Zeng.
\newblock Equilibrium investment strategy for defined-contribution pension
  schemes with generalized mean-variance criterion and mortality risk.
\newblock {\em Insurance: Mathematics and Economics}, 64:396--408, 2015.

\bibitem{yu1974cone}
P~Le Yu.
\newblock Cone convexity, cone extreme points, and nondominated solutions in
  decision problems with multiobjectives.
\newblock {\em Journal of Optimization Theory and Applications},
  14(3):319--377, 1974.

\bibitem{ZengLi2011}
Y.~Zeng and Z.~Li.
\newblock Optimal time-consistent investment and reinsurance policies for
  mean-variance insurers.
\newblock {\em Insurance: Mathematics and Economics}, 49(1):145--154, July
  2011.

\bibitem{ZhaoShenZeng2016}
H.~Zhao, Y.~Shen, and Y.~Zeng.
\newblock Time-consistent investment-reinsurance strategy for mean-variance
  insurers with a defaultable security.
\newblock {\em Journal of Mathematical Analysis and Applications},
  437(2):1036--1057, May 2016.

\bibitem{ZhouLi2000}
X.~Zhou and D.~Li.
\newblock Continuous time mean variance portfolio selection: a stochastic {LQ}
  framework.
\newblock {\em Applied Mathematics and Optimization}, 42:19--33, 2000.

\bibitem{zhou2000continuous}
X.Y. Zhou and D.~Li.
\newblock {Continuous-time mean-variance portfolio selection: A stochastic LQ
  framework}.
\newblock {\em Applied Mathematics and Optimization}, 42(1):19--33, 2000.

\bibitem{ZhouXiaoYinZengLin2016}
Z.~Zhou, H.~Xiao, J.~Yin, X.~Zeng, and L.~Lin.
\newblock Pre-commitment vs. time-consistent strategies for the generalized
  multi-period portfolio optimization with stochastic cash flows.
\newblock {\em Insurance: Mathematics and Economics}, (68):187--202, 2016.

\end{thebibliography}
\end{document}